\documentclass[a4paper,UKenglish,cleveref, autoref, thm-restate]{lipics-v2021}

\usepackage{amsmath}
\usepackage{amssymb}
\usepackage{graphicx}
\usepackage{color}
\usepackage{hhline}
\usepackage{tabu}
\usepackage{multirow}
\usepackage{multicol}

\newcommand{\ignore}[1]{}

\usepackage{tikz}
\usetikzlibrary{er,positioning,bayesnet}
\usetikzlibrary{calc,shapes.multipart,chains,arrows}
\usetikzlibrary{arrows,shapes,automata}
\tikzstyle{attribut} = [ellipse,draw,minimum size=22pt,inner sep=1pt]
\tikzstyle{edge} = [draw,thick,->]
\tikzstyle{weight} = [font=\small]
\usepackage[normalem]{ ulem }

\usepackage{caption}

\usepackage{tabu}

\newcommand{\boxtheorem}{\hfill $\blacksquare$\\}
\newcommand{\nit}[1]{{\it #1}}
\DeclareMathOperator*{\argmax}{\mathsf{argmax}}
\DeclareMathOperator*{\argmin}{\mathsf{argmin}}

\newcommand{\h}{\hspace*{3mm}}

\newcounter{Bexample-counter}
\newcounter{Bremark-counter}
{\vskip \abovedisplayskip \refstepcounter{Bexample-counter}%
\noindent {\bf Example \arabic{Bexample-counter}.}  }%

{\vskip \abovedisplayskip \refstepcounter{Bremark-counter}%
\noindent {\bf Remark \arabic{Bremark-counter}.}}%

\newcommand{\mc}[1]{\mathcal{ #1}}

\newcommand{\mbb}[1]{\mathbb{ #1}}
\newcommand{\mbf}[1]{\mathbf{ #1}}

\newcommand{\C}[1]{\mathcal{C}}
\newcommand{\T}[1]{\mathcal{T}}

\newcommand{\na}{{\bf \sf na}}

\newcommand{\Nn}{\mbox{\scriptsize \sf NULL}}

\newcommand{\red}[1]{\textcolor{red}{#1}}

\newcommand{\blue}[1]{\textcolor{blue}{#1}}

\newcommand{\perm}{\mathsf{perm}}

\newcommand{\comlb}[1]{{\vspace{2mm}\noindent \bf  {\red{COMM(LEO):}}}~ #1 \hfill {\bf
    END.}\vspace*{2mm} \\ }

\newcommand\independent{\protect\mathpalette{\protect\independenT}{\perp}}
\def\independenT#1#2{\mathrel{\rlap{$#1#2$}\mkern2mu{#1#2}}}

\hideLIPIcs  

\title{Databases with Missing Values that are Governed by Missingness Mechanisms} 

 \author{Leopoldo {Bertossi}}{Carleton University, Ottawa, Canada \and IMFD, Chile,   }{bertossi@scs.carleton.ca}{https://orcid.org/0000-0002-1825-0097}{}

  \author{Farouk Toumani}{LIMOS, CNRS, Clermont Auvergne University}{farouk.toumani@uca.fr}
 {[orcid]}{}

\authorrunning{L.Bertossi et al.} 


\EventEditors{John Q. Open and Joan R. Access}
\EventNoEds{2}
\EventLongTitle{42nd Conference on Very Important Topics (CVIT 2016)}
\EventShortTitle{CVIT 2016}
\EventAcronym{CVIT}
\EventYear{2016}
\EventDate{December 24--27, 2016}
\EventLocation{Little Whinging, United Kingdom}
\EventLogo{}
\SeriesVolume{42}
\ArticleNo{23}

\begin{document}

\maketitle

\begin{abstract}
We address the problems of giving a semantics to a relational database (RDB) that has missing values (MVs). The causes for the latter are governed by a
Missingness Mechanism that is modelled as a Bayesian Network (BN) that involves the DB attributes as variables. The BN is called a Missingness Graph (MG). Our approach considerable departs from the treatment of RDBs with NULL (values). The combination of the MG and the observed DB allows us to build a block-independent probabilistic DB. We identify two optimal classes of its possible worlds on which  QA can be performed. Those classes jointly capture probabilistic uncertainty and statistical plausibility of the implicit imputation of MVs. We obtain tractability results for the computation of some optimal classes; and we also obtain complexity results that characterize the computational feasibility of our approach.
\end{abstract}

\vspace{-2mm}\section{Introduction}\label{sec:intro}

\vspace{-2mm}
\h It is common to find missing values (MVs) in a  database (DB) $D^\star$, that is, values for some attributes that are not reported. We  use \na, for ``not available", to indicate that the true value is  absent. No other semantics is assigned to  \na.
\ We consider the {\em observed}  DB $D^\star$ as  obtained from an independent  sample from an outside reality.\footnote{This kind of observed DBs, including independence and MVs,  are common with census-related data \cite{gelman}.} We assume that  $D^\star$  is an incomplete representation of the ``true", possibly unknown and {\em underlying} DB $D$ that represents the external reality, and whose attributes have non-\na \ values.

\captionof{figure}{\ \ (a) Relation with Missing Values \hspace{0.5cm} (b) Missingness Graph}\label{table:first}
\vspace{-3mm}\begin{center}
\ignore{ {\tiny \vspace{-2.5cm}$\begin{tabu}{c|c|c|c|}\hline
R^\star & A^\star & B^o & C^\star \\ \hline
\tau_1 & a_1 & 0& c_1\\
\tau_2 & a_2& 1& \na\\
\tau_3 &\na & 0&c_3\\
\tau_4 & a_4&0&c_4\\
\tau_5 & \na&1&c_5\\
\tau_6 &\na&1& \na\\ \hhline{~---}
\end{tabu}$}}
\includegraphics[width=2.5cm]{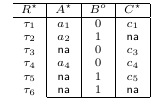}
~~~~~~~~\includegraphics[width=2cm]{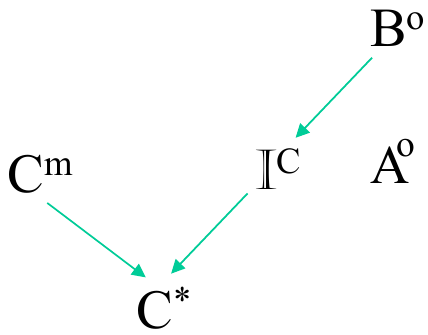}
\end{center}

\vspace{-4mm}
\begin{example}\label{ex:first}  Relation $R^\star$  in Figure \ref{table:first}(a) belongs to  an observed DB $D^\star$. We use global tuple identifiers (tids): $\tau_1$, etc. Some (true) values are missing. \
Using notation  from \cite{mohan21}, \ attribute \  $B^o$ \ is   fully observable, i.e. it never shows  \na, whereas  attributes  $A^\star, C^\star$ may have MVs.
\ $R^\star$ has the {\em observed schema} $R^\star(A^\star,B^o,C^\star)$, whereas the underlying  DB $D$ has the {\em underlying schema} $R(A^m,B^o,C^m)$, where the superscript ``$m$''  indicates that the attribute may have MVs. We do not observe $A^m$ directly, but via $A^\star$. When the latter does not show an \na, it shows the real value of $A^m$.
\boxtheorem
\end{example}

\vspace{-4mm}\h The main questions that underlie our work are the following: \ (a) {\em What is the semantics of the observed DB?}; and (b) {\em How can we query the ``underlying" DB $D$ in a meaningful way ``through" the observed DB?} \  This quest looks hopeless unless we have some additional information. Actually, we will assume we know about the {\em causes for the occurrence of MVs} by means of a Bayesian Network (BN) \cite{pearl} that stochastically models, as a {\em Missingness Graph} (MG) \cite{mohan21}, the  {\em Missingness Mechanisms} (MMs) at play \cite{rubin76,rubinBook}. Those MMs describe under what stochastic conditions MVs may appear.

\vspace{-1mm}\begin{example}\label{ex:second} (ex. \ref{ex:first} cont.) \  For simplicity, assume now that attribute $A$ is fully observed and denoted with $A^o$. The MG in Figure \ref{table:first}(b) shows $A^o$ as
independent from the other attributes. \ $\mbb{I}^C$ is an indicator variable for $C^m$ that (deterministically or with high probability)
takes the value $1$ when the value for $C^m$ is missing, and $0$,
otherwise. \
The MG tells us that whether $C^\star$ takes the value \na \ or not, depends on variable $C^m$ and $B^o$, for the latter, via $\mbb{I}^C$. \boxtheorem
\end{example}

\ignore{\begin{center}
    \includegraphics[width=2cm]{mg3.pdf}
\vspace{-2mm}\captionof{figure}{Missingness Graph.}\label{fig:mg3Intro}
\end{center}}

\vspace{-5mm}\h In general, we start with an observed DB, $D^\star$, and a MG, $\mc{M}$, that models the occurrence of MVs. Observed values are assumed to be certain. We also assume $D^\star$ is compatible with the joint distribution induced by the MG.

\h Using the
combination of $D^\star$ and $\mc{M}$, we create  (virtually or physically) a {\em Block-Independent Probabilistic DB} (BID) \cite{suciu}, which
gives rise to  a collection $\mc{W}$ of {\em possible worlds}, each of them an \na-free DB instance for the underlying schema, with an associated global probability.   Each possible world becomes a possible materialization -without MVs- of the partially observed underlying DB $D$.

\h Except for tuple identifiers (tids), that can be of an auxiliary nature, or intrinsic to data in some applications,  possible worlds may end up having, due to imputations,  duplicate tuples. Accordingly, we adopt  a {\em bag semantics}. This more general setting is quite natural in applications like census-like data. In other cases, one can always fall back to a set-semantics.\footnote{We have developed a set-semantics, on the basis of the bag semantics, but we do not report on it here.}

\h Our approach can be seen as a generalization of traditional {\em imputation techniques} \cite{gelman}, where, instead of a single ``clean" instance, a class of probability-weighted clean instances is considered. This makes for a more principled and uncertainty-aware data semantics and query answering (QA). Our way of dealing with MVs can be understood as a form of collective imputation  by means of  probability distributions over possible values; distributions that are related to each other through an underlying joint distribution determined by a MG. Those ``cell-level" distributions can be seen as footprints of that joint distribution.

\h Query answering (QA) can be done on the BID, which has a clear semantics \cite{suciu}. Without neglecting this ``direct" approach,   we propose a more efficient alternative semantics that builds on the BID, but takes into account the {\em statistical compliance} of possible worlds in relation to the distribution induced by the MG.   More precisely, we start by collecting in different {\em classes} those possible worlds that are essentially the same, the  {\em matching worlds}. They contain the same tuples as multisets. Accordingly, QA on matching worlds returns the same answer.  After that, we introduce a {\em measure of compliance} for classes of matching worlds. We  identify those classes that are maximally compliant with the MG, and do QA on top of them. \  More specifically, in this work we make the following contributions:

\vspace{1mm}
1. We bring into- and develop in the realm of relational DBs some concepts introduced by Mohan and Pearl \cite{mohan21} on the use of Bayesian Networks for the specification of missingness mechanisms. This leads us to introduce and exploit auxiliary probabilistic DBs.

\vspace{1mm}
2. We define the semantics of a DB with  MVs in relation to a MG. \ This is done through  a BID associated to the observed DB and the MG, giving rise to classes of matching possible worlds. Classes are used to define a bag-semantics for data and QA.

\vspace{1mm}
    3. We propose two particular  data and QA semantics:  The \emph{Most-Probable Classes} (MPC), and the \emph{Most-Compliant Classes} (MCC). \  Compliance is measured as a statistical  distance between a world's empirical distribution and that induced by the MG.

 \vspace{1mm}   4. We investigate the computational complexity of the  MCC- and MPC-semantics.  For a broad and common family of notions of compliance, we show one can compute a most-compliant class in polynomial-time in the size of the observed DB, on which QA can be done. Furthermore, although there may be exponentially many classes,  enumerating all the most-compliant classes and query answers on them  can be done with polynomial-time delay.

 \vspace{1mm}   5. We obtain several hardness results for the MCC-semantics and the MPC-semantics. Still, the former has better computational properties than the MPC-semantics.

\vspace{1mm}    6. We show that {\em recovering probabilities} from the the available, observed data is possible for certain classes of qualitative MGs.

\vspace{1mm}
\h In this work, we do not address the important related problem of learning the MGs, nor that of checking the compliance of the initially observed instance against  MGs. These are matter of separate ongoing research. Instead, we assume the MG is given and the observed data complies with it. \ The MG may come from learning from data about the same domain,  and from domain knowledge \cite{heckerman95}.

\h Learning the BN considers MVs as any other categorical values, and then, it is like general learning of BNs that do not involve MVs in the data domains \cite{heckerman}. However, there are some new issues: (a) The introduction of the auxiliary indicator functions; and most importantly (b) the possible lack of values for the unobserved  variables of the form $A^m$: We observe $\na$'s via $A^\star$, but in those cases not the real values of $A^m$ hidden underneath.
In this direction, {\em probability recoverability} techniques like those in  \cite{broeck2015} can be useful. Having said all this, we still retake some of these issues in Section \ref{sec:recover}.

\h
This paper is structured as follows. Section~\ref{sec:prel} provides background. Section~\ref{sec:mgs} introduces MMs and MGs. Section \ref{sec:genBIDs} introduces the BIDs associated to an observed DB. Section~\ref{sec:queries} introduces classes of matching worlds and class-based QA. Section \ref{sec:compliance-pw} introduces the notion of compliance.  Section \ref{sec:complexity} investigates algorithmic and complexity aspects of most-compliant classes and QA. Section \ref{sec:mpdbs} investigates algorithmic and complexity aspects of most-probable classes and QA.   Section \ref{sec:recover} revisits probability recoverability. Section \ref{sec:relWork} discusses related work.     Section~\ref{sec:conclusion}  summarizes our contributions and outlines directions of future work. The Appendices provide additional material and proofs of results.

\vspace{-2mm}
\section{Background and Preliminaries}
\label{sec:prel}

\vspace{-1mm}
{\bf Relational Databases.} \
A relational schema, $S$,  is a finite set of logical predicates, $R, \ldots$,  with fixed arities. Variables, a.k.a. attributes or features, are associated to predicate positions. A relational DB, $D$,  is a collection of relations with finite extensions for the predicates. Their elements, the tuples, have values from fixed and finite domains. $\nit{dom}(X)$ denotes the domain of  attribute $X$. For several attributes, $\bar{X}$, $\nit{dom}(\bar{X})$ is the cartesian product of the individual domains; and $\nit{dom}$ is the union of the domains. \ The string \na, indicating a missing value, \ does not belong to $\nit{dom}$. $\nit{dom}^{\!\star}$ denotes $\nit{dom} \cup \{\na\}$.
 \ Unless otherwise stated, DB relations may have duplicates, i.e. repeated tuples, which we tell apart by means of global {\em tuple identifiers} (tids) that appear in a first attribute of tuples, acting as a surrogate key, as in Example \ref{ex:first}. We denote tids with $\tau, \tau_1, \ldots$. We frequently refer to tuples by their tids.

 \h Relational  queries may have constants from $\nit{dom}$, and then, different from \na.
The set of answers to a query $\mc{Q}$ from a DB $D$ is denoted with $\mc{Q}[D]$.
Boolean queries (BQ) have answers $0$ or $1$. Queries do not mention the tids, nor their auxiliary attributes.

\vspace{1mm}
\noindent {\bf Probabilistic Databases.} \
  A PDB, $D^p$, associated to a relational DB, $D$, is the {\em collection $\mc{W}$ of possible worlds} that are  subinstances $W$ of  $D$. Each $W \in \mc{W}$ has a  probability $P^{\mc{W}}(W)$, such that  $\sum_{W\in\mc{W}}P^{\mc{W}}(W) = 1$ \cite{suciu}. $D^p$  becomes a discrete {\em probability space} $\langle \mc{W}, P^\mc{W}\rangle$. \ The probability of tuple $\tau$ being (true) in $D$ -as seen through $D^p$- is  \
$P(\tau) \ := \sum_{\tau \in W \in \mc{W}} P^{\mc{W}}(W)$.
\ A numeric query $\mc{Q}$ on $D$ becomes a random variable on  $\mc{W}$; with Bernoulli distribution if it is a BQ,  in which case, the probability of $\mc{Q}$ (being true) is \
$P(\mc{Q}) \ := \ P(\mc{Q} = 1) :=
\sum_{W \in \mc{W}: \ W \models \mc{Q}} P^{\mc{W}}(W)$.   We are interested in {\em block-independent} PDBs (BIDs) \cite{suciu}.  Tuples have probabilities, and tuples of a same relation are separated in mutually independent blocks of mutually exclusive tuples.

\vspace{-1mm}\begin{center}
\captionof{table}{A BID}\label{tab:pdbs}
\vspace{-7mm}{\tiny
$\begin{tabu}{c|c|c|c||c|}\hline
R & A & B & C& P \\ \hline
\tau^1_1& a_1 & b_1& c_1&p_1^1\\
\tau^2_1 & a_2& b_2& c_2&p_1^2\\
\hline
\tau^1_2 & a_3& b_3& c_3&p_2^1\\
\tau^2_2& a_4& b_4& c_4&p_2^2\\
\tau^3_2& a_5& b_5& c_5&p^3_2\\
\hline
\end{tabu}$}
\end{center}

\h  Table \ref{tab:pdbs} shows a BID $R$ with two blocks. \ In a possible world $W \in \mc{W}$, the corresponding  relation $R_W$ is built from $R$,
 by choosing, from each block $B_j$, one tuple $\tau_j^i$. The probability of $R_W$ is: \ $p(R_W) := \Pi p_j^i$. Any other  of subrelation of $R$ has probability $0$. Within a block, the sum, $p$, of the tuples' probabilities is not greater than $1$. If none of the tuples is chosen from the block, the block contributes with the factor $(1-p)$  to the relation's probability.
  The probability of a world $W$ is $P^{\mc{W}}(W) : = \Pi_{_{R \in S}} p(R_W)$. A  \emph{most probable database} (MPDB) is a possible world in $\mc{W}$  with maximum probability.  A {\em tuple-independent PDB} (TID) is a particular kind of BID: Each block has a single tuple with a probability not greater  than $1$.

\vspace{-2mm}
\section{Missingness and Observed DBs}\label{sec:mgs}

 \vspace{-2mm}\h A Missingness Mechanism (MM)
 specifies if and how the occurrence of MVs in a variable depends on the values that other (or the same) variables take \cite{rubin76}.  As in  \cite{mohan21}, we represent MMs with {\em Missingness Graphs} (MGs), i.e. BNs whose  directed edges capture stochastic-dependencies, and how the occurrence of MVs for a variable depends on other variables. Accordingly, MGs become  acyclic graphs subject to the {\em Markov Condition}: Given its parents, a variable is independent from its non-descendant variables \cite{darwicheCommACM,pearl}.

\begin{example}  \label{ex:mgs}
There are MMs commonly  found in practice \cite{rubinBook}. The MG in Figure \ref{fig:mg}(a) corresponds to the MCAR case (missing completely at random) in that variable $\mathbb{I}^C$ does not depend on any other variable, and has absolute probabilities of taking values $0$ or $1$.

\vspace{-2mm}
 \begin{center}   \includegraphics[width=6cm]{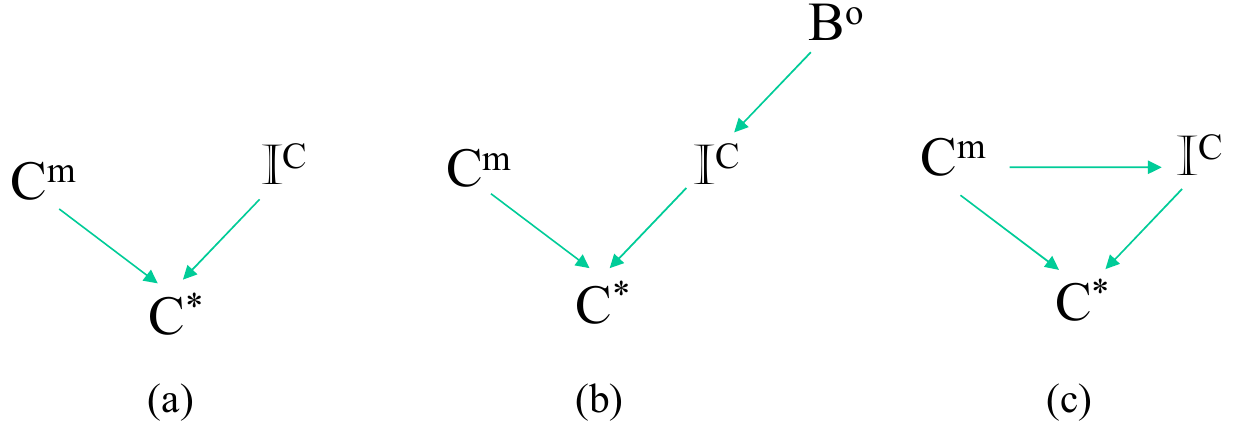}
\end{center}
\vspace{-6mm}\captionof{figure}{Three Missingness Graphs.}\label{fig:mg}

\vspace{2mm}\h The MG in Figure \ref{fig:mg}(b) shows a case of MAR (missing at random):   $C^\star$ stochastically depends on $C^m$ and $\mathbb{I}^C$; while the latter depends on $B^o$.
  As expected, $C^\star$ depends directly on its real, underlying version, $C^m$, and its own  indicator variable.  The BN  has the distributions: \ $P(C^m)$, $P(B^o)$, $P(\mathbb{I}^C|B^o)$, $P(C^\star|C^m,\mathbb{I}^C)$. Due to the  Markov condition, absolute and conditional independencies hold:
$B^o \independent C^m, \ \ \mathbb{I}^C \independent C^m, \ \ (C^\star \independent B^o)|\mathbb{I}^C$.

\h The MG in Figure \ref{fig:mg}(c) shows a case of MNCAR (missing not completely at random); actually of {\em self-censorship}:  $\mathbb{I}^C$ depends on $C^m$. For example, $C^m$ could be $\nit{Salary}^m$, that -with high probability- is not reported when it is very high. When $\mbb{I}^\nit{Salary}$ takes value $1$, $\nit{Salary}^\star$ takes value $\na$. Otherwise, $\nit{Salary}^\star$ takes the real value of $\nit{Salary}^m$.
\boxtheorem \end{example}

\vspace{-3mm}\h In a MG, each variable $X^\star$ is a sink, and has exactly two parents:  $X^m$ and $\mathbb{I}^X$.  $\mathbb{I}^X$ may depend on variables of the forms $Y^o$ and $\mathbb{I}^Z$, for other variables of the form  $Z^m$.  Variables of the form $B^o$ are sources.
\ A MG $\mc{M}$, as a BN, has an underlying {\em joint distribution} $P^\mc{M}$ over all its variables.    In particular, a tuple $\tau$ for the underlying schema $S$ has a marginal probability $P^\mc{M}(\tau)$. We assume variables in a MG are attributes of a relation schema.

\begin{example} \label{ex:new} \ (ex. \ref{ex:second} cont.) \ Consider the observed DB $R^\star$ in Table \ref{tab:dbex0}(a) for the observed schema $S^\star$.  MVs in it are governed by the MG $\mc{M}$ in Figure \ref{table:first}(b). \
 $\mc{M}$ induces the factorized joint distribution:

\vspace{-2mm}
{\footnotesize \begin{equation}
P^\mc{M}(A^o,B^o,C^\star,\mbb{I}^C,C^m)=P^\mc{M}(C^\star~|~C^m,\mbb{I}^C) \times P^\mc{M}(I^C~|~B^o)  \times P^\mc{M}(C^m)    \times P^\mc{M}(B^o) \times P^\mc{M}(A^o). \label{eq:joint}
\end{equation}}

\captionof{table}{\hspace{0.7cm}(a) Observed  \ \hspace{1cm}(b) Underlying  \hspace{1cm}(c) Expanded DB}\label{tab:dbex0} \vspace{-6mm}
\begin{center}
 \scalebox{0.6}{$\begin{tabu}{c|c|c|c|}\hline
 R^\star & A^o & B^o & C^\star\\ \hline
\tau_1& a_1 & 0& c_1\\
\tau_2& a_2& 1& \na\\
\tau_3&a_3 & 0&c_3\\
\tau_4& a_4&0&c_4\\
\tau_5& a_5&1&c_5\\
\tau_6&a_6&1& \na\\ \hhline{~---}
\end{tabu}$}~~~~~~~~~
\scalebox{0.6}{$\begin{tabu}{c|c|c|c|}\hline
 W & A^o & B^o & C^m\\ \hline
 \tau_1& a_1 & 0& c_1\\
 \tau_2& a_2& 1& \underline{c_5}\\
\tau_3&a_3 & 0&c_3\\
\tau_4& a_4&0&c_4\\
\tau_5& a_5&1&c_5\\
\tau_6&a_6&1& \underline{c_2}\\ \hhline{~---}
 \end{tabu}$}~~~~~~~~
\scalebox{0.6}{$\begin{tabu}{c|c|c|c|c|c|c|c|}\hline
R^\nit{ex} & A^o  & B^o & C^\star & \mathbb{I}^C&C^m\\ \hline
\tau_1& a_1 &  0& c_1 & 0&c_1\\
\tau_2& a_2&  1& \na & 1&c_5\\
\tau_3&a_3  &  0&c_3&0&c_3\\
\tau_4& a_4 & 0&c_4 &0&c_4\\
\tau_5& a_5 & 1&c_5 & 0&c_5\\
\tau_6&a_6& 1& \na &1&c_2\\ \hhline{~------}
\end{tabu}$}
 \end{center}

\vspace{-1mm}\h In Table \ref{tab:dbex0}(b), $W$ is an \na-free  instance for the underlying schema $S$,   obtained by imputation of MVs: domain values $c_5$ and $c_2$ (underlined) for $C^m$ replace the MVs. $W$ may coincide or not with the ``true" underlying DB $D$ that we only partially observe.   The ``expanded" relation $R^\nit{ex}$ in Table \ref{tab:dbex0}(c), for an {\em expanded schema} $S^\nit{ex}$, shows a possible sample from the outside reality represented by $\mc{M}$, with values for all variables  in it. $R^\star$ and $R$ can be seen as {\em footprints} of $R^\nit{ex}$.     A ``user" has access only to the observed instance $R^\star$, but queries will be posed in the language  of the underlying schema $S$ (see Sec. \ref{sec:queries}).
 \boxtheorem
\end{example}

\vspace{-4mm}\h {\em We will assume that tuples in observed  $D^\star$  are stochastically independent, as obtained from an independent sample of the real world}.\footnote{If we want to impose tuple correlations, we can use {\em soft constraints} via {\em MarkoViews} \cite{marko,forall,NOW}.} \
$D^\star$ may have duplicate tuples modulo the tids. For auxiliary or intrinsic tids, we assume that they are never missing; and there are no variables for them in a MG.
\ Due to imputations, a possible underlying DB, as $W$ in Example \ref{ex:second},  may also introduce  duplicates. Those duplicates are kept.

\h We  assume that an observed DB $D^\star$ is {\em compliant} with the given MG $\mc{M}$; i.e.  with the (in)dependencies and  induced distributions. We may also assume that each relation has its own MG, which involves variables that appear in the relation's expanded schema. MGs are meant to be used to specify missingness, as opposed to arbitrary and general dependencies.

\h  Finally,  the following basic {\bf Conditions on MGs} hold. \ For arbitrary, but different $c, c^\prime \in \nit{dom}(X^m)$:

\vspace{1mm}
\noindent {\bf (C)} \
  $P^\mc{M}(X^\star=\na|\mbb{I}^X=1,X^m=c, \ldots) = 1, \ P^\mc{M}(X^\star=\na,\ldots|\mbb{I}^X=0,X^m=c,\ldots) = 0,\\
\hspace*{8mm}P^\mc{M}(X^\star=c^\prime,\ldots|\mbb{I}^X=0,X^m=c,\ldots) = 0, \
 P^\mc{M}(X^\star=c|\mbb{I}^X=0,X^m=c, \ldots) = 1$.

 \ignore{
\vspace{1mm}\noindent {\bf (II)} \  $P^\mc{M}(X^\star=\na, \ldots | \mbb{I}^X=1,X^m=c,\ldots) = P^\mc{M}(\ldots~|~\mbb{I}^X=1,X^m=c, \ldots), \
 P^\mc{M}(X^\star=c, \ldots | \mbb{I}^X=0, X^m=c, \ldots)  =
 P^\mc{M}(\ldots~|~\mbb{I}^X=0,X^m=c, \ldots)$.
}

\vspace{1mm}
\h Conditions {\bf (C)}  can be replaced by deterministic {\em structural equations} at $C^\star$ nodes \cite{pearl}: \ $C^\star = \na \mbox{ iff } \mathbb{I}^C=1$, and $C^\star = C^m \mbox{ iff } \mathbb{I}^C=0$. MGs with this property will be called $\mbb{I}$-deterministic. Given the auxiliary role of indicator functions when modeling  the occurrence of MVs, they are in general $\mbb{I}$-deterministic \cite{mohan21}. The relationship between $\mbb{I}^C$ and $C^\star, C^m$ becomes deterministic.
 \ {\em In the rest of this work we will assume that all MGs are $\mbb{I}$-deterministic.}

\ignore{++{\footnotesize \begin{eqnarray} P^\mc{M}(X^\star=\na~|~\mbb{I}^X=1,X^m=c, \ldots) &=& 1, \label{eq:ass1} \\  P^\mc{M}(X^\star=\na,\ldots~|~\mbb{I}^X=0,X^m=c,\ldots) &=& 0, \label{eq:ass2} \\
P^\mc{M}(X^\star=c^\prime,\ldots~|~\mbb{I}^X=0,X^m=c,\ldots) &=& 0,  \label{eq:ass2prime} \\
 P^\mc{M}(X^\star=c~|~\mbb{I}^X=0,X^m=c, \ldots) &=& 1, \label{eq:ass3}
 \end{eqnarray}}

\vspace{-8mm}{\footnotesize \begin{eqnarray}
 P^\mc{M}(X^\star=\na, \ldots | \mbb{I}^X=1,X^m=c,\ldots) &=& P^\mc{M}(\ldots~|~\mbb{I}^X=1,X^m=c, \ldots), \label{eq:ass4}\\
 P^\mc{M}(X^\star=c, \ldots | \mbb{I}^X=0, X^m=c, \ldots) & =&
 P^\mc{M}(\ldots~|~\mbb{I}^X=0,X^m=c, \ldots). \label{eq:ass5}
\end{eqnarray}}
++}

\section{BIDs and Missing Values}\label{sec:genBIDs}

\vspace{-2mm}
 \h An observed DB $D^\star$, together with a MG $\mc{M}$, can be seen as a single representation for a set, $\mc{W}(D^\star)$, of {\em possible worlds} for a BID, $D^p(D^\star,\mc{M})$. They are {\footnotesize \na}-free databases, $W$, for the underlying schema $\mc{S}$. To build the BID, we start by replacing values from $\nit{dom}$ for each occurrence of an   \na \ in $D^\star$. Then,  each tuple with MVs in $D^\star$ gives rise to a  ``block" of  tuples without \na, each of them with a   probability (see Definition \ref{def:blocks}).  \ The following example shows how we assign probabilities to  tuples in a block.

\begin{example} (ex. \ref{ex:new} cont.) \label{ex:new2}  Assume $\nit{dom}(C)= \{c_1, \ldots, c_5\}$. Table \ref{tab:dbex0}(a) gives rise to the relation in Table \ref{tab:blocks}(a).  The first tuple, fully observed, gives rise to a {\em certain} block with one tuple,  with probability $1$. Tuple $\tau_2$ in $R^\star$ gives rise to  a block, $B(\tau_2)$, of five tuples, $\tau_2^1, \ldots, \tau_2^5$, with
 probabilities $p_2^1:=p^{\nit{BID}}(\tau_2^1), \ldots, p_2^5:= p^{\nit{BID}}(\tau_2^5)$; etc.
\ In this way, we obtain a BID as  in Table  \ref{tab:blocks}(b).
\ In  block
$B(\tau_2)$, the probability of tuple $\tau_2^1$ is defined on the basis of (or conditioned to)  the observed values in the same tuple (similarly for the other tuples):
{\footnotesize \begin{equation} p^{\nit{BID}}(\tau_2^1) := P^\mc{M}(C^m=c_1~|~A^o=a_2,B^o=1,C^\star=\na). \label{eq:bid-tuple}
\end{equation}}

\h By choosing one tuple per block, the BID in Table \ref{tab:blocks}(b)
gives rise to a set $\mc{W}(R^\star)$ of possible worlds, $W$, among them, those in Table \ref{tab:blocks}(c),  where the underlined values are obtained by replacing {\footnotesize \na} \ by domain values.
Their probabilities are: \  $P^{\mc{W}}(R_1) = 1\times p_2^1 \times 1 \times 1 \times 1 \times p_6^5$, \ and \ $P^{\mc{W}}(R_2) = 1\times p_2^2 \times 1 \times 1 \times 1 \times  p_6^3$, resp.  \boxtheorem
 \end{example}

 \vspace{-2mm}
 \captionof{table}{\hspace{0.3cm}(a) Blocks \hspace{1.5cm}  (b) A BID \hspace{1.5cm} (c) Two Possible Worlds}\label{tab:blocks}
 \vspace{-5mm}\begin{center}
 \scalebox{0.6}{
$\begin{tabu}{rr|c|c|c|}\hhline{~----}
&R & A^o & B^o & C^m\\ \hhline{~----}
&\tau_1& a_1 & 0& c_1\\
\hhline{-----}
&\tau_2^1& a_2& 1& c_1\\
&\tau_2^2& a_2& 1& c_2\\
B(\tau_2)\hspace*{-2mm}&\tau_2^3& a_2& 1& c_3\\
&\tau_2^4& a_2& 1& c_4\\
&\tau_2^5& a_2& 1& c_5\\
\hhline{-----}
&\cdots& \cdots&\cdots&\cdots
\\ \hhline{~~---}
\end{tabu}$}~~~~~~~~~\scalebox{0.6}{$\begin{tabu}{c|c|c|c||c|}\hline
R^p & A^o & B^o & C^m& p^\nit{BID}\\ \hline
\tau_1& a_1 & 0& c_1&1\\
\hline
\tau_2^1& a_2& 1& c_1&p_2^1\\
\tau_2^2& a_2& 1& c_2&p_2^2\\
\tau_2^3& a_2& 1& c_3& p_2^3\\
\tau_2^4& a_2& 1& c_4&p_2^4\\
\tau_2^5& a_2& 1& c_5&p_2^5\\
\hline
\cdots& \cdots&\cdots&\cdots& \cdots
\\ \hhline{~----}
\end{tabu}$}~~~~~~ \scalebox{0.6}{$\begin{tabu}{c|c|c|c|}\hline
R_1 & A^o & B^o & C^m\\ \hline
\tau_1& a_1 & 0& c_1\\
\tau_2^1& a_2& 1& \underline{c_1}\\
\tau_3&a_3 & 0&c_3\\
\tau_4& a_4&0&c_4\\
\tau_5& a_5&1&c_5\\
\tau_6^5&a_6&1& \underline{c_5}\\ \hhline{~---}
\end{tabu}$}
~~~~
 \scalebox{0.6}{$\begin{tabu}{c|c|c|c|}\hline
R_2 & A^o & B^o & C^m\\ \hline
\tau_1& a_1 & 0& c_1\\
\tau_2^2& a_2& 1& \underline{c_2}\\
\tau_3&a_3 & 0&c_3\\
\tau_4& a_4&0&c_4\\
\tau_5& a_5&1&c_5\\
\tau_6^3&a_6&1& \underline{c_3}\\ \hhline{~---}
\end{tabu}$}
\end{center}

\ignore{XXXX
\vspace{-2mm}

\captionof{table}{Two Possible Worlds.}\label{tab:two}\vspace{-5mm}
\begin{center}
\scalebox{0.6}{
$\begin{tabu}{c|c|c|c|}\hline
R_1 & A^o & B^o & C^m\\ \hline
\tau_1& a_1 & 0& c_1\\
\tau_2^1& a_2& 1& \underline{c_1}\\
\tau_3&a_3 & 0&c_3\\
\tau_4& a_4&0&c_4\\
\tau_5& a_5&1&c_5\\
\tau_6^5&a_6&1& \underline{c_5}\\ \hhline{~---}
\end{tabu}$
~~~~
$\begin{tabu}{c|c|c|c|}\hline
R_2 & A^o & B^o & C^m\\ \hline
\tau_1& a_1 & 0& c_1\\
\tau_2^2& a_2& 1& \underline{c_2}\\
\tau_3&a_3 & 0&c_3\\
\tau_4& a_4&0&c_4\\
\tau_5& a_5&1&c_5\\
\tau_6^3&a_6&1& \underline{c_3}\\ \hhline{~---}
\end{tabu}$
}
\end{center}
XXXXX}

\h By definition fo BID, a {\em possible world} in $\mc{W}(D^\star,\mc{M})$ is an instance $W$ for   schema $S$, with relations $R^W$ that contain, for each $\tau \in R^\star$, only one $\tau^\prime \in B(\tau)$, and nothing more.  Possible worlds  may contain duplicates
which we tell apart via tids. \ In Example \ref{ex:first}, duplicates are generated from $\tau_2$ and $\tau_6$ by replacing \na \ by $c_3$ in the former, and by $a_2$ and  $c_3$ in the latter.

\begin{definition} \em
\label{def:blocks}
Consider an MG $\mc{M}$ and  an observed instance $D^\star$ for schema $S^\star$, and $R^\star$ a relation in $D^\star$ with schema $R(\bar{A}^o,\bar{A}^\star)$, where $\bar{A}^o, \bar{A}^\star$ are lists of fully observed, and possibly taking \na \ attributes, resp.  Let $\tau$ be a tuple in $R^\star$, and $\tau[\bar{A}^o,\bar{A^\prime}^\star]$ its restriction to those attributes without an {\footnotesize \na}, with $\bar{A^\prime}^\star \subseteq \bar{A}^\star$.

\vspace{1mm}\noindent
(a) The {\em block} associated to $\tau$ is the set of tuples for schema $R(\bar{A}^o,\bar{A}^m)$: \ $B(\tau) := \{ \ \tau^\prime~|~\tau^\prime[\bar{A}^o,\bar{A^\prime}^\star] = \tau[\bar{A}^o,\bar{A^\prime}^\star], \mbox{ and},
\mbox{ for each } A \in (\bar{A}^\star \smallsetminus  \bar{A^\prime}^\star), \  \tau^\prime[A] \in \nit{dom}(A^m)\}$.

\vspace{1mm}
\noindent (b)
For $\tau^\prime \in B(\tau)$, its  probability $p(\tau^\prime)$ is the conditional probability: $p^{\nit{BID}}\!(\tau^\prime) := P^\mc{M}(\tau^\prime[\bar{A}^\star \smallsetminus  \bar{A^\prime}^\star]~|~\tau)$.
\ignore{++\begin{equation}
p^{\nit{BID}}\!(\tau^\prime) := P^\mc{M}(\tau^\prime[\bar{A}^\star \smallsetminus  \bar{A^\prime}^\star]~|~\tau).\label{eq:probTuple}
\end{equation}++}
 The probability of the tuple in a singleton block  is $1$.

\vspace{1mm}\noindent (c)
$D^p(D^\star,\mc{M})$ denotes the BID whose relations $R^p$ contain the blocks $B(\tau)$ for $\tau \in R^\star$, and each tuple $\tau^\prime \in R^p$ has probability $p^{\nit{BID}}(\tau^\prime)$.

\vspace{1mm}\noindent (d)
A {\em possible world} associated to $D^\star$ is an instance $W$ for the underlying  schema $S$, with relations $R^W$ that contain, for each $\tau \in R^\star$, only one $\tau^\prime \in B(\tau)$, and nothing more. \ $\mc{W}(D^\star,\mc{M})$ denotes the set of possible worlds; and $P^{\mc{W}}$, its probability distribution. \boxtheorem
\end{definition}

\vspace{-8mm}
\begin{example}\label{ex:connection}  \ (ex. \ref{ex:new2} cont.) \ It is easy to check that, starting from (\ref{eq:bid-tuple}), the  probabilities $p_2^i$ of the tuples $\tau_2^i$,  \ $i=1,\ldots,5$, in block $B(\tau_2)$ in Table \ref{tab:blocks}(b) are as follows:\footnote{Details and additional examples are given in Appendix \ref{app:example}.}

\vspace{-5mm}
{\footnotesize
\begin{eqnarray*}
    p^\nit{BID}(\tau_2^i)\!\!=\!\!\frac{P^\mc{M}(A^o=a_2,B^o=1,C^m=c_i, \mbb{I}^C=1)}{P^\mc{M}(A^o=a_2, B^o=1,\mbb{I}^C=1)} = P^\mc{M}(C^m =c_i).
\end{eqnarray*}}

\vspace{-7mm}\boxtheorem
\end{example}

\vspace{-2mm}
\h Each possible world $W$ of the BID becomes a possible \na-free version of the observed  DB $D^\star$, obtained by multiple imputation \cite{gelman}.
\ Since each possible world $W$ is an instance for the underlying schema $S$, a query $\mc{Q}$ posed to the BID  will be expressed in language of the underlying schema, using only attributes of the forms $A^o$ and $A^m$.
 {\em Then, by definition, {\em posing a query $\mc{Q}$ to the observed database $D^\star$} means querying the associated BID $D^p(D^\star,\mc{M})$.}

\h In this work, we will consider categorical attributes. \ If we assume that the attributes' domains coincide with their active domains (i.e. their values in $D^\star$) or are polynomially-bounded in the size of $D^\star$, the size of $D^p(D^\star,\mc{M})$ is polynomially-bounded in the size of $D^\star$. $\mc{M}$ is bound to be small in comparison with $D^\star$.

\h The number of possible worlds for  $D^p(D^\star,\mc{M})$ can be exponential in the size of $D^\star$. QA becomes hard in data complexity, because QA for BCQ on TIDs, which is $\#P$-hard (in data) \cite{suciu}, can be reduced in polynomial-time (in data) to QA under an observed DB with its MG.

    \begin{theorem} \label{thm:sharp} \em Query answering of BCQs  on observed databases with missing values and its MG via the generated BIDs is $\#P$-hard in data complexity. \boxtheorem
         \end{theorem}

 \vspace{-3mm}  \h  For the proof, a reduction from QA on TIDs  builds in polynomial-time in data, for a TID $D^{\nit{tid}}$ and a BCQ $\mc{Q}$ for the latter, an MG $\mc{M}$ representing a MAR case of MM, a BID $D^{\nit{bid}}$ (for an associated observed instance with MVs), and a BCQ $\mc{Q}^\prime$, such that the answer to $\mc{Q}$ from $D^{\nit{tid}}$ and that of $\mc{Q}^\prime$ from $D^{\nit{bid}}$ coincide.  $\mc{M}$ and $\mc{Q}^\prime$ are short. $\mc{Q}$ can be chosen to be self-join free and non-hierarchical, for which QA on TIDs is $\#P$-hard \cite{suciu}.  See Appendix \ref{sec:theo1}.

  \ignore{    \comlb{
What about other MMs? For the meta-reviewer it seemed to matter.}}

  \h  Instead of going deeper into the investigation of general QA on the resulting BIDs, we will propose, starting in Section \ref{sec:queries},  an alternative QA semantics. The presence of duplicates will be particularly important.

\vspace{-1mm}
\section{Class-Based Data Semantics}
\label{sec:queries}

\vspace{-1mm}
\h Our data semantics leverages the stochastic and statistical origins of the observed DB $D^\star$.  We adopt a two-dimensional perspective: \ (a) Each possible world is associated with a probability that reflects its stochastic uncertainty. This is what we have so far. \  (b) The second dimension is of a statistical nature: A notion of \emph{compliance}  quantifies how well a possible world aligns with the joint distribution induced by the MG. This second dimension is developed in detail in Section \ref{sec:compliance-pw}. In this section, we prepare the ground, and bring up some relevant issues using our running example.

\vspace{-1mm}
\begin{example}  \label{ex:aggre} (ex. \ref{ex:connection}  cont.)
  Consider the MG in  Figure \ref{table:first}(b),  the observed DB in Table \ref{fig:two}(a), and $\nit{dom}(C^m) = \{0,1,2\}$. \
   Assume: $P^{\mc{M}}(C^m = 0)=\frac{1}{2}$, $P^{\mc{M}}(C^m =1)= P^{\mc{M}}(C^m = 2) =\frac{1}{4}$, that, by Example \ref{ex:connection},  is all we need to build the BID in Table \ref{fig:two}(b). Table \ref{fig:two}(c) shows the most-probable possible world. It has duplicates.

 \vspace{2mm}
\captionof{table}{ \ (a) Observed $D^\star$  \hspace{1.5cm}(b) BID $D^p$ \hspace{2.5cm} (c) MPD w/prob. $(\frac{1}{2})^3$}\label{fig:two} \vspace{-4mm}
 \begin{center}
 {\tiny   $\begin{tabu}{c|c|c|l|}
\hhline{~---}
&A^o & B^o & C^\star \\ \hhline{~---}
\tau_1 &a & 0  & 0  \\ \hhline{~---}
\tau_2 &a & 0  & 0  \\ \hhline{~---}
\tau_3&a & 1  & \na \\ \hhline{~---}
\tau_4&a & 1  & 0  \\ \hhline{~---}
\tau_5&a & 1  & \na \\ \hhline{~---}
\tau_6&a & 1  & 1  \\ \hhline{~---}
\tau_7&a & 1  & \na \\ \hhline{~---}
\tau_8&a & 1  & 2  \\ \hhline{~---}
\end{tabu}$}~~~~~~~~~~~
\scalebox{0.6}{
\begin{tabular}{l|l|l|l||l|}
\cline{2-5}
              & $A^o$ & $B^o$ & $C^m$ & $p^\nit{BID}$ \\ \hline
$\tau_1$                  & $a$ & 0  & 0  & 1 \\ \hline
$\tau_2$                  & $a$ & 0  & 0  & 1 \\ \hline
\multirow{3}{*}{$B(\tau_3)$} & $a$ & 1  & 0  &  $1/2$ \\ \cline{2-5}
                    & $a$ & 1  & 1  & $1/4$   \\ \cline{2-5}
                    & $a$ & 1  & 2  & $1/4$  \\ \hline
$\tau_4$                  & $a$ & 1  & 0  & 1 \\ \hline
\multirow{3}{*}{$B(\tau_5)$} & $a$ & 1  & 0  & $1/2$  \\ \cline{2-5}
                    & $a$ & 1  & 1  &  $1/4$ \\ \cline{2-5}
                    & $a$ & 1  & 2  & $1/4$  \\ \hline
$\tau_6$                  & $a$ & 1  & 1  & 1 \\ \hline
\multirow{3}{*}{$B(\tau_7)$} & $a$ & 1  & 0  & $1/2$  \\ \cline{2-5}
                    & $a$ & 1  & 1  &  $1/4$ \\ \cline{2-5}
                    & $a$ & 1  & 2  &  $1/4$ \\ \hline
$\tau_8$                  & $a$ & 1  & 2  & 1 \\ \hline
\end{tabular}}~~~~~~~~~~~
{\tiny $\begin{tabu}{c|c|c|c|}
\hhline{~---}
&A^o & B^o & C^m \\ \hhline{~---}
\tau_1&a & 0  & 0 \\ \hhline{~---}
\tau_2&a & 0  & 0 \\ \hhline{~---}
\tau_3^1&a & 1  & 0 \\ \hhline{~---}
\tau_4&a & 1  & 0 \\ \hhline{~---}
\tau_5^1&a & 1  & 0 \\ \hhline{~---}
\tau_6&a & 1  & 1 \\ \hhline{~---}
\tau_7^1&a & 1  & 0 \\ \hhline{~---}
\tau_8&a & 1  & 2 \\ \hhline{~---}
\end{tabu}$}
\end{center}

 \h  Table
\ref{tab:ex-worlds} shows all possible worlds. World $W_1$ in it is that shown in Table \ref{fig:two}(c). We use the  notation $W^v$ for possible worlds where $v$ is a vector of length $n$, the number of non-singleton blocks in the BID, with values from $\nit{dom}(C^m)$. Here, $n=3$. World $W^{[0,0,0]}$ is obtained by choosing $0$ for the MVs for $C^m$ in each of the three blocks of the BID. Similarly, $W^{[1,2,0]}$ is obtained by choosing $1$ for block 1; $2$ for  block 2; and $0$ for block 3. The last column of Table \ref{tab:ex-worlds} shows the probability $P^\mc{W}(W)$ of  each world $W$. \boxtheorem
\end{example}

\vspace{-8mm}
 \begin{center}
 \captionof{table}{ \ Possible Worlds' Probabilities}
     \label{tab:ex-worlds}\vspace{-2mm}
\tiny{
$
 \begin{array}{|c|c|l|}
 \hline
  \mathbf{World} &
 \mathbf{Notation} &  ~~~\mathbf{P^\mc{W}(W)}   \\ \hline
W_1 & W^{[000]} &  (1/2)^3=\mbf{0.125} \\ \hline
W_2 & W^{[001]} &  (1/2)^2\times 1/4=0.06  \\ \hline
W_3 &  W^{[002]}  &(1/2)^2\times 1/4=0.06  \\ \hline
W_4 &  W^{[010]}  & (1/2)^2\times 1/4=0.06 	  \\ \hline
 %
 %
\ldots &  \ldots  & \ldots  \\ \hline
 %
 %
 %
 %
 %
 %
W_{20} & W^{[201]} & (1/4)^2\times 1/2=0.03 \\ \hline
W_{21} &  W^{[202]}  & (1/4)^2\times 1/2=0.03 \\ \hline
W_{22}  & W^{[210]}  &(1/4)^2\times 1/2=0.03 \\ \hline
\ldots &  \ldots  & \ldots  \\ \hline
%
 %
W_{26}  & W^{[221]}  &(1/4)^3=0.015 	 \\ \hline
W_{27} & W^{[222]}  & (1/4)^3=0.015   \\ \hline
 \end{array}$}
\end{center}

\ignore{
\vspace{-2mm}
\subsection{A Class-Based Data and QA Semantics}\label{sec:classes}
}

\h The example shows that some  possible worlds coincide as multisets, except for the tids, for example,  $W^{[002]}$ and $W^{[020]}$. \
Given the {\em accidental} coincidence of worlds due to the implicit imputation process, we
will group together worlds into {\em classes}, each class containing the worlds that coincide as multisets modulo tids. Possible worlds in a same class return the same answer to a query.

\vspace{-1mm}
\begin{definition} \label{def:matching} \em (a) Possible worlds $W,W^\prime \in \mc{W}(D^\star,\mc{M})$ obtained from a BID  are  {\em matching} if they become the same multiset instance when stripped from tids. \
(b) A {\em class of possible worlds}, $C$, is a maximal subset (under set-inclusion) of $\mc{W}$ where all  worlds in it are matching with each other. \  $\mc{C}(D^\star,\mc{M})$ denotes the collection of all classes. \
(c) The probability of a class $C$ is: \ $
P^\mc{C}\!(C) \ := \  \sum_{W \in C} P^\mc{W}\!(W)$. \
(d) $C$ is a {\em most probable class} (MP-class) if $P^\mc{C}(C)$ takes a maximum  in $\mc{C}(D^\star,\mc{M})$.   \ (e) Given a query $\mc{Q}$ (in the  the language of  schema $S$), and a class $C \in \mc{C}(D^\star,\mc{M})$, the answer to $\mc{Q}$ from $C$ is \ $\mc{Q}[C] := \mc{Q}[W]$, with arbitrary $W \in C$, and its probability is $P^\mc{C}(C)$.
 \boxtheorem
\end{definition}

\vspace{-6mm}
\begin{example} \label{ex:aggre+}  (ex. \ref{ex:aggre} cont.) Table \ref{tab:classes2prime} shows the different classes, the worlds they contain,  and their probabilities. \ With the aggregate query $\mc{Q}\!: \! \nit{sum}(C^m)$ over the partially observed attribute $C^m$, we obtain the  answers in the right-most column of Table \ref{tab:classes2prime}.
\ignore{$
\{\langle \nit{sum} \! = \! 3; \ 0.125 \rangle,\ \langle \nit{sum} \! = \! 4; \ 0.1875 \rangle,\ \langle \nit{sum} \! = \! 5; \ 0.28125 \rangle,\ \langle \nit{sum} \! = \! 6; \ 0.203125 \rangle,\ \langle \nit{sum} \! = \! 7; \ 0.140625 \rangle,\ \langle \nit{sum} \! = \! 8; \ 0.046875 \rangle,\ \langle \nit{sum} \! = \! 9; \ 0.015625 \rangle\}$.}
 Answer  $\nit{sum} = 5$ is obtained from class $C_1$,  one of the most probable classes, with probability $0.188$ on it.
\boxtheorem
\end{example}

\vspace{-7mm}
\begin{center}
\captionof{table}{Classes and Answers  (bag semantics).}\label{tab:classes2prime}
\vspace{-2mm}\scalebox{0.6}{ $\begin{tabu}{|c|c|c|c|c|}
 \hline
  \text{Classes}     & \text{Worlds}  & P^\mc{C}(C) &   \nit{sum}(C^m)[C]   \\ \hline
C_1 & W^{[002]},  W^{[020]},   W^{[200]} 	& 0.188 & 5 \\ \hline
C_2 & W^{[011]},  W^{[101]}, W^{[110]}  &	0.094  & 	 5 \\ \hline
C_3 &  	W^{[012]}, W^{[021]}, W^{[102]},  &0.188&  6\\
&W^{[120]},W^{[201]},W^{[210]}& &
\\ \hline
C_4 & W^{[111]} 	&0.016 &  6 \\ \hline
 C_5 & W^{[001]}, W^{[010]}, W^{[100]} &0.188	 & 4	\\ \hline
C_6 & W^{[022]}, W^{[202]}, W^{[220]} & 0.094	  & 7	\\ \hline
C_7 & W^{[112]}, W^{[121]}, W^{[211]}  & 0.047	  & 7	\\ \hline
C_8 & W^{[000]} 	 & 0.125   & 3\\ \hline
C_9 & W^{[122]}, W^{[212]}, W^{[221]} & 0.047  & 8	\\ \hline
C_{10} & W^{[222]}  & 0.016 4 & 9	 \\ \hline
 \end{tabu}$}
\end{center}

\ignore{
\vspace{-3mm}
\begin{example} \label{ex:sum} (ex. \ref{ex:aggre+} cont.)
  Table \ref{tab:classes2prime} shows the classes obtained from  the worlds in Table \ref{tab:ex-worlds}.
  \ The query answers  from the  classes are in the last column of Table \ref{tab:classes2prime}. All worlds in the same class yield the same answer for every query, under the set or bag semantics.
  \  Classes $C_1$ and $C_2$ have different probabilities. Worlds in a same class may have different global probabilities.
\boxtheorem
\end{example}
}

\vspace{-1mm}
\begin{remark} \label{rem:canon} ({\bf canonical representation of classes})
Given the BID $D^p(D^\star,\mc{M})$, with a set of $n$ blocks $\mc{B} = \{B_1, \ldots, B_n\}$, including singleton blocks, a class $C$ of worlds is determined by the multiplicities of the tuples they contain.  Let $T = \langle t_1,\ldots, t_m\rangle$, called the {\em support} of the BID, be the {\em vector of distinct} tuples appearing in the BID (without considering the tids).  So, $T$ has a fixed enumeration of tuples. (We will still use the notation $t \in T$ and $|T|$.)
Let $n_j$ denote the multiplicity of tuple $t_j$ across the blocks in $\mc{B}$, i.e. its  number of occurrences in the BID. Accordingly, a class $C$ is uniquely determined by a vector of positive integers $\mathbf{k}=\langle k_1, \ldots, k_m\rangle$, with $\sum_{j \in [1,m]} k_j=n$, where each $k_j \leq n_j$ is the number of occurrences of  $t_j$ in $C$. This class is denoted with $C_{\mathbf{k}}$. \ We denote with $\nit{adm}(\mathbf{k})$
 the fact that $\mathbf{k}$ is {\em admissible}, that is, there is a class $C_{\mathbf{k}}$ for $D^p(D^\star,\mc{M})$. When {\em searching} for such a $\mathbf{k}$, i.e. searching for a class with good properties, we have to check admissibility.
\boxtheorem
\end{remark}

\vspace{-6mm}
\begin{example} (ex. \ref{ex:aggre} cont.) \label{ex:support} Consider the BID in Table \ref{fig:two}(b). Here, $n = 8$, and $T = \langle(a,0,0),$ $ (a,1,0),$ $  (a,1,1), (a,1,2)\rangle$, in this order, with  $m=4$. For  $t_3 \ = (a,1,1)$,  $n_3 = 4$. Class $C_2$ contains $W^{[011]} = \{(a,0,0),$ $ (a,1,0), (a,0,0), (a,1,0), \ (a,1,1), \ (a,1,1),$ $ (a,1,2), (a,1,1)\}$, and its matching worlds (see Table \ref{tab:classes2prime}). $C_2$  is characterized by the vector $\mathbf{k} = \langle 2,2,3,1\rangle$, and denoted $C_{\langle 2,2,3,1\rangle}$. \boxtheorem
\end{example}

\ignore{\vspace{-4mm}
\begin{remark} \label{rem:admiss} Given a BID $D^p(D^\star,\mc{M})$ as in Remark \ref{rem:canon}, and a vector of integers $\mathbf{k}$, we denote with $\nit{adm}(\mathbf{k})$
 the fact that $\mathbf{k}$ is admissible, that is, there is a class $C_{\mathbf{k}}$ for $D^p$. In particular, $\sum_{j=1}^m k_j = n$.  When {\em searching} for such a $\mathbf{k}$, in essence searching for a class with good properties, we will assume admissibility can be checked (or, equivalently, that there are cardinality constraints in place that enforce it). We expect the BID to be clear from the context.\boxtheorem
\end{remark}}

\vspace{-3mm}
\h Now, we can consider a set $\mc{C}^\nit{pref} \subseteq \mc{C}$ of {\em preferred classes}, those with a desired property.  For example, we already have the {\em most-probable classes}. In Section \ref{sec:compliance-pw}, we will consider  {\em most-compliant classes}.  QA can be defined in general, on  an arbitrary set of preferred classes.

\vspace{-1mm}
\begin{definition} \label{def:preferred} \em (Preferred QA-Semantics)  Given $\mc{C}^\nit{pref} \subseteq \mc{C}(D^\star,\mc{M})$ and a
relational query $\mc{Q}$: \   The {\em set of all preferred answers}  is:
\ $\nit{Ans}(\mc{Q},\mc{C}^\nit{pref}) \ := \ \{\langle \mc{Q}[C], P^\mc{C}(C)\rangle \ | \ C \in \mc{C}^\nit{pref\!}\}$.
\boxtheorem
\end{definition}

\vspace{-5mm}
\h Several computational problems arise in relation to this general formulation of QA, and the notion of preferred class.  Some problems will be presented in  Section \ref{sec:complexity}, where we will concentrate mostly on class-related computational problems, leaving QA aside, which is easier than computing classes with certain properties.

\vspace{-2mm}
\section{Possible-\!World Compliance} \label{sec:compliance-pw}

\vspace{-1mm}
\h We can go beyond the purely probabilistic dimension by introducing a new and natural dimension for data and QA semantics: \textit{world compliance}. It quantifies how well a possible world conforms to the joint distribution induced by the MG. \ In  Section \ref{sec:mpdbs}, we will   develop the probabilistic dimension.

\h Some possible worlds (or classes thereof)  may be {\em more compliant} than others w.r.t. the underlying MG $\mc{M}$, with compliance as a measure {\em statistical fidelity}. It can be cast in terms of a {\em statistical distance}\footnote{We use the term ``statistical distance" for a measure of discrepancy between two probability distributions,  encompassing both metric distances, and divergence measures. We will simply talk about ``a distance".} between the data distribution of a world $W$ (and of those in its class) and the  distribution induced by $\mc{M}$.
\ In our running example,
some of the worlds that contribute to the most probable answer, $\nit{sum} = 5$, such as  $W^{[110]}$ and  $W^{[101]}$, may not be the most  compliant. \ignore{Later in this Section, we will define, as particular cases of Definition \ref{def:preferred}, some {\em preferred classes}, $\mc{C}^\nit{mc} \subseteq \mc{C}(D^\star,\mc{M})$,  that contain worlds that are most compliant.}

 \h In order to define compliance,  we start with the {\em empirical distribution} of a world (which is the same for all the worlds in its class). It is based on the multiset nature of worlds.

 \begin{definition} \label{def:emp} \em
 Given a BID $D^p(D^\star,\mc{M})$, and its {\em support} $T$ (see Remark \ref{rem:canon}), the {\em empirical distribution}, $P_{W}^{\!E}$, of a possible world $W \in \mc{W}(D^\star,\mc{M})$ as a multiset is, for $t \in T$: \ $P_{W}^{E}(t)  :=  \nit{mult}^W\!(t)/||W||$,
with $\nit{mult}^W\!(t)$ the multiplicity of $t$ in $W$, and $||W||$ is the bag-cardinality of $W$, i.e. counting duplicates.  \boxtheorem
\end{definition}

\vspace{-4mm}
\h Recall that the support $T$ does not have duplicates. A tuple in it that does not belong to a particular $W$ has empirical probability $0$ in that world.
The empirical distribution  will be compared with  $P^\mc{M}$, the distribution {\em induced} by the MG $\mc{M}$, as a marginal for the underlying schema $S$.  The comparison is made considering only tuples in $T$.

\vspace{2mm}
\captionof{table}{ \ A World and its Distributions}\label{tab:distr} \vspace{-5mm}
\begin{center}
\scalebox{0.6}{\begin{tabular}{l|l|l|l|}
\cline{2-4}
        $W^{[002]}$      & $A^o$ & $B^o$ & $C^m$ \\ \hline
$\tau_1$                  & $a$ & 0  & 0  \\ \hline
$\tau_2$                  & $a$ & 0  & 0  \\ \hline
$B(\tau_3)$ & $a$ & 1  & 0  \\  \hline
$\tau_4$                  & $a$ & 1  & 0 \\ \hline
$B(\tau_5)$ & $a$ & 1  & 0   \\ \hline
$\tau_6$                  & $a$ & 1  & 1  \\ \hline
$B(\tau_7)$
                    & $a$ & 1  & 2  \\ \hline
$\tau_8$                  & $a$ & 1  & 2   \\ \hline
\end{tabular}}~~~~~~~~~
\scalebox{0.6}{\begin{tabular}{l|l|l|l||c|c|}
\cline{2-6}
           $T$  &\!$A^o$\!&\!$B^o$\!&\!$C^m$\!&\!$P^E_{\!W[002]}$\!&\!$P^\mc{M}$\!\\ \hline
$t_1$                  & $a$ & 0  & 0 & 1/4 &0.225\\ \hline
$t_2$ & $a$ & 1  & 0  & 3/8& 0.225\\   \hline
$t_3$                  & $a$ & 1  & 1  & 1/8 &0.1125\\ \hline
$t_4$
                    & $a$ & 1  & 2 & 1/4 & 0.1125\\ \hline
\end{tabular}}
\end{center}

\vspace{-4mm}
\begin{example}\label{ex:emp} (ex. \ref{ex:support} cont.) Consider  world $W^{[002]}$ in Table \ref{tab:distr} belonging to class $C_1$ in Table \ref{tab:classes2prime}; and its  empirical and $\mc{M}$-induced distributions on the set of its four different tuples, those belonging to support $T$. \
The two other worlds in  $C_1$ share the same empirical distribution. \ignore{Since all the tuples shown in $D^p$ belong to this world, the support is $T = \{t_1,t_2,t_3,t_4\}$.}

\h  Let us assume, consistently with Example \ref{ex:aggre}, that $\nit{dom}(C^m) = \{0,1,2\},$ $ \nit{dom}(A^o) = \{a,b\},$ $ \nit{dom}(B^o) = \{0,1\}, \nit{dom}(C^\star) = \{0,1,2,\na\}$, and also: {\footnotesize $
P^\mc{M}(C^m=0)=1/2,
P^\mc{M}(C^m=1) = P^\mc{M}(C^m=2) = 1/4,
P^\mc{M}(B^o=0) = P^\mc{M}(B^o=1) = 1/2,
P^\mc{M}(A^o=a)=0.9,  \
P^\mc{M}(A^o=b) = 0.1$}.

 \h As noted in Example \ref{ex:aggre}, with these probabilities we can compute the induced probabilities (using conditions (I) and (II) on MGs, and   (\ref{eq:joint})): \  For $x \in \{a,b\}, y \in \{0,1\}, z \in \{0,1,2\}$:\ignore{\footnote{Details in Example \ref{ex:emp2} in the Appendix.}}
{\footnotesize $P^\mc{M}(A^o\!=\!x,B^o\!=\!y,C^m\!=\!z)= \!\!\!\!\!\sum\limits_{u \in \{0,1\}, v \in \{0,1,2,\na\}} \hspace{-8mm}P^\mc{M}(A^o\!=\!x,B^o\!=\!y,C^m\!=\!z,\mbb{I}^C\!=\!u, C^\star\!=\!v)
=
P^\mc{M}(C^m\!=\!z) \ \times$ $P^\mc{M}(B^o\!=\!y) \! \times \! P^\mc{M}(A^o\!=\!x)$}.
\ignore{ Accordingly, {\footnotesize $P^\mc{M}(a,0,0) = P^\mc{M}(a,1,0) = 0.9 \times 1/2 \times 1/2= 0.225$,  $P^\mc{M}(a,1,1) = P^\mc{M}(a,1,2) = 0.9 \times 1/2 \times 1/4= 0.1125$}.} We obtain the last column, $P^\mc{M}$, in Table \ref{tab:distr}, also shared with the other worlds in class $C_1$.
\boxtheorem
\end{example}

\vspace{-2mm}
\noindent {\bf Distance-Based Compliance.}
\ If  we have  an abstract measure of {\em distance}, $d(\cdot,\cdot)$,  between an empirical distribution   and  the joint distribution $P^\mc{M}$, we can define compliance.

\begin{definition} \label{def:dist} \em
 (a) The {\em  compliance degree} of $C \in \mc{C}(D^\star,\mc{M})$ is: \
    $d^c(C, \mc{M})  \ := \ d(P^E_W,P^\mc{M}), \ \mbox{ with any } W \in C$. \
(b) $C$ is a {\em most-compliant class} (an MC-class) w.r.t. $d$ if $d^c(C, \mc{M})$ takes a minimum value in $\mc{C}(D^\star,\mc{M})$.  $\mc{C}^\nit{mc}_{\!d}$ denotes the collection of MC-classes relative to $d$.
 \boxtheorem
\end{definition}

\vspace{-5mm}\h The distance in (a) is well defined since all worlds in a class  have the same empirical distribution. \
Several distances between probability distributions offer themselves to define the compliance degree, among them the {\em Kullback-Leibler Divergence} (KLD) \cite{larry}:

\vspace{-2mm}
{\footnotesize \begin{equation}
d^{\mathrm{KL}\!}(P^E_W, P^\mc{M}) \ := \ \mathrm{Div_{KL}}(P_W^E \parallel P^\mc{M})
\ := \ \sum_{t \in T} P_W^E(t) \, \ln \frac{P_W^E(t)}{P^\mc{M}(t)}, \label{eq:kld}
\end{equation}}

\vspace{-2mm} so as the {\em Euclidean Distance}, $d^{\nit{EU}}$\ignore{ \cite{hastie}}\!\!.  For a class $C$, $\nit{KLD}(C)$ denotes the KL-divergence in common to all worlds in $C$ to $P^\mc{M}$.
Any two matching worlds, $W, W^\prime$, become equally compliant, but they may have different global probabilities, $P^\mc{W}(W)$ and $P^\mc{W}(W^\prime)$.

\begin{center}
\captionof{table}{Classes and Compliance  (bag semantics). }\label{tab:classes2}
\vspace{-2mm}\scalebox{0.6}{
 $\begin{tabu}{|c|c|c|c|}
 \hline
  \text{Classes} \  C   & \text{Worlds}  & P^\mc{C}(C) &  \nit{KLD}(C)   \\ \hline
C_1 & W^{[002]},  W^{[020]},   W^{[200]} 	& 0.188 & 0.431\\ \hline
C_2 & W^{[011]},  W^{[101]}, W^{[110]}  &	0.094 &0.518\\ \hline
C_3 &  	W^{[012]}, W^{[021]}, W^{[102]},
 &0.188&    0.452\\
&W^{[120]},W^{[201]},W^{[210]} & &
\\ \hline
C_4 & W^{[111]} 	&0.016&   0.711\\ \hline
 C_5 & W^{[001]}, W^{[010]}, W^{[100]} &0.188	&  0.431\\ \hline
C_6 & W^{[022]}, W^{[202]}, W^{[220]} & 0.094	&  0.711\\ \hline
C_7 & W^{[112]}, W^{[121]}, W^{[211]}  & 0.047	&  0.929\\ \hline
C_8 & W^{[000]} 	 & 0.125 & 	 0.431 \\ \hline
C_9 & W^{[122]}, W^{[212]}, W^{[221]} & 0.047 &  0.972\\ \hline
C_{10} & W^{[222]}  & 0.016 &  1.111\\ \hline
 \end{tabu}$
 }
\end{center}

\begin{example} \label{ex:sum+} (ex. \ref{ex:emp} cont.) With (\ref{eq:kld}), we can compute the classes' KL-divergences to their induced distributions.
 Table \ref{tab:classes2}  shows them.
 The most-compliant classes  are $C_1, C_5, C_8$. The most probable classes are $C_1, C_3, C_5$.
\boxtheorem
\end{example}

\vspace{-3mm}\noindent {\bf Notation:} Consistently with Definition \ref{def:dist}(b), $\mc{C}^\nit{mc}_\nit{KL}(D^\star,\mc{M})$ and $\mc{C}^\nit{mc}_\nit{EU}(D^\star,\mc{M})$ denote the set of MC-classes with respect to the KL- and Euclidean distances, resp.

\vspace{2mm}
\noindent {\bf Goodness-of-Fit Compliance.} \
An alternative take  on possible-world compliance is based on {\em hypothesis testing}. Given  $W \in \mc{W}$ as a sample, we {\em test the hypothesis}, $H_0$, that  it  fits the induced distribution $P^\mc{M}$. \ We use the $\chi^2$-statistic, with $T = \{t_1, \ldots, t_m\}$:

\vspace{-1mm}
{\scriptsize $$\chi^2(W) := \sum_{i=1}^{m} \frac{(P^E_W(t_i) \ - \ P^\mc{M}(t_i))^2}{P^\mc{M}(t_i)},$$}

\vspace{-1mm}which, under $H_0$, has approximately a $\chi^2_{m-1}$-distribution \cite{larry}. It can be seen as a distance, actually a measure of the {\em relative square deviation} of $P^E_W$ from $P^\mc{M}$.

\h In order to compare worlds, we can use a {\em compliance order} based on the $p$-values for the test:    $p^{\!V\!\!}(W) \ := \ P_{H_0}(\chi^2 \ \geq \  \chi ^2(W))$, that is the probability (under $H_0$) that $\chi^2$ -as a random variable- is at least as contradictory to $H_0$ as the value of $\chi^2(W)$. It is defined by: $
W_1 <^\mc{M}_\nit{pv} W_2$ iff $p^{\!V\!\!}(W_1) \ < \ p^{\!V\!\!}(W_2)$. \ In this way, we have the family $\mc{C}^\nit{mc}_{pv}$ of most-compliant classes based on  the {\em p-value of the $\chi^2$-test}.

\vspace{-2mm}
\section{Computing Most-Compliant Classes}
\label{sec:complexity}

\vspace{-1mm}
\h With a general notion of    most-compliant  class, we can turn to
 computational problems, {\em the MC-problems}.   They will be formulated and addressed  using  the class representation in Remark \ref{rem:canon}. In particular,
we consider   classes $C_{\mathbf{k}}$, and a representative world therein, $W_{\mathbf{k}}$, that  can also be seen as encoding its class's degree of compliance.

\ignore{\footnote{These semantics provide alternative perspectives on QA, depending on whether the focus is on statistical likelihood or distributional fidelity.}}
\ignore{++Among them,  given a BID $D^p(D^\star,\mc{M})$, compute a most-compliant (MC) class, or compute a most- probable (MP) class. Given also a
 query $\mc{Q}$,  compute a world $W$ from a MC-class class $C$, and return $\mc{Q}[W]$ and $P^\mc{C}(C)$. Similarly with $C$ a MP-class. We may also want to do this for all MC- or all MP-classes. \
 Diverse computational problems++} \ignore{++emerge, among them, decision, functional, optimization, enumeration \cite{strozecki2023}, and counting problems \cite{toran91}.++}

\ignore{++\h We will concentrate first and mostly on most-compliant-related computational problems; and next we will provide some results on the most-probable-related problems.   They++}


\h The notion of {\em convex-separable} distance \cite{ahuja1993},  $d(\cdot,\cdot)$, will be
critical. Intuitively, it  can be computed by
aggregating independent contributions associated with individual tuples. In this way,
the total distance decomposes as a sum of per-tuple convex terms:
$d(C_{\mathbf{k}}, \mathcal{M})=\sum_{j=1}^m d_j\!\bigl(k_j,
P^\mathcal{M}(t_j)\bigr)$, where each $d_j$ is a convex function capturing the
local cost of the assignment made to tuple $t_j$. The distance is
\emph{strictly} convex-separable when, in addition, each $d_j$ is strictly
convex in $k_j$. Equivalently, the discrete marginal contributions
$(d_j(k_j{+}1,\cdot)-d_j(k_j,\cdot))$ are strictly increasing.

\h Many classical measures of discrepancy are convex-separable. These
include the $f$-divergences~\cite{ali1966} --- such as the KL-divergence,
the (squared) Hellinger distance, the total variation distance \cite{villani},
and the $\chi^2$-statistic --- as well as the
$L^p$ distances~\cite{royden1988}, for $1 \leq p < \infty$, in particular
the $L^1$ Manhattan distance and the squared $L^2$ Euclidean distance.
All of these are strictly convex-separable, except for the total variation
and the $L^1$ distance, which, being piecewise linear, are convex but not
strictly so.

\h Theorem \ref{th:complresultsmcc} below introduces both the computational problems and their computational complexities, in data complexity, i.e. in the size of the BID; and then, also in $|D^\star|$.

\begin{theorem} \label{th:complresultsmcc} \em
(MC-related problems)
Assume $D^p(D^\star,\mc{M})$ is a BID with $n$ blocks and a support $T$ with $m$
distinct tuples. For a  convex-separable distance, $d$, and the associated
family $\mc{C}_d^\nit{mc}$ of most-compliant classes, the following are
computational problems and their computational complexities:

\vspace{1mm}\noindent
(a) \ \textbf{MCC} (MC-Class): Compute an admissible\footnote{That is,
$\nit{adm}(\mathbf{k}^\star)$ holds, as in Remark \ref{rem:canon}; in particular
$\sum_{j=1}^{m} k_j = n$. } class-vector $\mathbf{k}^\star \in \mathbb{N}^{m}$ that minimizes the
compliance-distance (or maximizes compliance): \
$C_{\mathbf{k}^\star} \in \argmin_{\mathbf{k} \in \mathbb{N}^m,\ \sum_{j=1}^{m}
k_j = n} d^c(C_{\mathbf{k}}, \mc{M})$.
\ It holds: \ \textbf{MCC} is in \textbf{FP} (it can be computed in polynomial time).

\vspace{1mm}\noindent (b) \ \textbf{\#MCC} (Counting MC-Classes): Count the
class-vectors $\mathbf{k}^\star$ such that $C_{\mathbf{k}^\star}$ minimizes
$d^c(C_{\mathbf{k}},\mc{M})$.
\ If $d$ is strictly convex-separable, it holds: \ \textbf{\#MCC} is \textbf{\#P-complete}.

\vspace{1mm}\noindent (c) \ \textbf{MCC-Enum} (Enumerating MC-Classes):
Enumerate all class-vectors $\mathbf{k}^\star$ such that $C_{\mathbf{k}^\star}$
minimizes $d^c(C_{\mathbf{k}},\mc{M})$.
\ It holds: \ \textbf{MCC-Enum} is in \textbf{DelayP}.

\vspace{1mm}\noindent (d) \ \textbf{\#MCW} (Counting MC-Worlds): Count the worlds
$W$ minimizing $d^c(W,\mc{M})$, i.e.\ belonging to some MC-class.
\ \ If $d$ is strictly convex-separable, it holds: \ \textbf{\#MCW} is \textbf{\#P-complete}.

\vspace{1mm}\noindent (e) \ \textbf{\#Class} (Class Cardinality): Given a
class-vector $\mathbf{k}$, count the worlds of $C_{\mathbf{k}}$, i.e.\ compute
$|C_{\mathbf{k}}|$. \ \ \textbf{\#MCClass} (MC-Class Cardinality): given the BID,
compute $|C_{\mathbf{k}^\star}|$ for a most-compliant $C_{\mathbf{k}^\star}$.
\ \ If $d$ is strictly convex-separable, it holds: \ \textbf{\#Class} and \textbf{\#MCClass} are \textbf{\#P-complete}.
\boxtheorem
\end{theorem}

\ignore{+++++
\begin{theorem} \label{th:complresultsmcc} \em
(MC-related problems)
Assume $D^p(D^\star,\mc{M})$ is a BID with $n$ blocks and a support $T$ with $m$ distinct tuples.
\ For  a {\em convex-separable distance} $d$ and the associated family $\mc{C}_d^\nit{mc}$ of most-compliant classes,
the following are computational problems and their computational complexities:

\vspace{1mm}\noindent
(a) \ \textbf{MCC} (MC-Class): Compute an admissible\footnote{That is, $\nit{adm}(\mathbf{k}^\star)$ holds, as in Remark \ref{rem:canon}; in particular $\sum_{j=1}^{m} k_j = n$. From now on, we leave this condition on class-vectors implicit.} class-vector $\mathbf{k}^\star \in \mathbb{N}^{m}$  that minimizes the compliance-distance (or maximizes compliance): \
$C_{\mathbf{k}^\star} \ \in \argmin_{\mathbf{k} \in \mathbb{N}^m, \ \sum_{j=1}^{m} k_j = n} d^c(C_{\mathbf{k}}, \mc{M})$.

\vspace{1mm}
\noindent \textbf{MCC}  is in \textbf{FP}.
\ignore{++++++  More precisely, the functional problem of computing an element of the $\mc{C}_d^\nit{mc}$ class can be solved in $O(n^3 \times m^3)$ time. ++++++}

\vspace{1mm}\noindent (b) \  \textbf{\#MCC} (Counting MC-Classes): Count the number of MC-classes, i.e. all class-vectors $\mathbf{k}^\star$, such that $C_{\mathbf{k}^\star}$ minimizes $d^c(C_{\mathbf{k}},\mc{M})$.

 \vspace{1mm}\noindent
     \textbf{\#MCC} is \textbf{\#P-complete}.

\vspace{1mm}\noindent (c) \  \textbf{MCC-Enum} (Enumerating MC-Classes): Effectively enumerate all class-vectors $\mathbf{k}^\star$, such that $C_{\mathbf{k}^\star}$ minimizes $\nit{d}^c(C_{\mathbf{k}},\mc{M})$.

\vspace{1mm}\noindent \textbf{MCC-Enum} is in \textbf{DelayP}.
\ignore{++++, with polynomial delay.
\red{and polynomial space complexity $O(??)$}. +++}
\boxtheorem
\end{theorem}
+++}


\ignore{+++++
 We recall that the class-vectors mentioned in this definition (and everywhere) are subject to the {\em admissibility condition} (see Remark \ref{rem:canon}). \
The MCC problem is about computing one good class. QA on it is straightforward. Computing the probability of the class (as in Definition \ref{def:preferred} of QA) is considered in Section \ref{sec:mpcproblems}.
 Theorem \ref{th:complresultsmcc}  tells us that computing a single MC-class is tractable for a wide class of distance functions. Moreover, the (possibly exponentially many)  MC-classes can be effectively enumerated with polynomial delay. However, counting the number of MC-classes turns out to be $\#P$-complete.

These results are obtained through a parsimonious reduction from the \textbf{MCC}-problem to that of computing a ``minimum-weighted perfect matching" in bipartite graphs. Computing an MC-class corresponds to solving an optimal assignment problem where blocks are matched to tuple multiplicities, while minimizing the distance objective.

\comlb{Up to here.}

\ignore{We recall that the support  is  $T=\langle t_1,\ldots, t_m \rangle$.}

\ignore{\begin{definition} (convex-separable distances/divergences) \label{def:sep}
(a) A distance $d(C_{\mathbf{k}}, \mathcal{M})$, for classes $C_\mbf{k}$ with $\mathbf{k}=\langle k_1, \ldots, k_m\rangle$, is called \emph{separable} if it can be expressed as $d(C_{\mathbf{k}}, \mathcal{M})=
\sum_{j=1}^m d_j\!\bigl(k_j, P^\mathcal{M}(t_j)\bigr)$, where $d_j(k_j,P^\mc{M}(t_j))$ is a function that depends on the number of occurrences of tuple $t_j$ in $C_\mbf{k}$ ($t_j$'s contribution to the distance). \ (b)  $d_j(k_j, P^\mathcal{M}(t_j))$ is  {\em discretely convex} (in its first argument) if,  for all integers $k_j \ge 1$,
the discrete {\em marginal contributions} are non-decreasing: \ $d_j(k_j+1, P^\mathcal{M}(t_j)) - d_j(k_j, P^\mathcal{M}(t_j)) \geq d_j(k_j, P^\mathcal{M}(t_j)) - d_j(k_j-1, P^\mathcal{M}(t_j))$. \ $d(C_\mbf{k},\mc{M})$ is called {\em convex-separable} if, in addition, each of the $d_i$ is discretely convex.
\boxtheorem
\end{definition}
\vspace{-2mm}
The second argument of $d_j$ can be any parameter depending on $t_j$. For our application, we chose $P^\mc{M}(t_j)$.
 It is easy to check that $d^\mathrm{KL}$,  $d^\mathrm{EU}$, and the $|chi^2$-statistic are convex-separable, so as many distances used in practice. are convex-separable, so as all the distances, divergences and norms in Table~\ref{table:ex-distances} in the Appendix. Also the }


\ignore{\h The following definition formalizes  computational problems of interest related to MCC.}

\begin{definition} \label{def:mcc} \em
(MC-problems)
Assume $D^p(D^\star,\mc{M})$ is a BID.
We define the following  computational problems:

\vspace{1mm}\noindent
1. \textbf{MCC} (MC-Class): Compute an admissible\ignore{\footnote{That is, $\nit{adm}(\mathbf{k}^\star)$ holds; in particular $\sum_{j=1}^{m} k_j = n$. According to Remark \ref{rem:admiss}, from now on, we leave this condition on class-vectors implicit.}} class-vector $\mathbf{k}^\star \in \mathbb{N}^{m}$  that minimizes the compliance-distance (or maximizes compliance): \
$C_{\mathbf{k}^\star} \ \in \argmin_{\mathbf{k} \in \mathbb{N}^m, \ \sum_{j=1}^{m} k_j = n} d^c(C_{\mathbf{k}}, \mc{M})$.

\vspace{1mm}\noindent 2. \textbf{\#MCC} (Counting MC-Classes): Count the number of MC-classes, i.e. all class-vectors $\mathbf{k}^\star$, such that $C_{\mathbf{k}^\star}$ minimizes

\vspace{1mm}\noindent 3. \textbf{MCC-Enum} (Enumerating MC-Classes): Effectively enumerate all class-vectors $\mathbf{k}^\star$, such that $C_{\mathbf{k}^\star}$ minimizes $\nit{d}^c(C_{\mathbf{k}},\mc{M})$.
\boxtheorem
\end{definition}

\ignore{
\vspace{-4mm}Theorem \ref{thm:mcc} provides  computational complexities \cite{Papadimitriou1994} of MCC-related problems.}

\begin{theorem} \label{thm:mcc} \em
(MC-related problems)
Assume $D^p(D^\star,\mc{M})$ is a BID with $n$ blocks and a support $T$ with $m$ distinct tuples.
\ For  a {\em convex-separable distance} $d$,
we have:

\vspace{1mm}\noindent
1.   \textbf{MCC}  is in \textbf{FP}, the class of functional problems computable in deterministic polynomial time.

\vspace{1mm}\noindent 2.
     \textbf{\#MCC} is \textbf{\#P-complete}, i.e. it belongs and is hard for the class of counting solutions of $\mbf{NP}$ decision problems.

\vspace{1mm}\noindent 3.  \textbf{MCC-Enum} is in \textbf{DelayP}, the class of enumeration problems where the time delay between the output of any two consecutive solutions is polynomial in the input size.
\boxtheorem
\end{theorem}
++++++++++++++++++++++++++}

\vspace{-3mm}\h Strict convexity is used only for hardness in parts (b), (d) and (e). Membership holds with plain convexity. \
\textbf{MCC} is about computing one good class. Once we have it,
QA on it is straightforward. \ignore{Computing the probability of the class (as
in Definition \ref{def:preferred} of QA) is considered in Section
\ref{sec:mpcproblems}.}
Theorem \ref{th:complresultsmcc} tells us that computing a single MC-class can be
achieved in polynomial-time for a wide class of distance functions, and that the
(possibly exponentially many) MC-classes can be enumerated with polynomial delay.
\ Counting MC-classes and MC-worlds, and computing the cardinality of a class and of an MC-class are all intractable: They are
all $\#P$-complete.

\h The tractability results rest on a bijection between admissible class-vectors and
feasible integral flows of a bipartite network, under which the compliance
distance becomes the flow cost: Computing one MC-class is then a ``minimum-cost
flow" problem \cite[chap.\ 6]{ahuja1993}. Enumeration becomes a depth-first search
over class-vector prefixes, where a single min-cost flow computation prunes, at
each node, the prefixes that cannot reach the optimum.

\h The hardness
results, by contrast, are combinatorial, through reductions from counting perfect
matchings and transversal-matroid bases, the two isolating the orthogonal ways an optimal
set can be large, namely many worlds within one class versus many optimal
classes.

\h Notice that Theorem
\ref{th:complresultsmcc} holds, in particular, for $\mc{C}_{\nit{KL}}^\nit{mc}$,
$\mc{C}_{\nit{EU}}^\nit{mc}$ and $\mc{C}^\nit{mc}_{pv}$.

\section{Computing Most-Probable Classes} \label{sec:mpdbs}

\h Unlike the MC-setting,  computational problems related to  most-probable-classes, the MP-problems, have a significantly higher complexity.  Hardness  resides in the need to optimize a \#P-complete function (class-probability) over a possibly  exponentially large solution-space of class-vectors. Even evaluating the objective function at a single point requires solving a \#P-complete problem. Theorem \ref{th:complresultsmp} introduces the problems with their  complexities.

\begin{theorem}\label{th:complresultsmp} \em
  (MP-Related Problems)
Assume $D^p(D^\star,\mc{M})$ is a BID with $n$ blocks and a support $T$ with $m$ distinct tuples. The following are computational problems related to most-probable classes and their complexities, in data complexity, in the size of the BID:

\vspace{1mm}\noindent
(a) \ \textbf{Class Probability:} Given a class $C_{\mathbf{k}}$, with $\mathbf{k} \in \mathbb{N}^{m}$, compute its probability $P^{\mathcal{C}\!}(C_{\mathbf{k}})$. \ It holds: \
\textbf{Class Probability} is \textbf{\#P-complete}.

\vspace{1mm}\noindent (b) \  \textbf{MPC[D]} (Most-Probable Class-Decision Problem): Given a threshold $\theta \in [0,1]$, decide whether there exists a class-vector $\mathbf{k}$\ignore{ = (k_1, \ldots, k_m) \in \mathbb{N}^m$, such that \ \
$\sum_{i=1}^{m} k_i = n
\quad \text{and} \quad  }, such that
$P^\mathcal{C}(C_{\mathbf{k}}) \geq \theta$. \ It holds: \ \textbf{MPC[D]} is in \textbf{NP\textsuperscript{PP}}, the class of decision problems that can be solved in non-deterministic polynomial-time  with access to an oracle for $\mbf{PP}$ (probabilistic polynomial time).

\ignore{
\vspace{1mm}\noindent (c) \  \textbf{CMPC[D]} (Constrained MP-Decision Problem): Given a threshold $\theta \in [0,1]$ and an upper bound vector $\mathbf{b}=(b_1, \ldots, b_m) \in \mathbb{N}^m$, decide if there is \ignore{whether there exists} a class-vector $\mathbf{k}$, \ignore{$ = (k_1, \ldots, k_m) \in \mathbb{N}^m$,} such that, for all $j \in [1,m]$,  $k_j \leq b_j$, and $
P^\mathcal{C}(C_{\mathbf{k}}) \geq \theta
$.

\vspace{1mm}
\textbf{CMPC[D]} is \textbf{NP\textsuperscript{PP}}-complete.}

\vspace{1mm}\noindent (c) \  \textbf{CMPC[D]$^\mathcal{O}$} (Constrained MP-Decision Problem):
 Let $\mathcal{O}$ be a set of of arithmetic comparisons. Given a threshold $\theta \in [0,1]$, a bound-vector $\mathbf{b} = (b_1,\dots,b_m) \in \mathbb{N}^m$, and an operator-vector $\mathbf{op} = (op_1,\dots,op_m) \in \mathcal{O}^m$, \textbf{CMPC[D]} is about deciding if there is a class-vector $\mathbf{k} = (k_1,\dots,k_m) \in \mathbb{N}^m$, such that, for all $j \in [1,m]$,
$k_j \mathrel{op_j} b_j$,
and
$
P^{\mathcal{C}}(C_{\mathbf{k}}) \geq \theta
$. \ It holds: \
\textbf{CMPC[D]$^{\{=\}}$} is \textbf{PP}-complete; \ and \
 \textbf{CMPC[D]$^{\{=,\leq\}}$} is \textbf{NP$^{\textbf{PP}}$}-complete. \ The case $\mathcal{O}=\{=\}$, with $\mathbf{b}=(1,\ldots,1)$, captures key-constraints, with each tuple required to occur exactly once.

\vspace{1mm}\noindent
(d) \  \textbf{MPC[O]} (Most-Probable Class–Optimization Problem): Find a class-vector $\mathbf{k}^\star$, such that  $P^{\mathcal{C}\!}(C_{\mathbf{k}})$ takes a maximum value. \ It holds: \ \textbf{MPC[O]} is in \textbf{FP\textsuperscript{NP\textsuperscript{PP}}}, the class of function-computation problems that are solvable in deterministic polynomial-time with access to an $\mathsf{NP^{PP}}$\!-oracle.
\boxtheorem
\end{theorem}

\vspace{-5mm}\h
The $\#P$-completeness of \textbf{Class Probability}  poses a fundamental bottleneck: Any algorithm for MP-problems must repeatedly solve this hard  problem  during the search process.

\h The proofs, given in Appendix~\ref{proof:mpc}, rely on the algebraic characterization of class probabilities through the {\em permanent} of matrix. In particular, for a class vector $\mathbf{k}$, its probability is expressed as the permanent of a matrix obtained by duplicating columns according to $\mathbf{k}$, up to a factorial normalization. Conversely, the permanent of any stochastic matrix is encoded as the probability of the class corresponding to $\mathbf{k}=(1,\ldots,1)$. This establishes the $\#\mathsf{P}$-completeness of \textbf{Class Probability} and yields the $\mathsf{NP}^{\mathsf{PP}}$ upper-bounds for the decision variants, by guessing $\mathbf{k}$ and testing the resulting probability against the threshold.

\h
For $\textbf{CMPC[D]}^{\{=,\leq\}}$\!, \ $\mathsf{NP}^{\mathsf{PP}}$-hardness is established by a reduction from \textsc{E-MajSat}, asking  if there is an assignment to a set of existential variables, such that the majority of assignments to the remaining variables satisfy the formula.
Building on Valiant's reduction from \#SAT to \textsc{Permanent} \cite{Valiant1979}, we modify only the variable gadgets for the existential variables: A controlled column duplication encodes the existential assignment, while  relying on the remaining components of Valiant's construction,  to count satisfying assignments to the remaining variables; and hence, to test the majority condition. As a consequence, admissible class vectors correspond to existential assignments, and the permanent of the associated matrix counts, up to a fixed factor, the satisfying assignments to the remaining variables. In this way, the majority condition of \textsc{E-MajSat} is captured by the probability threshold in $\textbf{CMPC[D]}^{\{=,\leq\}}$. \
Finally, \textbf{MPC[O]} belongs to $\mathsf{FP}^{\mathsf{NP}^{\mathsf{PP}}}$ by polynomially-many threshold queries to \textbf{MPC[D]}. \
\ignore{\red{The slightly more expressive variant, \textbf{CMPC[D]$^{\{=,\leq\}}$}, that enables control over tuple multiplicities, is $\mathsf{NP^{PP}}$-complete.}} \ We can see that  Theorem
\ref{th:complresultsmp} leaves some lower-bounds open, in particular, for $\mbf{MPC[D]}$.  \ignore{The optimization problem \textbf{MPC[O]} sits higher up in this hierarchy, in $\mathsf{FP^{NP^{PP}}}$\!\!.}

\vspace{-1mm}
\section{Probability Recovery from Data and Causal MGs}\label{sec:recover}

\vspace{-2mm}\h Until now, we have relied on having a MG $\mc{M}$, represented as a Bayesian Network with its local distributions, and the observed data $D^\star$. We also made the assumption that $D^\star$ is generated according to (or is compliant with) $\mc{M}$, which could have been learned from a different dataset, but for the main application domain as that of $D^\star$. \ However, it could be the case that we only have $D^\star$ and a {\em qualitative} MG $\mc{M}$ that represents causal (in)dependencies among database attributes, without explicit local distributions.

\h Since we assume that $\mbb{I}^C =1$   iff $C^m =\na$,  the observed DB $D^\star$ can be extended with columns and contents for the indicator variables $\mbb{I}^C$: \ $\mbb{I}^C$ takes values $1$ or $0$ depending on whether $C^\star$ shows an $\na$ or not. This ``$\mbb{I}$-extended DB" is, then, $\mbb{I}$-deterministic (see Section \ref{sec:mgs}). We still denote it with $D^\star$.

\h In this new setting, we could still be in position to build the probabilistic PDB $D^p(D^\star,\mc{M})$ as in Section \ref{sec:genBIDs}. All we need are the estimates of the probabilities $p^{\nit{BID}}\!(\tau)$ in Definition \ref{def:blocks}(b), those for tuples $\tau$ in  non-singleton blocks $B$ (singleton blocks' tuples have  probability $1$). Then, the question becomes whether we can {\em recover} those probabilities from $D^\star$ by appealing to the qualitative $\mc{M}$. This relevant issue in the more general setting of statistical data is discussed more deeply in \cite[sec. 3]{mohan21}. In the following, we adapt it to our setting. It turns out that the answer may depend on the kind of MMs $\mc{M}$ represents.

\vspace{-1mm}
\begin{definition}\label{def:rec} \em Given qualitative MG $\mc{M}$, a functional $E$ from the class $\mc{D}$ of $\mbb{I}$-extended  observed
databases $D^\star$ to  $\mbb{R}$ is  {\em recoverable}  if, there is a computable functional $\hat{E}$ on $\mc{D}$ that is a {\em consistent estimate} of $E$ over all underlying positive data distributions $P(\bar{A^o},\bar{C^\star})$ that are compatible with $\mc{M}$, in the sense that they capture the same (in)dependencies. \boxtheorem
\end{definition}

\vspace{-5mm}\h We recall that, in this case, a {\em consistent estimate} $\hat{E}$ of $E$ converges in probability to $E$ as the database $D^\star$ grows. More precisely, for every $\epsilon >0$, the probability  that $|\hat{E} - E| > \epsilon$ tends to $0$, where the probability comes from the admissible distributions $P(\bar{C^\star},\bar{A^o}, \bar{\mbb{I}^C})$   \cite{larry}. To establish this, it suffices to express $E$ in terms of the {\em observed data distribution} (i.e. observed frequencies), denoted  $P^{D^\star}(\bar{C^\star},\bar{A^o}, \bar{\mbb{I}^C})$.

\h The estimand $E$ in this definition could be, among others, a marginal or conditional probability;  a numerical queries on the DB, such as an aggregation or a Boolean query.  \ For now, we are interested in estimating the probabilities $p^{\nit{BID}}\!(\tau)$ that are needed to build the BID, those in Definition \ref{def:blocks}(b). In Example \ref{ex:new2}, its is about consistently estimating the probabilities in (\ref{eq:bid-tuple}).

\begin{proposition} \label{prop:recov}\em Consider a  qualitative MG $\mc{M}$ that represents an MCAR or a MAR missingness mechanism.  For every $\tau$ in a non-singleton block of BID $D^p$, the probability $p^{\nit{BID}}\!(\tau)$  is recoverable from the observed database $D^\star$. \boxtheorem
\end{proposition}

 \vspace{-5mm}\h For MNCAR MGs the situation is much more complex. There are cases where the probabilities can be recovered, and other where they can not \cite{mohan21}, but there is no big picture yet. It is not clear that the recoverable cases can be characterized.

\section{Related Work}\label{sec:relWork}
\h Incomplete DBs, in the sense of missing values in tuples,  have been investigated for a long time, under different representations and semantics. Since the inception of relational DBs, null values, in the form of the SQL constant \Nn, have been used to represent MVs. Its use has been contentious and a subject of several papers. See \cite{libkin,lastLibkin} for a discussion and references. \ In our work, we do not refer to- or handle {\em null values} or \Nn~ as in SQL DBs.

\h Query answering on incomplete DBs is classically framed over the set of possible worlds a database represents, with \emph{certain} and \emph{possible} answers as the two reference semantics~\cite{imielinski}. \ A recent body of investigation  has shed light on the  semantic and algorithmic issues related to the use of \Nn, observing that  neither reference semantics is fully satisfactory: certain answers are often too pessimistic,  while possible answers are too permissive  \cite{libkin,lastLibkin}.

\ignore{\blue{For a related,  simple and practical reconstruction of the use of \Nn \ in SQL that is based on  query-rewriting,  see \cite[Sec. 4]{p2p}.
\  Actually, by replacing MVs by \Nn, and using an SQL-based DBMS for QA, we would be assigning a very particular ``semantics" to data with MVs; one that is embedded in the way SQL DBs {\em operates} with  \Nn. }}

\h An alternative line of work narrows the space of possible worlds by exploiting {\em a priori} knowledge about the missing values. \
More recently, and closer to our semantics, \cite{console} has considered numerical queries in the presence of multiple, {\em labeled} null values. (We use a single constant, \na, to denote a MV.) A numerical null is assigned a distribution over non-null  values. Correlations are represented by simultaneous occurrences of a same labeled null.

\h Instead, in our case, we assume {\em a priori} knowledge of different nature: The \emph{missingness mechanism} represented by an MG. It grounds our semantics
on the causal theory of missing data~\cite{rubin76,mohan21}. Here, an MG can play the role of correlating attribute values. As mentioned in Section \ref{sec:mgs}, we could further impose soft-constraints to capture other forms of correlation, in particular, inter-tuple correlations \cite{marko}.

\h
In PDBs, a query semantics defines how answers  are interpreted and computed under uncertainty. The \emph{possible worlds semantics} determines probability distributions over those worlds and  over query answers, by evaluating the query on each world \cite{suciu}. Alternatively, the \emph{confidence} or \emph{marginal semantics} focuses on marginal probabilities of individual tuples appearing in results, simplifying computations by ignoring correlations \cite{dalvi2004,soliman2007topk,re2005trio}. Extensions such as the \emph{top-k} and the \emph{expected score} semantics rank or score tuples based on probabilities or utilities \cite{soliman2007topk,sarma2006working,re07}. One can also decide to use the  \emph{most probable database}  to evaluate the query \cite{NOW}. \ For aggregate queries, similar semantics apply: the \emph{possible worlds semantics} yields a distribution over aggregate values \cite{dalvi2004}. \ignore{, \emph{marginal semantics} computes marginal probabilities for each aggregate result, and \emph{expected aggregate semantics} returns the expected aggregate value, also with group-by } \ignore{Advanced semantics also consider ranking or selecting likely aggregates via \emph{top-k} or \emph{most probable aggregate} methods \cite{soliman2007topk,soliman2008,re07}.}

\ignore{
\red{Several query semantics have been proposed for probabilistic databases, each providing a different way to interpret and compute query results: the \textit{possible worlds semantics} models the database as a distribution over deterministic worlds and derives a probability distribution over query answers by evaluating each world \cite{dalvi2004}, \emph{confidence} or \emph{marginal semantics} focuses on marginal probabilities of individual tuples appearing in results, simplifying computations by ignoring correlations \cite{dalvi2004,soliman2007topk,re2005trio}. Extensions such as \emph{top-k} and \emph{expected score} semantics rank or score tuples based on probabilities or utilities \cite{soliman2007topk,sarma2006working,re07}, while \emph{most probable world} semantics selects the single most likely deterministic instance satisfying the query \cite{mostProbDB}.
} }

\h In  \cite{benny}, the authors take a probabilistic and general approach to  data quality, which may in principle involve different dimensions of quality, including incompleteness. MVs, as those we deal with, are not specifically  addressed.

\h There is a large body of research on dealing with MVs, mostly in Statistics \cite{gelman}. A common technique is {\em imputation}, which amounts to filling in for them using  values from the domain.
\ MMs were proposed in that context \cite{rubinBook}.
Imputation methods and statistical estimates in the presence of MV come in different forms, and may   depend on MMs \cite{rubinBook,gelman}. \  Our work is not about these forms of  ``classic imputation", but, instead, we do something like an {\em implicit and multiple, BN-informed, probabilistic imputation} that  gives rise to several possible ``imputed" instances with attached probabilities.
\ignore{Various methodologies have been proposed to assess the compatibility of a dataset with a Bayesian Network,  such as {\em model-fitness measures}, e.g.  maximum-likelihood \cite{larry,allisonML}\ignore{\cite{Bishop2006,allisonML,larry}} and  {\em information-theoretic measures} \cite{MacKay2003}. }

\h MMs represented as MGs that take the form of BNs have been introduced and investigated in \cite{mohan13,broeck2015,mohanthesis2017,mohan21}, mostly concentrating on learning BNs under MVs,  {\em recovering} the right probabilities from data with MVs by means of qualitative MGs \cite{mohan21}; and doing inference from the  resulting BNs. Our applications to DBs via PDBs, and QA under them, and related DBs, such as the most-compliant ones,  are new.

\vspace{-3mm}
\section{Conclusions}
\label{sec:conclusion}
\vspace{-2mm}
\h We have provided a principled semantics to a DB  with MVs, whose occurrences are  governed by a quantitative Bayesian Network,  that represents missingness mechanisms.
 The data semantics relies on  the construction of a BID that induces a space of {\em possible worlds} with associated probabilities. Possible worlds are multiset-DBs without  MVs.
They  are classified into  classes of identical members, as multisets.

\h We introduced and investigated two particular collections of classes: The \emph{Most Compliant Classes} -leading to the MCC-semantics- whose members best align with the underlying MG; and the \emph{Most Probable Classes} -leading to he MPC-semantics- whose members are the most likely ones. Query answering (QA) can be done on  any of these classes.  These two semantics reflect the complementary dimensions of statistical plausibility and probabilistic likelihood.

\h
We presented  complexity results for the MCC-semantics. In addition to some intractability results, we notably showed that it is possible to efficiently compute a single MCC-class, on which QA can be performed.  Although enumerating the answers under the MCC-semantics is computationally intractable, it can still be performed with polynomial delay.
\ For the MPC-semantics we unveiled several intractability results, in data complexity.

\h We also studied the problem of using the observed data to {\em recover} the probabilities that could be used as parameters for a qualitative MG. For some classes of MGs this is possible, and doable in polynomial-time in data.

\ignore{
\red{As part of our ongoing work, not reported here, and with the aim of identifying classes of observed DBs for which efficient QA is tractable, we have uncovered
 a natural class of them for which the
MCC- and the MPC-semantics coincide, and QA becomes tractable.} }

\h
We are currently  investigating a hybrid QA semantics that jointly considers compliance and probability, allowing for a more expressive and flexible notion of QA. We have been able to efficiently compute most-probable worlds among the most-compliant ones; and also Pareto-optimal worlds that combine the two criteria.
\ignore{We are also interested in exploring how our class-based semantics could inform or be integrated with advanced imputation strategies, especially in statistical or machine learning pipelines.}
\ignore{We are investigating the  verification and enforcement of  compliance of an observed DB $D^\star$ with the underlying MG.} 
We have  implemented solutions for the just mentioned tractable problems. We have obtained extremely encouraging experimental results with real-world datasets \cite{we}.

\ignore{these techniques and conducting experimental evaluations on real-world datasets will be essential for assessing their practical applicability and performance.}


\bibliographystyle{plain}







\appendix

\newpage
\appendix

\ignore{++
\section{Appendix: Additional Definitions and PDBs}\label{sec:basic}

\begin{definition} \em
\label{def:blocks}
Consider an MG $\mc{M}$ and  an observed instance $D^\star$ for schema $S^\star$, and $R^\star$ a relation in $D^\star$ with schema $R(\bar{A}^o,\bar{A}^\star)$, where $\bar{A}^o, \bar{A}^\star$ are lists of fully observed and possibly taking \na \ attributes, resp.  Let $\tau$ be a tuple in $R^\star$, and $\tau[\bar{A}^o,\bar{A^\prime}^\star]$ its restriction to those attributes without an {\footnotesize \na}, with $\bar{A^\prime}^\star \subseteq \bar{A}^\star$.

\vspace{1mm}\noindent
(a) The {\em block} associated to $\tau$ is the set of tuples for schema $R(\bar{A}^o,\bar{A}^m)$:
\begin{eqnarray}
\mc{B}(\tau) &:=& \{ \ \tau^\prime~|~\tau^\prime[\bar{A}^o,\bar{A^\prime}^\star] = \tau[\bar{A}^o,\bar{A^\prime}^\star], \mbox{ and}, \nonumber \\
&&~~~~~~~\mbox{ for each } A \in (\bar{A}^\star \smallsetminus  \bar{A^\prime}^\star),\nonumber \\  &&~~~~~~~~\tau^\prime[A] \in \nit{dom}(A^m) \ \}. \hspace{-5mm}\label{eq:block}
\end{eqnarray}
\noindent (b)
For $\tau^\prime \in \mc{B}(\tau)$, its  probability $p(\tau^\prime)$ is the conditional probability:
\begin{equation}
p^{\nit{BID}}\!(\tau^\prime) := P^\mc{M}(\tau^\prime[\bar{A}^\star \smallsetminus  \bar{A^\prime}^\star]~|~\tau).\label{eq:probTuple}
\end{equation}
\h The probability of the tuple in a singleton block  is $1$.

\vspace{1mm}\noindent (c)
$D^p\!(D^\star,\mc{M})$ denotes the BID whose relations $R^p$ contain the blocks $\mc{B}(\tau)$ for $\tau \in R^\star$, and each tuple $\tau^\prime \in R^p$ has probability $p^{\nit{BID}}(\tau^\prime)$.

\vspace{1mm}\noindent (d)
A {\em possible world} associated to $D^\star$ is an instance $W$ for the underlying  schema $S$, with relations $R^W$ that contain, for each $\tau \in R^\star$, only one $\tau^\prime \in \mc{B}(\tau)$, and nothing more. \ $\mc{W}(D^\star)$ denotes the set of possible worlds. \boxtheorem
\end{definition}
++}

\section{Appendix: Additional Material for Section \ref{sec:genBIDs}.}\label{app:example}

\begin{example} \label{app:details}(Details for  Example \ref{ex:connection})

\begin{eqnarray*}
p_2^1 &:=& 
P(C^m=c_1~|~A^o=a_2,B^o=1,C^\star=\na)  =  P(C^m=c_1~|~B^o=1,C^\star=\na) \nonumber \\
&=& \frac{P(B^o=1,C^\star=\na,C^m=c_1)}{P(B^o=1,C^\star=\na)} 
= \frac{\sum_{A^o,\mbb{I}^C}P(A^o,B^o=1,C^\star=\na,\mbb{I}^C,C^m=c_1)}{\sum_{A^o,\mbb{I}^C,C^m}P(A^o,B^o=1,C^\star=\na,\mbb{I}^C,C^m)}. \label{eq:frac}
\end{eqnarray*}
For the numerator, we have:

\vspace{2mm}

$P(C^m=c_1) \times P(B^o=1) \times \sum_{A^o,\mbb{I}^C}[P(C^\star=\na~|~C^m=c_1,\mbb{I}^C) \times P(\mbb{I}^C~|~B^o=1) \times  P(A^o)]$ \
$= P(C^m=c_1) \times P(B^o=1) \times \sum_{A^o}[P(C^\star=\na~|~C^m=c_1,\mbb{I}^C=1) \times P(\mbb{I}^C=1~|~B^o=1) \times  P(A^o)]$ \ 
$=
P(C^m=c_1) \times P(B^o=1) \times P(C^\star=\na~|~C^m=c_1,\mbb{I}^C=1) \times P(\mbb{I}^C=1~|~B^o=1)$.

\vspace{2mm}
Similarly, for the denominator:

\vspace{2mm}
$P(B^o=1) \times \sum_{A^o,\mbb{I}^C,C^m}[P(C^\star=\na~|~C^m,\mbb{I}^C) \times P(\mbb{I}^C~|~B^o=1) \times P(C^m) \times P(A^o)]=$ \
$P(B^o=1) \times \sum_{A^o,C^m}[P(C^\star=\na~|~C^m,\mbb{I}^C=1) \times P(\mbb{I}^C=1~|~B^o=1) \times P(C^m) \times P(A^o)]$.

\vspace{1mm}\h By conditions (C) before Example \ref{ex:new},  \ $p_2^1$ becomes:

\vspace{1mm}
\centerline{$
  \frac{P(C^m=c_1)}{ \sum_{A^o,C^m}[  P(C^m) \times P(A^o)]}.$}

  \vspace{1mm}
\h Since $C^m \independent A^o$, the denominator is the sum of the joint distribution of both, which is $1$. Finally: \ $p_2^1=P(C^m=c_1)$. \boxtheorem
\end{example}

\begin{example}\label{ex:2mgs} Consider the MGs  in Fig. \ref{fig:extraMGs}. \ Table \ref{tab:newdata} shows the observed DB and the blocks it gives rise to. 

\vspace{-5mm}
\begin{figure}[h]
\begin{center}
\includegraphics[width=2cm]{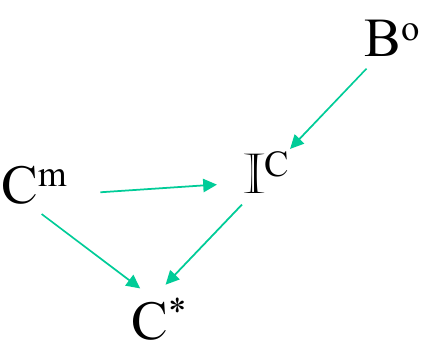}~~~~~~~~~~~~~~~~~~\includegraphics[width=2cm]{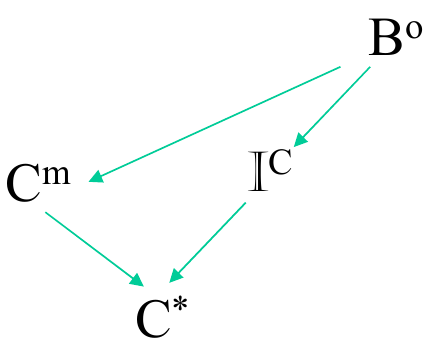}
\caption{MGs (a) and (b).}\label{fig:extraMGs}
\end{center}
\end{figure}

\vspace{-1.3cm}
\begin{center}
\captionof{table}{Observed DB and the blocks. }\label{tab:newdata}
\vspace{-2mm}\scalebox{0.6}{
$\begin{tabu}{c|c|c|}\hline
R^\star & B^o & C^\star\\ \hline
\tau_1& 0& c_1\\
\tau_2& 1& \na\\
\tau_3& 0&c_3\\
\tau_4&0&c_4\\
\tau_5&1&c_5\\
\tau_6&1& \na\\ \hhline{~--}
\end{tabu}$~~~~~~~~~~~~~~~~~~~~~~~~~~$\begin{tabu}{c|c|c||c|}\hline
R^p & B^o & C^m& P\\ \hline
\tau_1& 0& c_1&1\\
\hline
& 1& c_1&\\
& 1& c_2&\\
B(\tau_2)& 1& c_3&?\\
& 1& c_4&\\
& 1& c_5&\\
\hline
\cdots&\cdots&\cdots& 
\\ \hhline{~---}
\end{tabu}$}
\end{center}

\ For both MGs,  we need to compute:\begin{equation}
p^1_2 := P^\mc{M}(C^m=c_1|B^o=1,\mbb{I}^C=1,C^\star=\na).
\end{equation}

MG (a) corresponds to a MNCAR case, whereas (b) to a MAR case. In both cases, we obtain: $p^1_2 = \frac{P^\mc{M}(\mbb{I}^C=1,C^m=c_1,B^o=1)}{P^\mc{M}(\mbb{I}^C=1,B^o=1)}$. \ In case (b), since  $\mbb{I}^C$ and $C^m$ are independent given $B^o$, this can be further simplified into \ $p^1_2=\frac{P(C^m=c_1,B^o=1)}{P(B^o=1)}$. \ In fact, 
\begin{eqnarray}
p^1_2 &=& \frac{P(B^o=1,C^\star=\na,\mbb{I}^C=1,C^m=c_1)}{\sum_{C^m}P(B^o=1,C^\star=\na,\mbb{I}^C=1,C^m)}. \label{eq:frac2}
\end{eqnarray}

 For MG (a),
the numerator in (\ref{eq:frac2}) becomes:  

\vspace{1mm}
\begin{eqnarray*}
&&P(C^\star=\na|C^m=c_1,\mbb{I}^C=1) \times P(\mbb{I}^C=1|C^m=c_1,B^o=1) \times P(C^m=c_1) \times P(B^o=1)\\&& = P(\mbb{I}^C=1|C^m=c_1,B^o=1) \times P(C^m=c_1) \times P(B^o=1)\\&& = P(\mbb{I}^C=1,C^m=c_1,B^o=1).
\end{eqnarray*}
The denominator becomes: 
\begin{eqnarray*}
&&\sum\limits_{c \in \nit{dom}(C^m)} P(C^\star=\na|C^m=c,\mbb{I}^C=1) \times P(\mbb{I}^C=1|C^m=c,B^o=1) \times P(C^m=c) \times P(B^o=1)\\&& = \sum\limits_{c \in \nit{dom}(C^m)}P(\mbb{I}^C=1|C^m=c,B^o=1) \times P(C^m=c) \times P(B^o=1)\\&& =\sum\limits_{c \in \nit{dom}(C^m)}  P(\mbb{I}^C=1,C^m=c_1,B^o=1) \ = \ P(\mbb{I}^C=1,B^o=1).
\end{eqnarray*}
Then, \ $p^1_2 = \frac{P(\mbb{I}^C=1,C^m=c_1,B^o=1)}{P(\mbb{I}^C=1,B^o=1)}$.

\ignore{\frac{P(\mbb{I}^C=1|C^m=c_1,B^o=1) \times P(C^m=c_1) \times P(B^o=1)}{\sum_{c \in \nit{dom}(C^m)}P(\mbb{I}^C=1|C^m=c,B^o=1) \times P(C^m=c) \times P(B^o=1)}$ }

\vspace{1mm}
\h For MG (b),
the numerator becomes: 
\begin{eqnarray*}
&&P(C^\star=\na|C^m=c_1,\mbb{I}^C=1) \times P(C^m=c_1|B^o=1) \times P(\mbb{I}^C=1|B^o=1) \times P(B^o =1)\\
&&= P(C^m=c_1|B^o=1) \times P(\mbb{I}^C=1|B^o=1) \times P(B^o =1)\\&& = P(C^m=c_1,\mbb{I}^C=1|B^o=1) \times P(B^o =1) = P(C^m=c_1,\mbb{I}^C=1,B^o=1).
\end{eqnarray*}

\h The denominator becomes: 

\vspace{1mm}
\noindent
$\sum\limits_{c \in \nit{dom}(C^m)} P(C^m=c|B^o=1) \times P(\mbb{I}^C=1|B^o=1) \times P(B^o =1) = P(\mbb{I}^C=1,B^o=1)$.

Then, as for the MG (a), we obtain: \ $p^1_2 = \frac{P(\mbb{I}^C=1,C^m=c_1,B^o=1)}{P(\mbb{I}^C=1,B^o=1)}$.

\h In case (b), $\mbb{I}^C$ and $C^m$ are independent given $B^o$. Then, 
\begin{eqnarray*}
p^1_2 &=&\frac{P(\mbb{I}^C=1,C^m=c_1,B^o=1)}{P(\mbb{I}^C=1,B^o=1)} \ = \ \frac{P(\mbb{I}^C=1,C^m=c_1 | B^o=1)P(B^o=1)}{P(\mbb{I}^C=1|B^o=1)P(B^o=1)}\\
&= & \frac{P(\mbb{I}^C=1 | B^o=1)P(C^m=c_1 | B^o=1)P(B^o=1)}{P(\mbb{I}^C=1|B^o=1)P(B^o=1)}  \mbox{ \footnotesize{ \ \ (because  $\mbb{I}^C \independent C^m | B^o$)}}
\\
&=&P(C^m=c_1 |B^o=1)=\frac{P(C^m=c_1,B^o=1)}{P(B^o=1)}. \hspace{4.8cm} \blacksquare
\end{eqnarray*}
\end{example}

\section{Appendix: About Theorem \ref{thm:sharp}.}\label{sec:theo1}

For the notions used in the next result, we refer to Example \ref{ex:new}. 

\begin{lemma}\label{lemma:newOne} \em For every fixed underlying schema $\mc{S}$, and  TID $D^p$ for $\mc{S}$, one can efficiently compute (in the size of $D^p$) an  MG $\mc{M}$  whose variables are the attributes of an expanded schema $\mc{S}^\nit{ex}$, and  an associated  observed instance $D^\star$ for the observed schema $S^\star$, such that, for every BCQ $\mc{Q}$ for schema $\mc{S}$, it is possible to efficiently build a BCQ $\mc{Q}^\prime$ for the same schema, such that $\mc{Q}[D^p] = \mc{Q}^\prime[D^p(D^\star,\mc{M})]$.\boxtheorem
\end{lemma}

\vspace{-3mm}\h This result tells us that arbitrary TIDs can be obtained as special cases of the kind of BIDs we introduced. In fact, Lemma \ref{lemma:newOne} makes a more general claim than Theorem \ref{thm:sharp}:\ignore{Proposition \ref{lemma:BIDrepresentation}}\emph{Any BID can be obtained from a corresponding observed $D^\star$ and a MAR case of MG $\mc{M}$}. \ We now show an example that illustrates the proof of Lemma \ref{lemma:newOne}.

\begin{example} \label{ex:reduc}  Consider schema $\mc{S} = \{R(A,B), \ S(B)\}$ for a the TID in Table \ref{tab:reduc}(a). The schema for the associated observed DB with MVs is $\mc{S}^\star = \{R(A^o,B^o,M_1^\star), \ S(B^o,M_2^\star)\}$. \ Attributes $M_1^\star, M_2^\star$ may exhbit MVs, and $\nit{dom}(M_1^m) = \nit{dom}(M_2^m) = \{0,1\}$.

\h Table \ref{tab:reduc}(a) shows the initial TID that we want to represent as a BID, which we will obtain by first creating the observed instance in Table \ref{tab:reduc}(b). We concentrate on table $R$. 

\vspace{-2mm}
\begin{table}[h]\center
{\footnotesize $\begin{tabu}{c|c|c||c|}\hline
R&A & B& P\\ \hline
\tau_1 & a & b&p_1\\
\tau_2 & a^\prime & b^\prime &p_2\\ \hhline{~---}
\end{tabu}$~~~~~$\begin{tabu}{c|c||c|}
\hline
S&A&P\\ \hline
\tau_3 & a&p_3\\
\tau_4 & b&p_4\\
\tau_5 & b^\prime&p_5\\ \hhline{~--}
\end{tabu}$~~~~~~~~~~~~$\begin{tabu}{c|c|c|c|}\hline
R^\star&A^o & B^o&M_1^\star\\ \hline
\tau_1 & a & b &\na\\
\tau_2 & a^\prime & b^\prime&\na\\ \hhline{~---}
\end{tabu}$
~~~~~$\begin{tabu}{c|c|c|}
\hline
S^\star&A^o&M_2^\star\\ \hline
\tau_3 & a&\na\\
\tau_4 & b&\na\\
\tau_5 & b^\prime&\na\\ \hhline{~--}
\end{tabu}$}\vspace{1mm}
\caption{\ \ \ \ (a) Initial TID. \hspace{2cm} (b) Associated observed DB. } \label{tab:reduc}
\end{table}

\vspace{-5mm}
\h In order to obtain the BID, which should be the one in Figure \ref{fig:result}(b), we use, for $R^\star$, the MG in Figure \ref{fig:result}(a) (showing only the attributes for $R^\star$).

\h We need to define appropriate distributions in the MG, in such a way that we obtain column $P$ in Figure \ref{fig:result}(b). \ For example, for the first probability in that column, it should be: \ $p_1 = p(\tau_1^1) := P(M_1^m = 1| A = a, \ B=b, M_1^\star = \na)$.

\begin{figure}[h]
\begin{multicols}{2}
\centerline{\includegraphics[width=2.5cm]{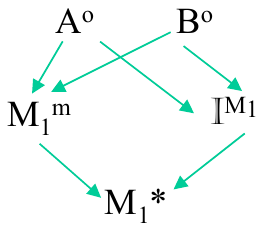}}

{\scriptsize   $\begin{tabu}{c|c|c|c||c|}\hline
R^p&A & B& M_1&P\\ \hline
\tau_1^1 & a & b &1 &p_1\\
\tau_1^2 & a & b&0&(1-p_1)\\ \hhline{~----}
\tau_2^1 & a^\prime & b^\prime &1&p_2\\
\tau_2^2 & a^\prime & b^\prime &0&(1-p_2)\\
\hhline{~----}
\end{tabu}$}
\end{multicols}
\vspace{-5mm}
\caption{\ \ \ (a) Missingness graph. \hspace{2cm} (b) Resulting BID.}\label{fig:result}
\end{figure}

\vspace{-5mm}\h For this MG, the probabilities are calculated as in Example \ref{app:details}, obtaining, for example, for the tuples in the first block of $R^p$:
{\footnotesize \begin{eqnarray*}
p(\tau_1^1) &=& P^\mc{M}(M_1^m=1|A=a, B=b) \times P^\mc{M}(A=a,B=b),\\
p(\tau_2^1) &=& P^\mc{M}(M_1^m=0|A=a, B=b) \times P^\mc{M}(A=a,B=b).
\end{eqnarray*}}

\h Accordingly, we define the probabilities in the MG in such a way that $p(\tau_1^1) = p_1, \ p(\tau_2^1) = (1-p_1)$, etc.

\h Now, consider the query posed to the original TID: \ $\mc{Q}\!: \ \exists x \exists y(R(x,y) \wedge S(y))$. In order to pose the query to the BID, rewrite it into: \ $\mc{Q}^\prime\!: \ \exists x \exists y(R(x,y,1) \wedge S(y,1))$, and we answer it via the BID.
\boxtheorem
    \end{example}

\h Theorem \ref{thm:sharp} follows from Lemma \ref{lemma:newOne}  and the $\#P$-hardness of BCQ evaluation on TIDs \cite{suciu}. It is good enough to pose to the TID a self-join-free BCQ $\mc{Q}$ that is non-hierarchical. We can see that the resulting query $\mc{Q}^\prime$ to be posed to the BID is also non-hierarchical. 

\h A full proof of Lemma \ref{lemma:newOne} can be obtained as a particular case of the proof of Proposition \ref{lemma:BIDrepresentation} below. The latter shows that every BID (not only a TID) can be represented as the BID associated to an observed DB and a  MAR MG.

\h 
For the proof of Proposition \ref{lemma:BIDrepresentation}
below,  we need some preliminary technical notions and simple results about them. We start with the definition of {\em equivalence between a pair of BIDs}, a condition that guarantees that their sets of possible worlds, along with the associated probabilities, coincide.

 \begin{definition} \em (equivalence of BIDs)
\label{def:bid-equiv}
Let $D$ and $D'$ be two BID over the same schema, $\mc{B}$ and $\mc{B}'$ their respective sets of blocks, and  $P^D, P^{D^\prime}$, their respective distributions over the blocks. \ $D$ and $D'$ are \emph{equivalent}, denoted $D \equiv D'$,
if there are bijections $g : D \to D'$, and 
$f : \mc{B} \to \mc{B}^\prime$,  
such that, for every $B \in \mc{B}$:
\begin{enumerate}
    \item[(a)] The restriction of $g$ to $B$ is a bijection from $B$ onto $f(B)$.
    \item[(b)] For every tuple $\tau \in B$, \ 
    $\tau\!\!\downarrow = g(\tau)\!\!\downarrow$, that is, the preimage and the image coincide modulo the tids.
    \item[(c)] For every tuple $\tau \in B$, \ $
    P^D_B(\tau) = P^{D^\prime}_{f(B)}\bigl(g(\tau)\bigr)$,
    where $P^D_B(\tau)$ denotes the probability of $\tau$ in block $B$. \ (Remember that $\tau$ may appear in more than one block, with different probabilities.)
\end{enumerate}
To make the bijections explicit, we sometimes write: \ $
D \stackrel{\tiny g,f}{\equiv} D'$.
\boxtheorem
\end{definition}


As a direct consequence of Definition~\ref{def:bid-equiv}, equivalence of BIDs preserves their possible-world semantics.

\begin{corollary}
\label{cor:equiv-bid} \em
Let $D$ and $D'$ be BIDs with $\mc{W}(D)$ and $\mc{W}(D')$ their respective sets of possible worlds. \ If 
$D \stackrel{{\tiny g,f}}{\equiv} D'$, there exists a bijection $
h : \mc{W}(D) \to \mc{W}(D')$, 
such that, for every $W \in \mc{W}(D)$ and  $\tau \in D$: \ (a) $\tau \in W$ if and only if $g(\tau) \in h(W)$; and 
(b) $P^D(W) = P^{D'}\!\bigl(h(W)\bigr)$.\boxtheorem
\end{corollary}

\vspace{-4mm}
\h Accordingly, if $D \equiv D'$, the two instances are indistinguishable w.r.t. QA, as every query returns exactly the same answers on both, and with the same probabilities.

\h Using the notion of equivalence, we can now show that, for every BID, there is an observed instance, and a MG whose associated BID is equivalent to it.  The proof uses a MAR case of MG. This result can be seen as a {\em representation lemma}. For the notions and notation involved, see Section \ref{sec:genBIDs}, in particular, Definition \ref{def:blocks}.

\begin{proposition} \em
\label{lemma:BIDrepresentation}
For every BID $D$ over a schema $S$, there are an observed instance $D^\star$ for a schema $S^\star$, and a MAR MG $\mathcal{M}$, such that, $D \equiv D^p(D^\star,\mc{M})$. $D^\star$ can be constructed in polynomial-time and size in the size of $D$. \boxtheorem
\end{proposition}

\vspace{-5mm}
\noindent {\bf Proof}: \
Assume $D$ consists of $n$ blocks $B_1, \cdots, B_n$, each block $B_i$ containing a set of tuples $t_{i1}, \cdots, t_{1{m_i}}$.
W.l.o.g., we assume that the schema $S$ consists of a single attribute $\mathtt{T}$ that encodes each tuple. 
\ Formally, $S = \{\mathtt{T}\}$, and for every $t_{ij} \in B_i$, $t_{ij}[\mathtt{T}] := \mathtt{``t_{ij}\!\!"}$ (used a value, a constant). With $p_i^j$ we denote the probability of tuple $t_j$ in block $B_i$. 
\ By the definition of a BID (as used in our work): \ $
\sum\limits_{j=1}^{m_i} p_i^j = 1 \quad \text{for each } i \in \{1, \cdots, n\}
$. \ Figure \ref{fig:arbitrarydBID} shows BID $D$. 

\vspace{-3mm}

\begin{figure}[h!]
    \centering
{\scriptsize $\begin{array}{c|c|c|c|}
\hline
\textbf{blockId} & \textbf{tupleId} &  \mathtt{T}  & P   \\ \hline
 B_1 & \tau_1^1  &  t_{11} &   p_1^1\\ \cline{3-4}
  &  \cdots  &  \cdots & \cdots  \\ \cline{3-4}
& \tau_1^{m_1}  &  t_{1{m_1}} &   p_1^{m_1}\\ \hline
  &  \cdots  &  \cdots & \cdots  \\ \hline
 B_n & \tau_n^1  &  t_{n1} &  p_n^1 \\ \cline{3-4}
     &  \cdots  &  \cdots & \cdots  \\ \cline{3-4}
& \tau_n^{m_n}  &  t_{n{m_n}} &   p_n^{m_n}\\ \hline
\end{array}$}
\caption{A BID $D$}
    \label{fig:arbitrarydBID}
\end{figure}

\vspace{-4mm}
 \h We now build  an observed  DB with MVs $D^\star$, a MG $\mc{M}$, and a joint distribution $P^{\mc{M}}$. From this construction, we obtain a BID $D^p$, for which $D^p \equiv D$ holds.

 \h
$D^\star$ is an instance for  schema $\{\mathtt{T}^\star, A^\star_1, \cdots, A^\star_n\}$. Attribute $\mathtt{T}$ has domain $\nit{dom}(\mathtt{T}) = \{t_1, \cdots, t_n\}$, while each attribute $A_i$, $i \in \{1, \cdots, n\}$, has domain $\nit{dom}(A_i) := \{0, 1\}$. Each attribute $A_i$ encodes the block $B_i$, with values $0$ or $1$ depending on whether the tuple belongs to the block or not.
\ For example, for a tuple $t$ in $D^p$ with $t[T] = ``t"$, $t[A_1] = t[A_3] = 1$ tell us  that $t$ belongs to blocks $B_1$ and $B_3$. \ See Figure~\ref{fig:proof-dstar} for an illustration of $D^\star$.

 \begin{figure}[h!]
    \centering
{\scriptsize  $\begin{array}{c|c|c|c|c|}
\hline
\mathbf{blockId} &  \mathtt{T}^\star & A^\star_1 & \cdots & A^\star_n  \\ \hline
 \tau_{B_1}  &  \na & 1 & \na & \na  \\ \cline{2-5}
  \cdots  &  \cdots & \cdots & \cdots & \cdots  \\ \cline{2-5}
   \tau_{B_n}  &  \na & \na & \na & 1  \\ \cline{2-5}
\end{array}$}
\vspace{2mm}\caption{Observed DB $D^\star$ with MVs}
    \label{fig:proof-dstar}\vspace{-0.8cm}
\end{figure}

\h
$D^\star$ contains $n$ tuples $\tau_{B_i}$, each corresponding to a block $B_i$ in  BID $D$. Each tuple $\tau_{B_i}$ has value $1$  for the attribute $A_i$, and value $\na$ for all the other  attributes.  That is, $\tau_{B_i}[A_i] = 1$, and $\tau_{B_i}[\mathtt{T}] = \tau_{B_i}[A_j] = \na$, for $j \neq i$.

\h The MG $\mathcal{M}$, shown in Figure~\ref{fig:proof-mg}, encodes the dependency structure between attributes in the incomplete database. \ $\mathcal{M}$ specifies that the value of  attribute $\mathtt{T}$ depends on the values of the attributes ${A^\star_i}$.

\vspace{-5mm}
\begin{figure}[h!]
    \centering
    \includegraphics[width=5cm]{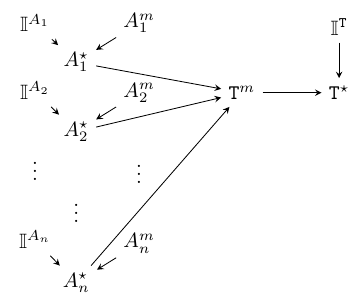}
   \caption{ \ The MG $\mathcal{M}$} \ 
    \label{fig:proof-mg}\vspace{-0.6cm}
\end{figure}







%

\ignore{++++
\begin{figure}[h]
    \centering
    \begin{tikzpicture}[->, >=stealth, node distance=0.6cm and 1cm]
        \node (IB1) {$\mbb{I}^{A_1}$};
        \node (IB2) [below=of IB1] {$\mbb{I}^{A_2}$};
        \node (IBdots) [below=of IB2] {$\vdots$};
        \node (IBn) [below=of IBdots] {$\mbb{I}^{A_n}$};

        \node (B1) [right=of IB1] {$A^m_1$};
        \node (B2) [below=of B1] {$A^m_2$};
        \node (Bdots) [below=of B2] {$\vdots$};
        \node (Bn) [below=of Bdots] {$A^m_n$};

        \node (B1star) [below=of IB1,xshift=0.7cm,yshift=0.5cm] {$A^\star_1$};
        \node (B2star) [below=of B1star] {$A^\star_2$};
        \node (Bstardots) [below=of B2star] {$\vdots$};
        \node (Bnstar) [below=of Bstardots] {$A^\star_n$};

        \node (T) [right=of B2] {\texttt{$\mathtt{T}^m$}};

        \node (TStar) [right=of T] {\texttt{T}$^\star$};
        \node (IT) [above=of TStar] {$\mbb{I}^{\texttt{T}}$};

        \draw[->] (IB1) -- (B1star);
        \draw[->] (IB2) -- (B2star);
        \draw[->] (IBn) -- (Bnstar);

        \draw[->] (B1) -- (B1star);
        \draw[->] (B2) -- (B2star);
        \draw[->] (Bn) -- (Bnstar);

        \draw[->] (B1star) -- (T);
        \draw[->] (B2star) -- (T);
        \draw[->] (Bnstar) -- (T);

        \draw[->] (T) -- (TStar);

        \draw[->] (IT) -- (TStar);

    \end{tikzpicture}
    \caption{\red{The MG $\mathcal{M}$}}
    \label{fig:proof-mg}\vspace{-0.9cm}
\end{figure}
++++++++++++++++++++++++++}

\h The joint distribution $P^{\mathcal{M}}$, shown in Figure~\ref{fig:proof-dist}, contains one row $J^j_{B_i}$ for each tuple $t_j$ in block $B_i$. Each row is defined as follows: \ (a) $J^j_{B_i}[T^m] = t_j$. \ (b) 
$J^j_{B_i}[A_l^m] := \text{in}(t_j, B_l)$, with $\text{in}(t_j, B_l) := 0$ \ if \ $t_j \in B_l$, and $0$, otherwise. \ (c) $J^j_{B_i}[A^\star_l] = 1$ \ if \ $t_j \in B_l$, and $\na$, otherwise; and $J^j_{B_i}[\mbb{I}^{A_l}] = 0$ \ if \ $t_j \in B_l$, and $1$, otherwise. 

\h The probability associated with each row in $D^p$ is: \ $
\Pr(J^j_{B_i}) = \frac{p^j_i}{n}$. \ 
These probabilities are normalized, ensuring that
\ $\sum_{i=1}^n \sum_{j=1}^{|B_i|} P^{\mathcal{M}}(J^j_{B_i}) = 1$.

\vspace{-2mm}
\begin{figure}[h!]
    \centering
{\scriptsize $\begin{array}{c|c|c|c|c|c|c|c|c|c|c|c|c|}
\hline
  \mathbf{tupleId} & \mathtt{T}^m & A_1^m & \cdots & A_n^m& A_1^\star & \mbb{I}^{A_1}&\cdots & A_n^\star& \mbb{I}^{A_n}& \mathtt{T}^\star &  \mbb{I}^{\mathtt{T}} &  \nit{Prob} \\ \hline
J^1_{B_1} & t_1 & in(t_1,B_1) & \cdots &   in(t_1,B_n) & 1 & 0 & \cdots & \na & 1 & \na & 1 & \frac{p^1_1}{n} \\ \cline{2-13}
\cdots & \cdots & \cdots & \cdots & \cdots & \cdots & \cdots & \cdots & \cdots & \cdots & \cdots & \cdots & \cdots\\ \cline{2-13}
J^n_{B_1} & t_n & in(t_n,B_1) & \cdots &   in(t_n,B_n) & 1 & 0 & \cdots & \na & 1 & \na & 1 & \frac{p^n_1}{n} \\ \cline{2-13}
\cdots & \cdots & \cdots & \cdots & \cdots & \cdots & \cdots & \cdots & \cdots & \cdots & \cdots & \cdots & \cdots\\ \cline{2-13}
J^n_{B_n} & t_n & in(t_n, B_n) & \cdots &   in(t_n, B_n) & \na & 1 & \cdots & 1 & 0 & \na & 1 & \frac{p^n_n}{n} \\ \cline{2-13}
\end{array}$}
\vspace{2mm}\caption{Joint distribution $P^{\mc{M}}$}
    \label{fig:proof-dist}\vspace{-8mm}
\end{figure}

\ignore{++
\[
\begin{aligned}
& J^j_{B_i}[T^m] = t_j, \\
& J^j_{B_i}[A_l^m] = \text{in}(t_j, B_l), \quad \text{where } \\ &~~~~~~~~~~~~~~~~~~~~~~~\text{in}(t_j, B_l) =
\begin{cases}
0 & \text{if } t_j \in B_l, \\
1 & \text{otherwise},
\end{cases} \\
& J^j_{B_i}[A^*_l] =
\begin{cases}
1 & \text{if } t_j \in B_l, \\
\text{na} & \text{otherwise},
\end{cases}\\
&
J^j_{B_i}[I^{A_l}] =
\begin{cases}
0 & \text{if } t_j \in B_l, \\
1 & \text{otherwise}.
\end{cases}
\end{aligned}
\]
++}

\begin{figure}[h!]
    \centering
{\scriptsize $\begin{array}{cc|c|c|c|c|c|}
\hline
 \mathbf{blockId} & \mathbf{tupleId} &  \mathtt{T} & A_1 & \cdots & A_n & P   \\ \hline
 f(\tau_{B_1}) & \tau_{B_1}^1  &  t_1 & 1 & \cdots & A_{n,1} &  p_1^1\\ \cline{3-7}
  &  \cdots  &  \cdots & \cdots & \cdots & \cdots & \cdots  \\ \cline{3-7}
& \tau_{B_1}^{n}  &  t_n & 1 & \cdots & A_{n,n} & p_1^n\\ \hline
\cdots  &  1  &  \cdots & \cdots & \cdots & \cdots & \cdots  \\ \hline
 f(\tau_{B_n}) & \tau_{B_n}^1  &    t_1 & A_{1,1} & \cdots & 1 &  p_n^1\\ \cline{3-7}
  &  \cdots  &  \cdots & \cdots & \cdots & \cdots & \cdots  \\ \cline{3-7}
& \tau_{B_1}^{n}  &  t_n & A_{1,n} & \cdots & 1 & p_n^n\\ \hline
\end{array}$}
\vspace{2mm}\caption{BID $D^p$ derived from $D^\star$}
    \label{fig:derivedBID}\vspace{-0.5cm}
\end{figure}

\h The BID instance $D^p$, derived from $(D^\star, \mathcal{M},P^{\mathcal{M}})$ is shown in Figure~\ref{fig:derivedBID}. \ It contains a block $f(\tau_{B_i})$ for each  tuple with MVs $\tau_{B_i}$ in $D^\star$. \ Each block $f(\tau_{B_i})$ consists of tuples $\tau^j_{B_i}$ that correspond to the possible imputations of $\tau_{B_i}$, derived from tuples $t_j$ encoded in $D^p$.

\h The probability $P^{\mathcal{M}}(\tau^j_{B_i})$ assigned to a tuple $\tau^j_{B_i}$ in block $B_i$ of the BID $D^p$ is as follows:
{\footnotesize \begin{equation}
P^\mathcal{M}(\tau^j_{B_i}) =
P^\mathcal{M}\Big(
    T^m = t_j \;\Big|\;
    \bigwedge_{l=1}^n A^m_l \in (t_j, B_l), \  \mathtt{T}^\star = \text{na}, \ 
A^\star_i = 1, \mbb{I}^{A_i} = 0, \ 
\bigwedge_{l \neq i} A^\star_l = \na, \ 
\bigwedge_{l \neq i} \mbb{I}^{A_l} = 1
\Big). 
\end{equation}}
\h According to the MG $\mc{M}$, $T^m$ depends \emph{only} on $A^\star_l$, for $l \in [1, n]$.
Hence, we can simplify the conditional probability to:

\vspace{-5mm}
{\footnotesize \begin{equation*}
\begin{aligned}
P^\mathcal{M}(\tau^j_{B_i}) &= \ P^\mathcal{M}(\mathtt{T}^m = t_j \mid A^\star_i = 1, \bigwedge_{l \neq i} A^\star_l = \na) \ = \ \frac{P^\mathcal{M}(\mathtt{T}^m = t_j,  A^\star_i = 1,  \bigwedge_{l \neq i} A^\star_l = \na)}{P^\mathcal{M}(A^\star_i = 1, \bigwedge_{l \neq i} A^\star_l = \na)} \\
&= \ \frac{\frac{p^j_i}{n}}{\sum\limits_{l=1}^n \frac{p^l_i}{n}} \ = \ \frac{p^j_i}{\sum\limits_{l=1}^n p^l_i} \ = \ p^j_i.
\end{aligned}
\end{equation*}}

\h The following lemma establishes that the transformation from a BID instance $B$ to the derived instance $D^p$ via the intermediate construction $(D^*, \mathcal{M}, P^{\mathcal{M}})$ preserves the probabilistic semantics over the original schema $S$.

\h It should be clear that all the elements associated to the observed DB $D^\star$ can be built in polynomial-time in the size of the original BID DB $D$. \ It remains to verify that for the BID $D^p$ associated to $D^\star$, it holds: \ $D \equiv D^p$.

\ignore{++
\begin{lemma} \em
Let $D$ be an arbitrary BID over a schema $S = \{\mathtt{T}\}$. Let $D^*$, $\mathcal{M}$, and the joint distribution $P^{\mathcal{M}}$ be constructed from $D$ as described above. Then, we have
$D^p \equiv B$, where  $D^p$ is the BID derived from $D^*$, $\mathcal{M}$, and $P^{\mathcal{M}}$.
\end{lemma}

\vspace{1mm}
\noindent {\bf Proof:}++}

\h In fact, let $B_i$, with $i \in \{1, \cdots, n\}$, denote the blocks in the BID  $D$.
By construction, the mapping $f$ which assigns each block
 $B_i$ in $B$  to the block $f(\tau_{B_i})$ in $D^p$ is bijective. Additionally, there is a one-to-one correspondence between the tuples in each block $B_i$ and those in the corresponding block $\mathcal{B}(\tau_{B_i})$, such that:
$\tau_i^j \in B_i, \text{ with } \tau_i^j[\mathtt{T}] = \mathtt{``t_j\hspace{-1mm}"},\text{ and: } p(\tau_i^j) = p_i^j \ \mbox{ iff } \
\tau_{B_i}^j \in \mathcal{B}(\tau_{B_i}), \text{ with } \tau_{B_i}^j[\mathtt{T}] = \mathtt{``t_j}\hspace{-1mm}", \text{ and } p(\tau_{B_i}^j) = p_i{}^j$.
\ We conclude that \ $D \equiv D^p$.\boxtheorem



\section{Appendix: Proof of Theorem \ref{th:complresultsmcc}.}

\h The results in Theorem \ref{th:complresultsmcc}(a) \ignore{\ref{thm:mcc}} are obtained through a  reduction of the \textbf{MCC}-problem to a particular case of {\em minimum-cost flow problem  in bipartite graphs} (MCFP) \cite[chap. 6]{ahuja1993}. The key insight is that finding a most-compliant class corresponds to solving a ``min-cost flow problem" where blocks are matched to tuple multiplicities, while minimizing the distance objective.

\ignore{++++
\h \red{The MCFP, as appearing  in our context, can be formulated in general as follows: \ Consider a directed bipartite graph $\mc{G} = \langle V_1 \stackrel{\cdot}{\cup} V_2, E\rangle$, with edges in $E$ from vertices in $V_1$ to vertices in $V_2$. Each edge $(o,d)$ has a numerical weight (or cost) $w(o,d)$ and a capacity $c(o,d)$. Furthermore, $\mc{G}$ becomes $\mc{G}^\star = \langle \{s_1,s_2\} \cup V_1 \cup V_2, E^\star\rangle$ by adding to extra nodes, $s_1, s_2$, for {\em source} and {\em sink} nodes. The former is connected to all nodes in $V_1$, and the latter becomes a connection for the vertices in $V_2$. 
The new arcs from $s_1$ have capacity $1$ and weight $0$; whereas the arcs pointing to $s_2$ have capacity $1$ and weight $w(v,s_2) \geq 0$. 
\ The optimization problem is that of minimizing the overall cost $\sum_{e \in E^\star} w(e)$, subject to the constraints XXXX}
\comlb{Finish this up above. Below is it stated in terms of "lower-bounds".}
+++++}

\h We first recall the definition  of {\em convex-separable} distance \cite{ahuja1993} as applied to our setting.

\begin{definition} \em (convex-separable distance/divergence) \label{def:sep}\\
(a) A distance $d(C_{\mathbf{k}}, \mathcal{M})$, for classes $C_\mbf{k}$ with
$\mathbf{k}=\langle k_1, \ldots, k_m\rangle$, is \emph{separable} if it can be
expressed as $d(C_{\mathbf{k}}, \mathcal{M})=
\sum_{j=1}^m d_j\!\bigl(k_j, P^\mathcal{M}(t_j)\bigr)$, where
$d_j(k_j,P^\mc{M}(t_j))$ is a function that depends on the number of occurrences
of tuple $t_j$ in $C_\mbf{k}$ ($t_j$'s contribution to the distance). \\
(b) $d_j(k_j, P^\mathcal{M}(t_j))$ is {\em discretely convex} (in its first
argument) if, for all integers $k_j \ge 1$, the discrete {\em marginal
contributions} are non-decreasing: \
$d_j(k_j+1, P^\mathcal{M}(t_j)) - d_j(k_j, P^\mathcal{M}(t_j)) \geq
d_j(k_j, P^\mathcal{M}(t_j)) - d_j(k_j-1, P^\mathcal{M}(t_j))$; and
{\em strictly discretely convex} if these inequalities are strict ($>$) for all
$k_j \ge 1$. \\
(c) $d(C_\mbf{k},\mc{M})$ is {\em convex-separable} if, in addition, each of the
$d_j$ is discretely convex; and {\em strictly convex-separable} if each
$d_j$ is strictly discretely convex.
\boxtheorem
\end{definition}

\h We formalize now the reduction of our MCC-setting (and problems) to a MCFP.
The reduction constructs a flow network whose feasible integral flows are in one-to-one correspondence with the possible worlds of the BID; and the cost of a flow equals the compliance distance of the corresponding class-vector.

\subsection{From MCC to MCFP.}\label{sec:constr}

We first recall the ingredients of our MCC-problems. \ $D^\star$, the observed instance, has an associated BID $D^p$ with a set of blocks $\mc{B} = \{B_1,\dots,B_n\}$; the MG $\mathcal{M}$ indices a joint probability distribution  $P^\mathcal{M}$ on the tuples in $D^p$. The different tuples (no duplicates) are $T = \{t_1, \ldots, t_m\}$. \ See  Remark \ref{rem:canon} for some notation and terms we use here.

\h An instance of the MCC-problem (to be later reduced to a MCFP-problem,) is given by:
\ (a)  The set of \emph{blocks}  $\mc{B}$.
    \ (b)   The set $T$ of distinct \emph{tuples}. \ 
(c) For each block $B_i$ and tuple $t_j$, with $t_j \in B_i$, a probability $p_i(t_j) \in [0,1]$. Notice that a tuple in $T$ may appear in more than one block, and with different probabilities.

\h The goal is to compute an \emph{admissible class-vector} $\mathbf{k}^\star = \langle k_1^\star,\dots,k_m^\star\rangle \in \mathbb{N}^m$, such that:
\ignore{\[
C_{\mathbf{k}^\star} \ \ \in \ \ \argmin_{\substack{\mathbf{k}\in\mathbb{N}^m\\ \sum_{j=1}^m k_j = n}} d^c\!\left(C_{\mathbf{k}}, \mathcal{M}\right).
\]}
\begin{equation}
C_{\mathbf{k}^\star} \ \ \in \ \ \argmin_{\mathbf{k}\in\mathbb{N}^m,\ {\tiny \sum_{j=1}^m} k_j = n} \!\!\!\! \!\! d^c\!\left(C_{\mathbf{k}}, \mathcal{M}\right).\label{eq:vector}
\end{equation}

\h The want to obtain an optimal class-vector as in (\ref{eq:vector}) from an optimal solution to an MCFP. For the latter, 
 we construct a {\em flow network}, namely, a directed graph $\mc{G}(D^p) = \langle V,E\rangle$ with a source and sink nodes $s_l$ and $s_r$, for left and right, resp.; as we now proceed to explain.

\begin{definition} \em (flow network $\mc{G}(D^p)$)\label{def:network}
Let $D^p(D^\star,\mc{M})$ be a BID with blocks $\mc{B}=\{B_1,\dots,B_n\}$ and
tuples $T=\langle t_1,\dots,t_m\rangle$. The associated network is a directed
graph $\mc{G}(D^p)=(V,E)$ equipped with a lower-bound function
$\ell:E\to\mathbb{N}$, a capacity function $c:E\to\mathbb{N}$, and, on the
sink-arcs, a family of convex arc-cost functions.

\smallskip\noindent\textbf{Nodes.}\quad
$V=\{s_l\}\cup L\cup R\cup\{s_r\}$, where
$L=\{b_i \mid B_i\in\mc{B}\}$ is the set of \emph{block-nodes},
$R=\{t_j \mid t_j\in T\}$ is the set of \emph{tuple-nodes},
and $s_l,s_r$ are the \emph{source} and \emph{sink}. The pair $(L,R)$ forms
the bipartition of the middle layer.

\smallskip\noindent\textbf{Arcs.}\quad
$E=E_{s_l}\cup E_{LR}\cup E_{s_r}$, defined layer by layer:
\begin{enumerate}
  \item \emph{Source arcs:} $E_{s_l}=\{(s_l,b_i)\mid b_i\in L\}$, with
        $\ell(s_l,b_i)=c(s_l,b_i)=1$.
  \item \emph{Middle arcs:} $E_{LR}=\{(b_i,t_j)\mid t_j\in B_i \text{ and } p_i(t_j)>0\}$,
        with $\ell(b_i,t_j)=0$, and $c(b_i,t_j)=1$.
  \item \emph{Sink arcs:} $E_{s_r}=\{(t_j,s_r)\mid t_j\in R\}$, with
        $\ell(t_j,s_r)=0$, and
$c(t_j,s_r)=\operatorname{in\text{-}deg}(t_j)$, the number of middle
        arcs entering $t_j$ (equivalently, the number of blocks that may select
        $t_j$).
\end{enumerate}

\vspace{-2mm}\noindent\textbf{Cost.}\quad
\begin{enumerate}
    \item All arcs in $E_{s_l}\cup E_{LR}$ have zero cost. 
    \item Each sink-arc
$(t_j,s_r)$ is assigned a convex \emph{cost function}, \ 
$\operatorname{cost}_j\!\!: \{0,1,\dots,n\}\to\mathbb{R}_{\ge 0}$, of
 non-negative integer throughput variables $x_j$, defined by:
\begin{equation}\label{eq:fj}
  cost_j(x_j)\;:=\;d_j\!\left(\tfrac{x_j}{n},\,P^{\mc{M}}(t_j)\right),
\end{equation}
where $d_j$ is the $j$-th separable component of the distance $d$ \ (see Definition~\ref{def:sep}). The argument
$x_j$ ranges over admissible throughput values on $(t_j,s_r)$; it is
instantiated at the flow value $f(t_j,s_r)$ once a flow $f$ is fixed (see Definition~\ref{def:flow} right below).
\end{enumerate}
\end{definition}

\begin{definition} \em (feasible integral flow) \label{def:flow}
\ (a) \ A \emph{feasible integral flow} on $\mc{G}(D^p)$ is a function
$f:E\to\mathbb{N}$ that satisfies the following constraints:

\begin{enumerate}
\item {\bf Integrality.}
\ignore{\smallskip\noindent\textbf{(i) Integrality.}\quad}
For every arc $e\in E$, $\;f(e)\in\mathbb{N}$.


\item {\bf Bound constraints.}
For every arc $e\in E$, $\;\ell(e)\le f(e)\le c(e)$; namely:
\begin{eqnarray*}
 &&f(s_l,b_i) = 1, \ \mbox{ for } \ b_i\in L, \\
&&0\le f(b_i,t_j) \le 1,  \ \mbox{ for } \ (b_i,t_j)\in E_{LR}, \\
  &&0\le f(t_j,s_r) \le \operatorname{in\text{-}deg}(t_j), \ \mbox{ for } \ t_j\in R.
\end{eqnarray*}

\ignore{\begin{align*}
  f(s_l,b_i) &= 1, && b_i\in L, \\
  0\le f(b_i,t_j) &\le 1, && (b_i,t_j)\in E_{LR}, \\
  0\le f(t_j,s_r) &\le \operatorname{in\text{-}deg}(t_j), && t_j\in R.
\end{align*}}

\vspace{1mm}
\item {\bf Flow conservation at block-nodes.} \ 
For every $b_i\in L$, the inflow from the source equals the outflow into the
middle layer:
\[
  f(s_l,b_i)\;=\!\!\sum_{j:\,(b_i,t_j)\in E_{LR}}\!\! f(b_i,t_j).
\]
Combined with {\bf 2.}, this forces each block-node to send exactly one unit along
a single middle arc.

\item {\bf Flow conservation at tuple-nodes.}
For every $t_j\in R$, the inflow from the middle layer equals the outflow into
the sink:
\[
  \sum_{i:\,(b_i,t_j)\in E_{LR}}\!\! f(b_i,t_j)\;=\;f(t_j,s_r).
\]
\end{enumerate}

\smallskip\noindent
The source $s_l$ and the sink $s_r$ are the only nodes exempt from
conservation; by {\bf 3.} and the source bounds, the total flow leaving $s_l$ and
entering $s_r$ equals $n$. 

\vspace{1mm}
\noindent (b) \smallskip\noindent\textbf{Induced flow-vector and cost.}\quad
Every feasible integral flow $f$ induces the \emph{flow-vector}:

\vspace{-5mm}\begin{equation}\label{eq:flow-vector}
  \mathbf{k}(f)
  \;=\;\bigl\langle f(t_1,s_r),\,f(t_2,s_r),\,\dots,\,f(t_m,s_r)\bigr\rangle
  \;\in\;\mathbb{N}^{m},
  \ \  \mbox{ with } \ \textstyle\sum_{j=1}^{m} f(t_j,s_r)=n;
\end{equation}
and has \emph{cost} obtained by evaluating each sink-arc cost
function \eqref{eq:fj} at the realized throughput $x_j=f(t_j,s_r)$, as follows:
\[
  cost(f)\;=\;\sum_{j=1}^{m}\operatorname{cost}_j\bigl(f(t_j,s_r)\bigr).
\] 

\vspace{-10mm}\boxtheorem
\end{definition}

\begin{definition} \em (minimum-cost flow problem on $\mc{G}(D^p)$) \ \label{def:mcf}
For the network $\mc{G}(D^p)$ in Definition~\ref{def:network}, and 
$\mathcal{F}$ the set of feasible integral flows on it
(as in Definition~\ref{def:flow}). \  The \emph{minimum-cost flow problem}  is computing:
\[
  \min_{f\in\mathcal{F}}\; cost(f)
  \;:=\;\min_{f\in\mathcal{F}}\;\sum_{j=1}^{m}\operatorname{cost}_j\bigl(f(t_j,s_r)\bigr),
\]
with $cost$ and $\operatorname{cost}_j$ as in Definition~\ref{def:flow}. A flow
$f^\star\in\mathcal{F}$ attaining the minimum is an \emph{optimal flow}, and
$cost(f^\star)$ is its \emph{optimal cost}. \boxtheorem
\end{definition}

\begin{remark} (flow value fixed by feasibility) \label{rem:value}
For every $f\in\mathcal{F}$, the constraint in  Definition~\ref{def:flow}.{\bf 2.} \ 
requires $f(s_l,b_i)=1$, for all $b_i\in L$. Then, the value of $f$ satisfies:
\[
  \operatorname{val}(f)\;=\;\sum_{b_i\in L} f(s_l,b_i)\;=\;n\;=\;|L|.
\]
Thus, every feasible integral flow carries exactly $n$ units from $s_l$ to
$s_r$ and is, in particular, a maximum flow. \boxtheorem
\end{remark}

\begin{lemma} \em (complexity) \label{lem:complexity} \ 
Assume each sink-arc cost, $\operatorname{cost}_j$, is convex. Then, the
minimum-cost flow problem in Definition~\ref{def:mcf} is solvable in polynomial-time, and admits an integral optimal flow.
\end{lemma}

\noindent {\bf Proof}: \
The objective function in Definition~\ref{def:mcf} is separable and convex in the
sink-arc throughputs, whose capacities are bounded by
$\operatorname{in\text{-}deg}(t_j)\le n$. Separable convex-cost flow
problems can be reduced to linear min-cost flow problems by means of a standard
arc-splitting technique \cite{ahuja1993}. Here, the reduction adds only
$\sum_j\operatorname{in\text{-}deg}(t_j)=|E_{LR}|=O(mn)$ arcs, and is, therefore,
polynomial. The resulting problem instance is a linear min-cost flow problem, solvable
in polynomial-time with an integral optimum \cite{orlin1988}.
\qquad\boxtheorem

\ignore{\begin{definition}[Network $\mc{G}(D^p)$]\label{def:network}
Let $D^p(D^\star,\mc{M})$ be a BID with a set of blocs $\mc{B}= \{B_1,\dots,B_n\}$ and a set of tuples 
$T=\langle t_1,\ldots,t_m\rangle$.
The associated network $\mc{G}(D^p)=(\mc{B}\cup\mc{T}\cup\{s\},A)$ has three layers.

\smallskip\noindent
\textbf{Block layer} $\mc{B}=\{b_1,\ldots,b_n\}$: node $b_i$
corresponds to block $B_i$ with unit supply $s_i=1$.

\smallskip\noindent
\textbf{Tuple layer} $\mc{T}=\{T_1,\ldots,T_m\}$: node $T_j$
corresponds to $t_j\in T$ and is a pure transshipment node.

\smallskip\noindent
\textbf{Sink} $s$: absorbs all supply via \emph{compliance arcs}.

\smallskip\noindent
\textbf{Compliance arcs}: arc $(T_j,s)$ carries
throughput $x_{js}=\sum_{i}x_{ij}$,
integer capacity $u_j=|\{i:t_j\in\mc{B}(\tau_i)\}|$, and convex cost
\begin{equation}\label{eq:fj}
  f_j(x_{js}) = d_j\!\left(\tfrac{x_{js}}{n},\,P^{\mc{M}}(t_j)\right),
\end{equation}
where $d_j$ is the $j$-th separable component of the distance $d$
(Definition~\ref{def:sep}).
\end{definition}

The  compliance objective now read:
\begin{equation}\label{eq:objectives}
  \displaystyle F(x)=\sum_{j=1}^{m}f_j(x_{js})
\end{equation}

******************
\vspace{1mm}
\noindent {\bf Nodes:} \ 
\red{$V = \{s_l\} \cup L \cup R \cup \{s_r\}$}, with
\ $L = \{b_i \mid B_i \in B\}$, the set of \emph{block-nodes};  \ and 
$R = T$, a set of \emph{tuple-nodes}. \ $L$ and $R$ for the left- and right-sets of nodes in a bipartite graph.

\vspace{1mm}
\noindent {\bf Edges:} \ 
The edge set $E = E_{s_l} \cup E_{LR} \cup E_{s_r}$ consists of three layers: \vspace{1mm}
\begin{enumerate}
    \item \textbf{Source edges:} For each $b_i \in L$, an edge $(s_l,b_i)$, with a capacity $c(s_l,b_i)=1$, and a \red{lower-bound $\ell(s_l,b_i)=1$, to be used later to constrain a flow}. 

    \item \textbf{Middle edges:} For each block $B_i \in \mc{B}$, and tuple $t_j \in T$, with $t_j \in B_i$ and $p_i(t_j) > 0$, an edge $(b_i,t_j)$, with capacity $c(b_i,t_j)=1$.
    
    \item \textbf{Sink edges:} For each $t_j \in R$, an edge $(t_j,s_r)$, with capacity $c(t_j,s_r) := \mbox{\nit{in-deg}}(t_j)$, the in-degree of $t_j$  (the number of its incoming  edges), which, in this case,   is the number of blocks that can select $t_j$ (in the sense that $t_j$ appeared already in the block in $D^p$).
\end{enumerate}

\noindent {\bf Flow:} \ 
A \emph{feasible integral flow} $f$ assigns, to each edge $e\in E$, a non-negative integer $f(e)$ that satisfies the constraints imposed by the lower-bounds and capacities; and also {\em flow conservation} at all intermediate nodes, those in $L\cup R$.  

\h \red{More precisely, XXX}

\comlb{State here, in general, the constraints that have to be satisfied. And, at the same time, what is the general optimization problem.}

\h For the source edges, the lower-bounds $\ell(s_l,b_i)$ forces $f(s_l,b_i)=1$, for every $i \in [1,n]$. 
This can be interpreted as  each block-vertex sending exactly one unit of flow into the network.  By flow conservation, that unit must travel along a middle edge (in $E_{LR}$) to some tuple node $t_j$ (with $p_i(t_j)>0$); and then, to the sink $s_r$ via the corresponding sink-edge.  

\comlb{It is not clear how they are determined.}

\h As a consequence, the flow-values on the sink-edges \red{determine} a {\em flow-vector}:
\begin{equation}
\mathbf{k}(f) = \langle f(t_1,s_r), \ f(t_2,s_r), \ \dots, \ f(t_m,s_r)\rangle \in \mathbb{N}^m  \label{eq:flow-vector}
\end{equation}
that satisfies $\sum_{j=1}^m f(t_j,s_r) = n$, because each block contributes exactly one unit to the total outflow from $s_l$.

\vspace{1mm}
\noindent {\bf Cost:} \ 
The {\em cost of a flow} $f$ is defined solely in terms of the vector $\mathbf{k}(f)$, as follows: $\mathbf{k}(f)$ can be seen as a vector of the form $\mathbf{k}=\langle k_1, \ldots, k_m\rangle$ that characterizes a class, those appearing in (\ref{eq:vector}). We still have to verify, right below, that this is an admissible class-vector (see Remark \ref{rem:canon}). For now, we define: \ (see Definition \ref{def:dist})
\[
\nit{cost}(f) \ := \ d^c\!\left(C_{\mathbf{k}(f)},\mathcal{M}\right).
\]

\comlb{Explain here how this ``cost of a flow" is used, to what constraints it leads, etc. What exactly is maximized? We haven't seen  the full instance of the MCFP yet.}

}

\h We now establish the formal correspondence between feasible integral flows in $\mc{G}(D^p)$ and admissible class-vectors, and prove that an optimal flow yields an optimal solution to the MCC-problem.

\begin{lemma} \label{lemma:reduc} \em
(a) For every feasible integral flow $f$ in $\mc{G}(D^p)$, the flow-vector
$\mathbf{k}(f)$ in (\ref{eq:flow-vector}) is an admissible class-vector; that is,
$\sum_j k_j = n$, and $k_j$ is the number of blocks assigned to tuple $t_j$.
Conversely, for every admissible class-vector $\mathbf{k}$, there exists a
feasible integral flow $f$, such that $\mathbf{k}(f) = \mathbf{k}$. \\ (b) Moreover, in
either case, the cost of the flow equals the distance from the induced class to
the model:
$\nit{cost}(f)\;=\;d^{c}\!\bigl(C_{\mathbf{k}(f)},\,\mathcal{M}\bigr)$.
\end{lemma}

\vspace{1mm} 
\noindent {\bf Proof}: \ (a) \ 
$(\Rightarrow)$ Let $f$ be a feasible integral flow. For each block-node
$b_i$, the source bound in Definition \ref{def:flow}.{\bf 2.} \  gives $f(s_l,b_i)=1$, and conservation at $b_i$, as in Definition \ref{def:flow}.{\bf 3.}, 
forces $\sum_{j}f(b_i,t_j)=1$; since each middle arc satisfies $f(b_i,t_j)\in\{0,1\}$,
exactly one middle arc leaves $b_i$. Thus, $f$ assigns every block to a unique tuple.
By conservation at each tuple-node $t_j$ as in Definition \ref{def:flow}.{\bf 4.},  
\[
  k_j \;=\; f(t_j,s_r) \;=\!\!\sum_{i:\,(b_i,t_j)\in E_{LR}}\!\! f(b_i,t_j)
       \;=\; \bigl|\{\,b_i : b_i \text{ assigned to } t_j\,\}\bigr|.
\]
 Then, $k_j$ is precisely the number of blocks assigned to $t_j$; in particular
$0\le k_j\le\operatorname{in\text{-}deg}(t_j)$; and then,  the sink-bound is automatically met. \
Summing over $j$ and using Definition \ref{def:flow}.{\bf 3.},
$\sum_{j}k_j=\sum_{i}\sum_{j}f(b_i,t_j)=\sum_i 1=n$. Hence, $\mathbf{k}(f)$ is an
admissible class-vector.

\vspace{1mm}
$(\Leftarrow)$ Let $\mathbf{k}$ be an admissible class-vector. By
admissibility, $\mathbf{k}$ arises from an assignment $\sigma$ that maps each block
$b_i$ to a tuple $\sigma(b_i)$ with $(b_i,\sigma(b_i))\in E_{LR}$, and satisfies
$|\sigma^{-1}(t_j)|=k_j$, for every $j$. \ Now, define $f$ by: \   $f(s_l,b_i) :=1$, \  $f(t_j,s_r) := k_j$,  and 
\[
f(b_i,t_j)=\begin{cases}1 & t_j=\sigma(b_i),\\ 0 & \text{otherwise.}\end{cases}
\]
Integrality holds by construction. The bounds are also satisfied: Source arcs carry $1$;
middle arcs carry $0$ or $1$; and $0\le k_j\le\operatorname{in\text{-}deg}(t_j)$. The latter upper-bound holds because the $k_j$ blocks mapped to $t_j$ are among the
$\operatorname{in\text{-}deg}(t_j)$ blocks adjacent to it. Conservation holds at each
$b_i$:  $1=f(b_i,\sigma(b_i))$; and at each $t_j$:
$\sum_i f(b_i,t_j)=|\sigma^{-1}(t_j)|=k_j=f(t_j,s_r)$. Thus, $f$ is a feasible
integral flow with $\mathbf{k}(f)=\mathbf{k}$.

\vspace{1mm} (b) \ 
\emph{Cost-distance identity:}  For any feasible integral flow $f$ with
$\mathbf{k}(f)=\langle k_1,\dots,k_m\rangle$, the definition of the flow cost and of
the sink-arc cost \eqref{eq:fj} give:
\[
  \nit{cost}(f)\;=\;\sum_{j=1}^{m}\operatorname{cost}_j(k_j)
             \;=\;\sum_{j=1}^{m} d_j\!\left(\tfrac{k_j}{n},\,P^{\mathcal{M}}(t_j)\right).
\]

  $k_j$ is the number of blocks assigned to $t_j$, so
$k_j/n$ is the frequency of $t_j$ in the class $C_{\mathbf{k}(f)}$; that is,
$P^{C_{\mathbf{k}(f)}}(t_j)=k_j/n$. Hence, by separability of $d$
(as in Definition~\ref{def:sep}),
\[
  \nit{cost}(f)\;=\;\sum_{j=1}^{m} d_j\!\left(P^{C_{\mathbf{k}(f)}}(t_j),\,P^{\mathcal{M}}(t_j)\right)    \;=\;d^{c}\!\bigl(C_{\mathbf{k}(f)},\,\mathcal{M}\bigr). 
\]

\vspace{-12mm}\boxtheorem

\begin{theorem}\label{thm:mcc-mcf}\em
Solving the minimum-cost flow problem on $\mc{G}(D^p)$
(Definition~\ref{def:mcf}) solves the MCC-problem instance: if $f^\star$
minimizes $\nit{cost}(f)$ over all feasible integral flows, then
$\mathbf{k}(f^\star)$ is an optimal class-vector for the MCC-problem, and
$$\nit{cost}(f^\star)=d^{c}(C_{\mathbf{k}(f^\star)},\mathcal{M})=\min\limits_{\mathbf{k}\text{ admissible}}
 d^{c}(C_{\mathbf{k}},\mathcal{M})$$
\end{theorem}

\noindent {\bf Proof}: \
Let $\mathcal{F}$ be the set of feasible integral flows and $\mathcal{K}$ the set
of admissible class-vectors. By Lemma~\ref{lemma:reduc}, $f\mapsto\mathbf{k}(f)$
maps $\mathcal{F}$ onto $\mathcal{K}$, and $\nit{cost}(f)=
d^{c}(C_{\mathbf{k}(f)},\mathcal{M})$ for every $f\in\mathcal{F}$. Thus
$\nit{cost}$ factors through this surjection as
$\mathbf{k}\mapsto d^{c}(C_{\mathbf{k}},\mathcal{M})$, and the two functions have
the same image over $\mathcal{F}$ and $\mathcal{K}$ respectively; hence
\[
  \min_{f\in\mathcal{F}}\nit{cost}(f)
  \;=\;\min_{\mathbf{k}\in\mathcal{K}} d^{c}(C_{\mathbf{k}},\mathcal{M}).
\]
In particular, if $f^\star$ attains the left-hand minimum, then
$\mathbf{k}(f^\star)$ attains the right-hand one and is therefore an optimal
class-vector for the MCC-problem. \qquad\boxtheorem

\h We have established Theorem \ref{th:complresultsmcc}(a), except for the polynomial-time computational cost, which we now address.

\vspace{1mm}
\noindent {\bf Complexity Analysis:} \ 
By Theorem~\ref{thm:mcc-mcf}, solving the MCC-problem instance reduces to the
minimum-cost flow problem on $\mc{G}(D^p)$, a network with $n+m+2$ nodes, and
$n+|E_{LR}|+m=O(mn)$ arcs. By Lemma~\ref{lem:complexity}, this problem is
solvable in polynomial-time, with an integral optimum. By
Lemma~\ref{lemma:reduc}, that optimum yields an optimal class-vector. Hence the
MCC-problem instance is solved in polynomial-time in $n$ and $m$.

\ignore{XXXXXXXX
\subsection{\bf Complexity Analysis.}
We now analyze the computational complexity of solving the MCC-problem via the minimum-cost flow reduction. The complexity depends critically on the structure of the distance function $d^c$. When $d^c$ is separable convex, it induces a separable convex cost on the sink arcs of the flow in $\mc{G}(D^p)$. It is known that the  {\em minimum-cost flow problems with separable convex cost}, can be solved in polynomial-time \cite{Minoux86}. \red{Now, we analyze what we obtain for our MCC-problem at hand.} 

\begin{lemma} \em
\label{lem:complexitymcc}
With   a  convex-separable distance  $d^c(\cdot, \mathcal{M})$ for an MCC-problem instance with $n$ blocks and $m$ tuples, an optimal admissible class-vector $\mathbf{k}^\star$ can be computed polynomial- time \red{XXX}. 
\end{lemma}

\comlb{Theorem 7 gives a precise polynomial upper-bound. Here below, it is only argued it is polynomial-time.}

\noindent {\bf Proof}: \
By construction, the network $\mc{G}(D^p)$ has $N = O(n + m + nm)$ edges; and all capacities are integral and polynomially-bounded by the input size. In fact,
\begin{enumerate}
    \item[(a)] Source-edges: $n$ edges with unit capacity.
    \item[(b)] Middle-edges: At most $n \times m$ edges, only where $p_i(t_j) > 0$, with unit capacity.
    \item[(c)] Sink-edges: $m$ edges with capacity $\deg(v_j) \le n$.
\end{enumerate}

\h The objective is to minimize a separable convex cost-function of the flow on the sink edges. This is precisely a minimum-cost flow problem with an \emph{edge-separable convex cost-function} \cite{Minoux86}.
The algorithm in  \cite{Minoux86} solves those problems  for integral flows on directed graphs \red{in polynomial-time????}. \boxtheorem

\comlb{We should give a precise bound if we keep the bound in the theorem.}
XXXXXXXXXXXXXXXXXXX}


\subsection{\bf \#P-Completeness of \textbf{\#MCC}.}

We now turn to  Theorem \ref{th:complresultsmcc}(b), proving that \textbf{\#MCC} is \textbf{\#P-complete}, through Lemma \ref{lemma:mcccountmember} and its corollary that establish membership of \#P; and Lemma \ref{lemma:mcccounthard} that establishes hardness.

\begin{lemma}
\label{lemma:mcccountmember} \em
Verifying whether a class $C_{\mathbf{k}}$ is an optimal solution to the MCC-problem  can be done in polynomial-time. 
\end{lemma}

\noindent {\bf Proof}: \
We establish polynomial-time verification through two independent checks, each of them reducible to a classic network-flow problem as appearing in our general reduction in 
Section {\bf C.1}.  

\vspace{2mm}
{\bf Step 1: Admissibility Verification.} \ 
Given a candidate vector $\mathbf{k}=\langle k_1,\dots,k_m\rangle \in\mathbb{N}^m$ with $\sum_{j=1}^m k_j = n$, we have to determine whether there exists a feasible assignment of the $n$ blocks to tuples, such that exactly $k_j$ blocks are assigned to tuple $t_j$. This is equivalent to checking whether the constructed flow network $\mc{G}(D^p)$ admits a feasible integral flow $f$ with $f(t_j,s_r)=k_j$, for all $j$.

\h To test this, construct a modified network $\mc{G}_{\mathbf{k}}$ from $\mc{G}(D^p)$ by setting the capacity of each sink arc $(t_j,s_r)$ to $k_j$ (all other capacities remain as in $\mc{G}(D^p)$, with unit capacities on source arcs and middle arcs). A feasible integral flow of value $n$ in $\mc{G}_{\mathbf{k}}$ exists iff $\mathbf{k}$ is admissible, because such a flow saturates all source arcs (each block sends one unit), and exactly meets the prescribed outflows at the sink.

\h The network $\mc{G}_{\mathbf{k}}$ has $O(n+m)$ nodes and $O(nm)$ edges. Computing its maximum flow value can be done in polynomial time using, e.g., Dinic's algorithm \cite{dinic1970}, which runs in $O(\sqrt{n},|E|)$ for unit-capacity bipartite networks. Thus, admissibility is decidable in polynomial-time.

\vspace{2mm}
{\bf Step 2: Optimality Verification.}
Let $d_{\min}$ denote the minimum compliance distance for the given MCC instance. By Theorem~\ref{th:complresultsmcc}(a), an optimal class-vector can be computed in polynomial-time. So, $d_{\min}$ can be obtained efficiently. Once $d_{\min}$ is known, verifying whether $C_{\mathbf{k}}$ is optimal reduces to computing its compliance distance $d^c(C_{\mathbf{k}},\mathcal{M})$, and checking whether $d^c(C_{\mathbf{k}},\mathcal{M}) = d_{\min}$. The compliance distance depends only on $\mathbf{k}$ and the fixed instance parameters; and is assumed computable in polynomial-time (as is typical for distance functions in such settings). Thus, optimality can be checked in polynomial-time. \boxtheorem

\begin{corollary}
\label{cor:mccinp}
\em The problem of counting the number of distinct optimal MC-classes, denoted \#MCC, belongs to the complexity class \#P.
\end{corollary}

\noindent {\bf Proof}: 
Since both verification steps can be done in polynomial-time, the decision problem ``Is $\mathbf{k}$ an optimal MC-class?"   belongs to the class \textbf{P}; and, therefore,  \#MCC is in \#P. \boxtheorem

\begin{lemma} \em
\label{lemma:mcccounthard}
\textbf{\#MCC} is \#P-hard.
\end{lemma}

\vspace{1mm}\noindent {\bf Proof}: \ 
We do a reduction  from counting bases of a transversal matroid, which is \#P-complete
\cite{colbourn95}. A transversal matroid $M$ of rank $r$ is given
by a bipartite graph between a ground set $\Omega=\{g_1,\dots,g_p\}$ and defining
nodes $A=\{a_1,\dots,a_r\}$. Its bases are the $r$-subsets of $\Omega$ that admit
a perfect matching saturating $A$.

\h Let us fix an instance of a matroid $M=(\Omega \cup A, E_M)$, with $E_M \subseteq \Omega \times A$. We construct, in time-polynomial in $|E_M|$, a BID $D^p$ whose MC-classes are in bijection with the bases of $M$. The BID can be seen,  by
Proposition~\ref{lemma:BIDrepresentation}, as derived from  an observed database $D^\star$, and an  MG $\mathcal{M}$ that defines a join distribution $P^{\mathcal M}$.

\vspace{1mm}
{\bf Construction of the BID.}
 \ Introduce one tuple $t_g$ per ground element
$g\in\Omega$, so the support has size $p$; and one block
$B_a=\{t_g:(g,a)\in E_M\}$ per defining node $a\in A$, with uniform
probabilities inside the block. Make  $P^{\mathcal M}$ uniform: $P^{\mathcal M}(t_g):=1/p$. This is
polynomial in the size of $M$.

\vspace{1mm}
\h \textbf{Classes and cost.}
Each world selects one tuple per block, matching every $a\in A$ to an adjacent
ground element; the induced class-vector is $k_g= |\{a:B_a\text{ selects }t_g\}|$,
with $\sum_g k_g=r$.  Assume a strictly convex separable distance function $d(C_{\mathbf{k}}, \mathcal{M})=
\sum_{j=1}^m d_j\!\bigl(k_j, P^\mathcal{M}(t_j)\bigr)$. Since $P^{\mathcal M}$ is uniform, 
every tuple has the same target value
$1/p$. The per-tuple cost functions, therefore, all reduce to a single strictly convex function  of
the count.  Denoting
$\phi(k):=d_1(k,1/p)$, the cost is $d(C_{\mathbf k},\mathcal M)=\sum_g \phi(k_g)$.

\vspace{2mm}

\h \textbf{Case 1:} \  {\bf $A$ can be saturated (i.e. $M$ has at least one basis).} \ Suppose $E_M$ admits a matching
saturating $A$, i.e.\ $M$ has a basis. Then, there exist admissible $0/1$-vectors, i.e.\ classes $C_{\mathbf k}$ whose
count-vector $\mathbf k$ has every entry $k_g\in\{0,1\}$ (equivalently, in which
no tuple occurs more than once). By
strict convexity of $\phi$, the increments $(\phi(k)-\phi(k-1))$ are strictly
increasing. So, any admissible vector with a coordinate $k_g\ge2$ is strictly
improved by shifting one unit to a coordinate equal to $0$, \ which exists since
$\sum_g k_g=r\le p$. Hence, the minimizers of $\sum_g\phi(k_g)$ are exactly the
$0/1$-vectors with $r$ ones. Such a vector selects $r$ distinct tuples, one per
block, i.e.\ a matching saturating $A$ onto a set $S\subseteq\Omega$. with
$|S|=r$. As a consequence,  $S$ is a basis of $M$. 
Conversely, every basis $S$ is matchable to $A$ and yields the admissible
$0/1$-class, $C_{\mathbf k_S}$, of common minimum cost $r\,\phi(1)$, where
$\mathbf k_S\in\{0,1\}^{p}$ is the indicator vector of $S$, i.e.\
$(\mathbf k_S)_j=1$ iff $g_j\in S$, and $0$ otherwise.

\h The map $S\mapsto\mathbf C_{\mathbf k_S}$ is, therefore, a bijection between
the bases of $M$ and the MC-classes: It is well defined and surjective by the
above, and injective since distinct bases are distinct subsets, and then, distinct
indicator vectors. Consequently, $\#\textbf{MCC}=$ number of bases of $M$.

\vspace{2mm}

\h \textbf{Case 2:} \ {\bf $A$ cannot be saturated.} \ In this case,  $M$ has no basis, and
the number of bases of $M$ is $0$. Now, no \textbf{\#MCC} oracle call is needed.

\vspace{1mm}
\h Whether $E_M$ admits a matching saturating $A$ is decidable in polynomial-time. \ Then,  the number of bases of $M$ is obtained from a polynomial-time test; 
and, in Case~1, with a single \textbf{\#MCC} oracle call. \ Since counting
transversal-matroid bases is \#P-complete \cite{colbourn95},
\textbf{\#MCC} is \#P-hard. \boxtheorem

\subsection{\textbf{MCC-Enum} is in \textbf{DelayP}.}

We give a polynomial-delay algorithm enumerating all MC-classes. It performs a
depth-first traversal of the tree of \emph{class-vector prefixes}, using a
min-cost flow computation at each node, to prune any prefix that cannot be
extended to a most-compliant class.

\vspace{1mm}
\noindent {\bf Search Tree.} \ 
Recall that a class is an admissible class-vector
$\mathbf{k}=\langle k_1,\dots,k_m\rangle\in\mathbb{N}^m$ with
$\sum_{j=1}^m k_j=n$ (Remark~\ref{rem:canon}), where $m=|T|$ is the number of
distinct tuples, and $n$ the number of blocks. We organize the class-vectors in a
rooted tree of depth $m$:
\begin{itemize}
  \item[(a)] The root (depth $0$) is the empty prefix $\langle\,\rangle$.
  \item[(b)] A node at depth $l$ is a \emph{prefix} $\mathbf{h}=\langle
        k_1,\dots,k_l\rangle$, fixing the occurrence counts of the first $l$
        tuples $t_1,\dots,t_l$.
  \item[(c)] The children of $\mathbf{h}=\langle k_1,\dots,k_l\rangle$ are the
        prefixes $\langle k_1,\dots,k_l,k_{l+1}\rangle$ obtained by assigning to
        the next tuple $t_{l+1}$ a count $k_{l+1}\in\{0,1,\dots,n\}$. A child is
        kept only if it satisfies the necessary condition
        $\sum_{j=1}^{l+1} k_j\le n$.
\end{itemize}
Thus, level $l$ fixes the count of tuple $t_l$, and the branching factor is at
most $n+1$. A leaf, at depth $m$, is a complete vector
$\langle k_1,\dots,k_m\rangle$. The retained leaves, those additionally
satisfying $\sum_{j=1}^m k_j=n$, are exactly the admissible class-vectors, i.e.\
the classes. Among these, the MC-classes are those of minimum compliance
distance. The tree structure alone does not single them out, and it is the
validity test below that does it.

\vspace{1mm}{\bf Global Optimum.} \ 
First compute, in polynomial time, the global minimum compliance cost: \ $F^\star \;=\; \min\nolimits_{f}\ \nit{cost}(f)$ \ 
of an optimal flow $f$ on $\mc{G}(D^p)$, by solving the min-cost flow problem of
Appendix~\ref{sec:constr}. This is done once, at initialization.

\vspace{1mm}
\noindent {\bf Validity Test.} \ A node $\mathbf{h}=\langle k_1,\dots,k_l\rangle$ is \emph{valid} if it can be
extended to (the prefix of) some MC-class, i.e.\ if there exists an admissible
class-vector $\mathbf{k}$ with prefix $\mathbf{h}$ and
$d^c(C_{\mathbf{k}},\mc{M})=F^\star$. 

\h We decide validity by a single min-cost
flow computation, as follows: \  Construct the network $\mc{G}(D^p)^{\mid\mathbf{h}}$ from
$\mc{G}(D^p)$ by fixing the sink-arcs of the already validated tuples, and leaving the
others free:
\[
  \ell(t_j,s_r)=c(t_j,s_r):=k_j, \mbox{ for } 1\le j\le l;
  \ \ 
  \ell(t_j,s_r)=0; \ \  c(t_j,s_r)=\operatorname{in\text{-}deg}(t_j), \mbox{ for } l< j\le m.
\]

\h The equalities for $j\le l$ force each validated tuple $t_j$ to be selected
exactly $k_j$ times.  The free bounds for $j>l$ let the remaining
$(n-\sum_{j\le l}k_j)$ units of flow distribute over the not yet validated tuples. Then:
\[
  \mathbf{h}\ \text{is valid}
 \ \  \mbox{ iff } \ \ 
  \mc{G}(D^p)^{\mid\mathbf{h}} \ \text{admits a feasible flow, and its minimum cost equals } F^\star .
\]

\h If $\mc{G}(D^p)^{\mid\mathbf{h}}$ is infeasible (the fixed counts $k_j$ are not
jointly realizable), or its minimum cost exceeds $F^\star$, then no completion of
$\mathbf{h}$ is an MC-class. and $\mathbf{h}$ is invalid. Both the construction of
$\mc{G}(D^p)^{\mid\mathbf{h}}$ and the computation of its minimum cost can be done in
polynomial-time.

\vspace{1mm} {\bf Algorithm.} \ 
Traverse the tree depth-first from the root, expanding a node only if it is
valid, and outputting a class-vector each time a valid leaf (at depth $m$) is
reached. Since a valid leaf is an admissible vector of cost $F^\star$, it is an
MC-class. Conversely every MC-class is a leaf all of whose prefixes are valid,
hence is reached. Each MC-class is output exactly once, as distinct MC-classes
have distinct vectors. and thus, distinct root-to-leaf paths.

\vspace{1mm} \noindent {\bf Polynomial Delay.} \ 
The validity test guarantees that \emph{every expanded node lies on a
root-to-leaf path ending at an MC-class}: A valid node has, by definition, a
completion of cost $F^\star$, and every prefix of that completion is itself valid. \ 
Consequently, the search never enters a subtree devoid of solutions. Between two
consecutive outputs, the traversal, therefore, follows at most one descending path,
and one ascending path in the tree, i.e.\ it visits $O(m)$ nodes, each having at
most $n+1$ children to test. So, at most $O(mn)$ validity tests are performed
between consecutive MC-classes (and before the first, and after the last). As
each validity test is a polynomial-time min-cost flow computation, the delay
between consecutive outputs is polynomial.  \boxtheorem










\subsection{\bf \#P-Completeness of \textbf{\#MCW}.}

Now we prove  Theorem~\ref{th:complresultsmcc}(d). 

\vspace{1mm}
\noindent \emph{Membership.} \ 
A world $W$ is a polynomial-size certificate that can be verified in polynomial-time:
$W$ is an admissible world (verified by exhibiting a concrete flow, as in the
admissibility verification of Lemma~\ref{lemma:mcccountmember}), and, since, by (a), the
minimum compliance cost $F^\star$ is computable in \textbf{FP}, the
predicate ``$d^c(W,\mc M)=F^\star$'' is polynomial-time verifiable. As a consequence, 
\textbf{\#MCW} counts the accepting certificates of a polynomial-time predicate,
and belongs in \#P.

\vspace{1mm}\noindent\emph{Hardness.}\quad
Lemma~\ref{lemma:mcwcounthard} below establishes hardness by a reduction from
counting perfect matchings in $3$-regular bipartite graphs \cite{dagum92}.

\vspace{1mm}
\begin{lemma} \em
\label{lemma:mcwcounthard}
\textbf{\#MCW} is \#P-hard for every convex-separable compliance distance whose
per-tuple components are strictly convex; in particular for the (squared)
Euclidean distance.
\end{lemma}

\vspace{1mm}\noindent {\bf Proof}: \
We prove hardness by a polynomial-time reduction from {\em counting perfect
matchings in 3-regular bipartite graphs}, which is \#P-complete \cite{dagum92}.

\h Let $\mc{G} = (U \cup V, E)$ be a 3-regular bipartite graph with
\(|U| = |V| = n\). We construct a BID $D^p$. By
Proposition~\ref{lemma:BIDrepresentation}, the resulting BID can be realised as
the BID induced by an observed database and a missingness graph; we may
therefore concentrate on building $D^p$, as follows.

\vspace{-2mm}
\begin{enumerate}
    \item \textbf{Tuples.} For each edge \((u_i, v_j) \in E\), create a
    distinct tuple \(t_{ij}\). Let \(T = \{ t_{ij} : (u_i, v_j) \in E \}\).
    Since \(\mc{G}\) is 3-regular with \(|U|=|V|=n\), \  \(|T| = 3n\).

\vspace{1mm}
    \item \textbf{Primary Blocks (of type \(B\)).} For each \(u_i \in U\),
    define a block \(b_i\) containing the tuples of its incident edges: If
    \(u_i\) is adjacent to \(v_{j_1}, v_{j_2}, v_{j_3}\), then
    \(b_i = \{ t_{ij_1}, t_{ij_2}, t_{ij_3}\}\). The distribution in \(b_i\) is
    uniform: \  \(p_i(t_{ij}) = \tfrac{1}{3}\).

\vspace{1mm}
    \item \textbf{Supplementary Blocks (of type \(S\)).} For each \(v_j \in V\),
    list its incident tuples in a fixed order \(t_{i_1j}, t_{i_2j}, t_{i_3j}\).
    Add two supplementary blocks:  \(s_{j1} = \{t_{i_1j}, t_{i_2j}\}\) and
    \(s_{j2} = \{t_{i_2j}, t_{i_3j}\}\), each with the uniform distribution
    \(\tfrac{1}{2}\) on its two tuples. Then,  \(t_{i_1j}\) belongs only to
    \(s_{j1}\); \(t_{i_2j}\) to both \(s_{j1}\) and \(s_{j2}\); and \(t_{i_3j}\)
    only to \(s_{j2}\); while each \(t_{ij}\) belongs to exactly one primary
    block \(b_i\).

\vspace{1mm}
    \item \textbf{Joint Distribution.} \(P^{\mathcal{M}}\) is uniform over all
    tuples: \ \(P^{\mathcal{M}}(t) := \tfrac{1}{3n}\), for every \(t \in T\).
\end{enumerate}

The construction of $D^p$ is computable in polynomial-time in the size of
\(\mc{G}\).

\vspace{1mm}
\h \emph{Cost.} We can use any convex-separable compliance distance:
\begin{equation}\label{eq:sep}
  d(C_{\mathbf k},\mathcal M)\;=\;\sum_{j=1}^{3n} d_j\!\bigl(k_j,\,P^{\mathcal M}(t_j)\bigr),
\end{equation}
with  each component  \emph{strictly convex} in its integer argument
\(k_j\). \ That is, its discrete second differences are strictly positive:
\begin{equation}\label{eq:strict}
  d_j(k+1)-2\,d_j(k)+d_j(k-1)\;>\;0, \ \ \mbox{ for } k\ge 1.
\end{equation}

\h Since \(P^{\mathcal M}\) is uniform, the components are identical:
\(d_j(\cdot,P^{\mathcal M}(t_j))=\delta(\cdot)\), for a common strictly convex
\(\delta\). For example,  the squared Euclidean distance
\(\delta(k)=\bigl(\tfrac{k}{3n}-\tfrac{1}{3n}\bigr)^2\), in this case. \ignore{ is the running example.} \ 
Then, the compliance cost of a world depends only on its occurrence-count vector
\(\mathbf k\), through \ \(F(\mathbf k)=\sum_{j}\delta(k_j)\).

\vspace{1mm}
\h \emph{A global count.} \ Each of the \(3n\) blocks (\(n\) primary, \(2n\)
supplementary) selects exactly one tuple, so every world \(W\) satisfies:
\begin{equation}\label{eq:sum3n}
  \sum_{j} k_j(W) \;=\; 3n .
\end{equation}

\h Now, consider minimizing \(F(\mathbf k)=\sum_{j}\delta(k_j)\) over non-negative integer
vectors, subject to \eqref{eq:sum3n} on \(3n\) coordinates. By \eqref{eq:strict}
the marginal increments \(\Delta(k):=\delta(k)-\delta(k-1)\) are strictly
increasing in \(k\). Hence, for any vector with two coordinates \(k_a\ge k_b+2\),
moving one unit from \(a\) to \(b\) changes the cost by
\(\Delta(k_b+1)-\Delta(k_a)<0\), a strict decrease. \ Thus, no minimizer has two
coordinates differing by \(2\) or more; with total mass \(3n\) spread over exactly
\(3n\) coordinates. The only such vector is  that for which \(k_j=1\), for all \(j\). It is, 
therefore, the \emph{unique} minimizer of \(F\) under \eqref{eq:sum3n}. We call a
world where every tuple is selected exactly once, a \emph{unit world}. The
minimum-cost worlds are precisely the unit worlds, when they exist.

\vspace{1mm}
\h \emph{Column analysis.} Fix \(v_j\in V\) with incident tuples
\(t_{i_1j},t_{i_2j},t_{i_3j}\). The supplementary blocks \(s_{j1},s_{j2}\) select
one tuple each, both from this triple, contributing exactly with two selections to
column \(v_j\). Let \(a_j\) be the number of \emph{primary} blocks selecting a
\(v_j\)-column tuple. \  Then, the three tuples of \(v_j\) receive \(a_j+2\)
selections. In a unit world, each is selected once, so \(a_j+2=3\), i.e.
\begin{equation}\label{eq:aj}
  a_j \;=\; 1, \ \ \mbox{ for every } \ v_j\in V.
\end{equation}

\h The exclusion of \(a_j=0\) and \(a_j\ge2\) is, thus, \emph{global}, forced by
\eqref{eq:sum3n}, and the two mandatory supplementary selections per column,
rather than by any local property of a single column.

\vspace{1mm}
\h \emph{Forced supplementary completion.} Assume \eqref{eq:aj}, and let the
unique primary selection at \(v_j\) be \(t\). The two other tuples of \(v_j\)
must each be selected once, by \(s_{j1},s_{j2}\). Checking the three cases
against \(s_{j1}=\{t_{i_1j},t_{i_2j}\}\), \(s_{j2}=\{t_{i_2j},t_{i_3j}\}\): If
\(t=t_{i_1j}\), then \(s_{j1}=t_{i_2j},\,s_{j2}=t_{i_3j}\). If \(t=t_{i_3j}\), then
\(s_{j1}=t_{i_1j},\,s_{j2}=t_{i_2j}\). And, if \(t=t_{i_2j}\), then
\(s_{j1}=t_{i_1j},\,s_{j2}=t_{i_3j}\). In each case the assignment of
\(s_{j1},s_{j2}\) is unique. Hence, the supplementary part of a unit world is
uniquely determined by its primary part.

\vspace{1mm}
\h \emph{Bijection with perfect matchings.} \ In a unit world, each primary block
\(b_i\) selects one edge incident to \(u_i\), and by \eqref{eq:aj}, each \(v_j\)
receives exactly one such primary-selected edge. The selected edges form, therefore,
a perfect matching \(M\) of \(\mc G\). Conversely, given a perfect matching
\(M\), assigning to each \(b_i\) its matched edge, yields \(a_j=1\), for all \(j\),
which the forced completion extends to a unique unit world. \ The correspondence is
a bijection.\ignore{: Distinct matchings differ on some edge, and then, assign to some primary
block a different tuple, and yield distinct worlds.\ignore{; and  every unit world \red{maps
to a specific matching as above???}.}} \ As a consequence: 

\h {\em The number of minimum-cost worlds is the same as the number of perfect matchings of $\mc{G}$, whenever a unit worlds exists.} \hfill (*) \ignore{Consequently
\begin{equation}\label{eq:mcw-pm}
  \#\{\text{minimum-cost worlds}\}
  \;=\; \#\{\text{perfect matchings of }\mc G\}
  \qquad\text{whenever a unit world exists.}
\end{equation}}

\vspace{1mm}
\h \emph{Reduction.} A unit world exists \ iff \ \(\mc G\) has a perfect matching,
which is decidable in polynomial-time. If a unit world exists, the most-compliant worlds are
exactly the unit worlds and, by (*), \ 
$\#\textbf{MCW}=$  the number of perfect matchings of $\mc{ G}$. 

\h If there are no unit worlds, the number of 
perfect matchings of $\mc{G}$ is $0$,  and no \textbf{\#MCW} query is
needed. \ In either case, the number of perfect matchings of \(\mc G\) is obtained
from a polynomial-time matching-existence test, and at most one \textbf{\#MCW}
oracle call. Since counting perfect matchings in 3-regular bipartite graphs is
\#P-complete \cite{dagum92}, \textbf{\#MCW} is \#P-hard. \boxtheorem

\ignore{
\vspace{1mm}\h 
Membership and hardness together establish that \textbf{\#MCW} is \#P-complete,
proving Theorem~\ref{th:complresultsmcc}(d).}

\subsection{\bf \#P-Completeness of \textbf{\#Class} and \textbf{\#MCClass}.}
We establish Theorem~\ref{th:complresultsmcc}(e): both \textbf{\#Class} and
\textbf{\#MCClass} are \#P-complete.

\vspace{1mm}\noindent\emph{Membership.}\quad
For both problems a world $W$ is a polynomial-size certificate whose validity is
checkable in polynomial time. \ In fact, first, that $W$ is an admissible world (one tuple
selected per block) is verified in polynomial-time, it exhibits a concrete
flow. \ 
Second, the defining predicate for \textbf{\#Class}, namely ``$W$ has
count-vector $\mathbf k$''  is polynomial-time checkable.  And, since the minimum
compliance cost $F^\star$ is computable in \textbf{FP} (by Theorem~\ref{th:complresultsmcc}~(a)), the predicate
``$d^c(W,\mc M)=F^\star$'' for \textbf{\#MCClass} is also verifiable in polynomial-time. As a consequence, both \textbf{\#Class} and \textbf{\#MCClass} count the accepting
certificates of a polynomial-time predicate; and then, both lie in \#P.

\vspace{1mm}\noindent\emph{Hardness.}\quad
We use the construction and cost of Lemma~\ref{lemma:mcwcounthard}. By the
strict-convexity argument, the all-ones vector $\mathbf k^\star=(1,\dots,1)$ is
the \emph{unique} minimiser of $F$ under \eqref{eq:sum3n}. Then, when $\mc G$
admits a perfect matching, $C_{\mathbf k^\star}$ is the unique minimum-cost admissible class. Its cardinality, $|C_{\mathbf k^\star}|$, is the number of worlds
with count-vector $\mathbf k^\star$, i.e.\ the number of unit worlds, which, by the
bijection of Lemma~\ref{lemma:mcwcounthard}, equals the number of perfect matchings of
$\mc G$. 

\h As perfect-matching existence is decidable in polynomial-time (and the
count is $0$, otherwise), a polynomial-time test together with one oracle call
yields the number of perfect matchings of $\mc G$. This single instance witnesses
hardness for both \textbf{\#Class} (with $\mathbf k^\star$ given) and
\textbf{\#MCClass} (with $C_{\mathbf k^\star}$ selected by optimality). \ Since
counting perfect matchings in $3$-regular bipartite graphs is \#P-complete
\cite{dagum92}, both problems are \#P-hard.
\boxtheorem



\ignore{XXXXXXXXX

\section{Appendix: \red{Material for Section \ref{sec:mpdbs}}}
\ignore{of Theorem \ref{th:complresultsmp}}\label{sec:app3}

\comlb{\red{\ul{\bf Why this is not part of section E????}}}

\comlb{OLD: This stuff here below was in the main body of the paper. I moved it here, but it has to be harmonized with what follows. See next comment below.}

 The complexity results in this theorem are established through two main reductions that exploit the algebraic structure of class probabilities.

 \begin{lemma} \em (class probability to permanent)\label{lem:class-permanent}

\vspace{1mm}\noindent
    1. For any $n \times n$ stochastic matrix $M$, there exists a polynomial-time reduction that constructs a BID instance $D^p$ associated with $M$, consisting of $n$ blocks and $n$ distinct tuples, such that:
    $$\red{\mathsf{perm}(M)} = P^{\mathcal{C}}(C_{\mathbf{k}^\star}), \text{ where } \mathbf{k}^\star = (1, \cdots, 1),$$
    where $\mathsf{perm}(M)$ denotes the permanent of matrix $M$.

    \vspace{1mm}\noindent
    2. Let $D^p$ be a BID with $n$ blocks and $n$ distinct tuples, and let $C_{\mathbf{k}}$ be a class with $\mathbf{k} \in \mathbb{N}^n$, such that $\sum_{i=1}^n k_i = n$. There exists a polynomial-time reduction that constructs a stochastic matrix $M \in [0,1]^{n \times n}$, such that:
    $$P^{\mathcal{C}}(C_{\mathbf{k}}) = \frac{\mathsf{perm}(M_{\mathbf{k}})}{\prod_{i=1}^n k_i!},$$
    where $M_{\mathbf{k}}$ is an $n \times n$ matrix obtained by duplicating each column $i$ of $M$ exactly $k_i$ times.
\boxtheorem
\end{lemma}

\h This lemma establishes the fundamental algebraic connection between class probability computation and permanent evaluation, providing the technical foundation for all complexity results. The first part demonstrates that permanent computation can be encoded as a class probability problem, establishing the hardness direction by reducing the $\#P$-complete permanent problem to class probability computation. The second part shows the converse direction, revealing how any class probability can be expressed in terms of a permanent of a structured matrix with duplicated columns, thereby establishing membership in $\#P$. Together, these bidirectional reductions yield $\#P$-completeness for the \textbf{Class Probability} problem.

\h Moreover, the second reduction naturally extends to handle the optimization and decision variants by associating different classes with a parametric family of matrices obtained through systematic column duplication according to the class specification vector $\mathbf{k}$. This parametric construction makes the membership of \textbf{MPC[D]} and \textbf{CMPC[D]} in $\text{NP}^{\text{PP}}$ straightforward: the algorithm non-deterministically guesses a class vector $\mathbf{k}$ (the NP step), constructs the corresponding matrix $M_{\mathbf{k}}$, and queries a PP oracle to decide whether $P^{\mathcal{C}}(C_{\mathbf{k}}) \geq \theta$.
The $\text{NP}^{\text{PP}}$-completeness of \textbf{CMPC[D]} is established via a polynomial-time many-one reduction from the canonical $\text{NP}^{\text{PP}}$-complete problem \textbf{E-MAJSAT}, as outlined below.

\h Given a Boolean formula $\varphi(x_1,\dots,x_n, y_1,\dots,y_m)$, \textbf{E-MAJSAT} asks whether there exist an assignment $\alpha \in \{0,1\}^n$ to the $x$ variables such that $\varphi(\alpha,\beta) = 1$ for at least half of all $2^m$ assignments $\beta \in \{0,1\}^m$ to the $y$ variables?

\textbf{E-MAJSAT} first guesses an assignment $\alpha$ for the $x$ variables (NP step), then checks whether $\varphi(\alpha,y)$ is satisfied by a majority of $y$ assignments (PP oracle step).

\h Building on Valiant’s reduction from \#SAT to the permanent \cite{Valiant1979}, we adapt his gadget-based construction: by employing a subtle mechanism of column duplication, the variable gadgets corresponding to the existential variables $x$ are modified to encode a class vector $\mathbf{k}$, thereby replacing nondeterministic guessing of $x$ with the deterministic selection of $\mathbf{k}$. The counting over the universal variables $y$ is then realized via the permanent of a carefully constructed matrix, in accordance with Valiant’s original proof (see Appendix \ref{sec:app3} for details).

\h Finally, since \textbf{MPC[D]} belongs to $\text{NP}^{\text{PP}}$, the corresponding optimization problem \textbf{MPC[O]} lies in the complexity class $\text{FP}^{\text{NP}^{\text{PP}}}$, which consists of function problems solvable in polynomial time using an $\text{NP}^{\#\text{P}}$ oracle.
This is achieved via a binary search over the possible values of the projection cost $\theta$. Each query asks whether there exists a class assignment $\mathbf{k}$ such that $P^{\mathcal{C}}(C_{\mathbf{k}}) \geq \theta$, which is an $\text{NP}^{\text{PP}}$ decision problem.
Although $\theta$ is a rational number in $[0,1]$, the set of achievable probabilities $P^{\mathcal{C}}(C_{\mathbf{k}})$ is finite and can be represented using a number of bits polynomial in the input size. Therefore, the binary search only needs to consider polynomially many thresholds to identify the maximum probability.

\h Consequently, \textbf{MPC[O]} is polynomial-time reducible to \textbf{MPC[D]} with an $\text{NP}^{\text{PP}}$ oracle, placing it in $\text{FP}^{\text{NP}^{\text{PP}}}$.

\comlb{OLD: This below was the proof that was already in the appendix.}

XXXXXXXXXXXXXXXXXXXXXXXXXXX}

\section{Appendix: Proof of Theorem \ref{th:complresultsmp}.}
\label{proof:mpc}

The following  lemma establishes the fundamental algebraic connection between class probability computation and permanent evaluation, providing the technical foundation for all complexity results in  Theorem \ref{th:complresultsmp}.

 \begin{lemma} \em (class probability to permanent)\label{lem:class-permanent}

\vspace{1mm}\noindent
    1. For any $n \times n$ stochastic matrix $M$, there exists a polynomial-time reduction that constructs a BID instance $D^p$ associated with $M$, consisting of $n$ blocks and $n$ distinct tuples, such that:
    $$\mathsf{perm}(M) \ = \ P^{\mathcal{C}}(C_{\mathbf{k}^\star}), \ \text{ where } \mathbf{k}^\star = (1, \cdots, 1),$$
    where $\mathsf{perm}(M)$ denotes the permanent of matrix $M$.

    \vspace{1mm}\noindent
    2. Let $D^p$ be a BID with $n$ blocks and $n$ distinct tuples, and let $C_{\mathbf{k}}$ be a class with $\mathbf{k} \in \mathbb{N}^n$, such that $\sum_{i=1}^n k_i = n$. There exists a polynomial-time reduction that constructs a stochastic matrix $M \in [0,1]^{n \times n}$, such that:
    $$P^{\mathcal{C}}(C_{\mathbf{k}}) \ = \ \frac{\mathsf{perm}(M_{\mathbf{k}})}{\prod_{i=1}^n k_i!},$$
    with $M_{\mathbf{k}}$ an $n \times n$ matrix obtained by duplicating each column $i$ of $M$ exactly $k_i$ times.
\boxtheorem
\end{lemma}

\ignore{+++ \noindent {\bf Proof (Lemma \ref{lem:class-permanent}):} \
Consider an $n \times n$ stochastic matrix $M = (m_{ij})$ with rows $b_1, \cdots, b_n$ and columns $t_1, \cdots, t_n$. We define a BID instance $D^P$ associated with $M$, consisting of $n$ blocks, $B = \{b_1, \cdots, b_n\}$, and a set of  distinct tuples $T = \{t_1, \cdots, t_n\}$. The probability of tuple $t_j$ in block $b_i$ is given by:
$
p_{b_i}(t_j) = m_{ij}$.
Let $n_j$ denote the number of occurrences of  $t_j$ across the blocks in $B$ (i.e. the number of blocks that contain $t_j$).
 
\h A world $W$ is a multiset of (bag) cardinality $n$. Consequently, a class $C_{\mathbf{k}}$ of the BID $D^p$ is uniquely defined by a vector of positive integers $\mathbf{k}=(k_1, \cdots, k_n)$, with $\sum\limits_{i \in [1,n]} k_i=n$. Here, each $k_i \leq n_i$ gives the number of occurrences of tuple $t_i$ in class $C_{\mathbf{k}}$.
A class $C_{\mathbf{k}}$, with $\mathbf{k}=(k_1, \cdots, k_n)$ is associated  to a $n \times n$ matrix $M_{\mathbf{k}}$, the matrix formed by $k_i$ copies of each column $t_i$ of $M$.

\h The probability of a class $C_{\mathbf{k}}$ defined by $\mathbf{k}$ is given by:
\begin{equation*}
P^\mc{C}(C_{\mathbf{k}}) \ := \  \sum_{W \in C} P^\mc{W}\!(W) \ = \ \sum_{W \in C_{\mathbf{k}}} \prod_{b_i\in B} P_{b_i}(t^W_i) \
=  \ \frac{\mathsf{perm}(M_{\mathbf{k}})}{\prod\limits_{i=1}^n k_j!}.
\end{equation*}

\h In particular, when $\mathbf{k_0} = (1, \cdots, 1)$, the matrix $M_{\mathbf{k_0}}$ coincides with $M$, and  $\prod\limits_{j=1}^n k_j! = 1$. \ Accordingly, we obtain
$P^{\mathcal{C}}(C_{\mathbf{k_0}}) = \mathsf{perm}(M).
$
\boxtheorem  ++++}

\noindent {\bf Proof (Lemma \ref{lem:class-permanent}):} \ 
Both directions rely on the same correspondence between stochastic matrices and
BIDs. We formulate it once, and prove the general identity of part~2, from which
part~1 follows by specialization.

\vspace{1mm} \noindent  {\bf The Correspondence.} \ 
Given an $n \times n$ stochastic matrix $M = (m_{ij})$, define a BID instance
$D^{p}$ with $n$ blocks $B = \{b_1, \dots, b_n\}$, and $n$ distinct tuples
$T = \{t_1, \dots, t_n\}$, where tuple $t_j$ occurs in block $b_i$ with
probability \ $
  p_{b_i}(t_j) \ = \ m_{ij}$.
  
\h Since $M$ is stochastic, $\sum_{j} p_{b_i}(t_j) = 1$ for every block. So,  each
block contributes exactly with one tuple. This mapping is a bijection between such BIDs
and stochastic matrices; and is clearly computable in time polynomial in $n$ in
both directions.
This establishes the reductions claimed in parts~1 and~2. \ Let
$n_j$ denote the number of blocks that contain $t_j$, i.e. for which $m_{ij}>0$.

\vspace{1mm}
\noindent 
{\bf Worlds and Classes.} \ 
Since a world $W$ selects one tuple per block, it defines an assignment
$B_W\!: [n] \to [n]$, with $B_W(i)$ the index of the tuple drawn from $b_i$. Its
probability is:
\[
  P^{\mathcal{W}}(W) \ = \ \prod_{b_i \in B} p_{b_i}\!\big(t_{B_W(i)}\big)
  \ = \ \prod_{i=1}^n m_{i,\,B_W(i)}.
\]
Each world has a multiplicity vector $\mathbf{k} = \langle k_1, \dots, k_n\rangle$, where
$k_i = \lvert B_W^{-1}(i)\rvert$ counts the blocks that draw $t_i$.  Since the bag has
cardinality $n$, it holds $\sum_{i=1}^n k_i = n$. The class $C_{\mathbf{k}}$ collects the
worlds with a given $\mathbf{k}$, and is admissible only if $k_i \le n_i$, for all
$i$. To $\mathbf{k}$, we attach the $n \times n$ matrix $M_{\mathbf{k}}$ obtained
from $M$, by taking $k_i$ copies of column $i$ (the column of tuple $t_i$). It has
$\sum_i k_i = n$ columns, and then, it is a square matrix.

\vspace{1mm}
\noindent {\bf Class Probability.} \  
Summing world probabilities over the class, we obtain:
\[
  P^{\mathcal{C}}(C_{\mathbf{k}})
  \ := \ \sum_{W \in C_{\mathbf{k}}} P^{\mathcal{W}}(W)
  \ = \ \sum_{W \in C_{\mathbf{k}}} \ \prod_{i=1}^n m_{i,\,B_W(i)}.
\tag{$\ast$}
\]
To relate this to the permanent, let $c(\ell)\in[n]$ be the type of column $\ell$ of $M_{\mathbf{k}}$ (defined as the corresponding original column in $M$). Hence,    $c(\ell)$ is the column of $M$ duplicated as the $l$ column in $M_{\mathbf{k}}$; so that, with $\mathfrak{S}_n$ the set of permutations of $[1,n]$, it holds:
\[
  \mathsf{perm}(M_{\mathbf{k}})
  \ = \ \sum_{\pi \in \mathfrak{S}_n} \prod_{i=1}^n (M_{\mathbf{k}})_{i,\pi(i)}
  \ = \ \sum_{\pi \in \mathfrak{S}_n} \prod_{i=1}^n m_{i,\,c(\pi(i))}.
\]
\h Each permutation $\pi$ determines a world $W$ with $B_W = c \circ \pi$, which
lies in $C_{\mathbf{k}}$: Since $M_{\mathbf{k}}$ has exactly $k_i$ columns of type
$i$, and $\pi$ is a bijection, exactly $k_i$ rows map to type $i$; that is, 
$\lvert B_W^{-1}(i)\rvert = k_i$. \ 

\h Every world of $C_{\mathbf{k}}$ can be obtained in  this way,
and $\pi \mapsto W$ is exactly $\prod_{i=1}^n k_i!$-to-one: Two permutations give
the same world iff they differ by a permutation of the identical columns within
each type; and there are $\prod_{i=1}^n k_i!$ of these. Then, 
\[
  \mathsf{perm}(M_{\mathbf{k}})
  \ = \ \Big(\prod_{i=1}^n k_i!\Big)
        \sum_{W \in C_{\mathbf{k}}} \prod_{i=1}^n m_{i,\,B_W(i)}
  \ = \ \Big(\prod_{i=1}^n k_i!\Big)\, P^{\mathcal{C}}(C_{\mathbf{k}}).
\]
Dividing by $\prod_{i=1}^n k_i!$ gives the identity of part~2:
\[
  P^{\mathcal{C}}(C_{\mathbf{k}})
  \ = \ \frac{\mathsf{perm}(M_{\mathbf{k}})}{\prod_{i=1}^n k_i!}.
\]

\noindent {\bf Specialization (part~1).} \ 
For $\mathbf{k}^\star = (1, \dots, 1)$, every tuple is selected exactly once; and then, 
$M_{\mathbf{k}^\star} = M$, and $\prod_{i=1}^n k_i! = 1$. \ Therefore,
\[
  P^{\mathcal{C}}(C_{\mathbf{k}^\star})
  \ = \ \frac{\mathsf{perm}(M)}{1}
  \ = \ \mathsf{perm}(M).
\]
Both constructions run in polynomial-time in $n$. This  completes the proof.
\boxtheorem

\subsection{Class Probability is \#P-complete (Theorem \ref{th:complresultsmp}(a)).}

 
\textit{Hardness.} Follows from Lemma~\ref{lem:class-permanent}(1), since computing the permanent of an arbitrary stochastic matrix is $\#\mathsf{P}$-hard.

\emph{Membership.} By Lemma~\ref{lem:class-permanent}(2),
$P^{\mathcal C}(C_{\mathbf k})=\mathsf{perm}(M_{\mathbf k})/\prod_i k_i!$. Writing the entries of $M$ over
a common denominator, $\mathsf{perm}(M_{\mathbf k})$ scaled by a polynomial-time-computable integer is a
$\#\mathsf P$ function, and $\prod_i k_i!$ is polynomial-time computable; thus the (scaled) probability is
$\#\mathsf P$-complete under metric reductions.

\vspace{1mm}\noindent{\bf Single-threshold test.}\ 
For a fixed $\mathbf k$, deciding $P^{\mathcal C}(C_{\mathbf k})\ge\theta$ is in $\mathsf{PP}$:
by (a), $P^{\mathcal C}(C_{\mathbf k})$ is $\#\mathsf P$-complete, so (up to a poly-time integer
factor) it is a $\#\mathsf P$ function, and comparing a $\#\mathsf P$ function to a threshold is
exactly $\mathsf{PP}$. We use this in (b) and (c).


\subsection{MPC[D] \ and \ CMPC[D]$^{\{=, \leq\}}$ \ are in \ NP\textsuperscript{PP} (Membership for Theorem \ref{th:complresultsmp}(b) and (c)).}
\label{sec:membership-bc}


$\mathsf{NP^{PP}}$ is the class of decision problems solvable by a nondeterministic
polynomial-time Turing machine with access to a $\mathsf{PP}$ oracle. Recall the
\textbf{MPC[D]} problem: given $\theta \in [0,1]$, decide whether there exists
$\mathbf{k} = (k_1, \dots, k_m) \in \mathbb{N}^m$ such that
\[
  \sum_{i=1}^{m} k_i = n
  \quad \text{and} \quad
  P^{\mathcal{C}}(C_{\mathbf{k}}) \ \geq \ \theta.
\]
\h The search space is exponential: each candidate $\mathbf{k}$ has $\sum_i k_i = n$ with every
$k_i \le n$, so there are up to $(n+1)^m$ of them. A nondeterministic machine sidesteps this by
\emph{guessing} a single $\mathbf{k}$. The guess is a valid certificate: in data complexity $m$ is at most the input size and each
$k_i \le n$, so $\mathbf{k}$ has $O(m\log n)$ bits --- polynomially many --- and can be guessed
in nondeterministic polynomial time.

Verifying a guess takes one oracle call. The feasibility check $\sum_i k_i = n$ is polynomial-time,
and deciding $P^{\mathcal{C}}(C_{\mathbf{k}}) \ge \theta$ is a single $\mathsf{PP}$ query: by
Theorem~\ref{th:complresultsmp}(a), $P^{\mathcal{C}}(C_{\mathbf{k}})$ is (up to a polynomial-time
factor) a $\#\mathsf{P}$ quantity, and comparing a $\#\mathsf{P}$ quantity to a threshold is exactly
what a $\mathsf{PP}$ oracle decides. The machine accepts iff both checks pass, so
$\textbf{MPC[D]} \in \mathsf{NP^{PP}}$.

The same argument covers \textbf{CMPC[D]$^{\{=,\le\}}$}: the machine additionally verifies the
component-wise constraints $k_i \mathrel{op_i} b_i$ for all $i$, which is polynomial-time and leaves
the single $\mathsf{PP}$ query unchanged. Hence \textbf{CMPC[D]$^{\{=,\le\}}$} $\in \mathsf{NP^{PP}}$ as well.

\subsection{CMPC[D]$^{\{=\}}$ is \textbf{PP}-complete (Theorem \ref{th:complresultsmp}(c), first part).}
This is for the first part of Theorem \ref{th:complresultsmp}(c). 
Membership was established in Section~\ref{sec:membership-bc}.

\textit{Hardness.}
 With $\mathcal{O} = \{=\}$, the constraints force $\mathbf{k} = \mathbf{b}$ (and if $\sum_i b_i\neq n$ then
$C_{\mathbf b}=\varnothing$ and $P^{\mathcal C}(C_{\mathbf b})=0$). The problem is therefore the single test
$P^{\mathcal C}(C_{\mathbf b})\ge\theta$, which is in $\mathsf{PP}$ by the single-threshold test and
$\mathsf{PP}$-hard because computing $P^{\mathcal C}$ is $\#\mathsf P$-complete by (a). Hence it is $\mathsf{PP}$-complete.

\subsection{CMPC[D]$^{\{=, \leq\}}$ is NP\textsuperscript{PP}-complete (Theorem \ref{th:complresultsmp}(c), second part).}
This is the proof for the second part of Theorem \ref{th:complresultsmp}(c). Membership was established in Section~\ref{sec:membership-bc}.

\textit{Hardness.} For the $\mathsf{NP^{PP}}$-hardness of
\textbf{CMPC[D]$^{\{=, \leq\}}$}, we give a reduction from \textsc{E-MajSat}, a canonical
$\mathsf{NP^{PP}}$-complete problem. Let $\varphi(\bar{x},\bar{y})$ be a Boolean formula with
disjoint variable sets $\bar{x}=(x_1,\dots,x_r)$ and $\bar{y}=(y_1,\dots,y_l)$, called
respectively the \emph{existential} and the \emph{counting} variables. \textsc{E-MajSat} asks
whether there exists $\alpha\in\{0,1\}^r$ such that
\[
    \bigl|\{\beta\in\{0,1\}^l : \varphi(\alpha,\beta)=1\}\bigr| \;\ge\; 2^{l-1},
\]
that is, whether some assignment to $\bar{x}$ makes a majority of the assignments to $\bar{y}$
satisfy $\varphi$. In our reduction the nondeterministic choice of $\alpha$ is realized by the
choice of the duplication applied to the existential gadgets.

\vspace{1mm}
\h For a directed graph, $\mc{G}$, possibly with self-loops, with adjacency matrix, $A(\mc{G})$, a \emph{cycle
cover} is a set of vertex-disjoint directed cycles covering all vertices. $\perm(A(\mc{G}))$ coincides with
the number of cycle covers of $\mc{G}$. \ Valiant's parsimonious reduction from $\#\textsc{Sat}$ to
\textsc{Permanent} \cite{Valiant1979} rests on this identity: From a CNF formula,  a graph
$\mc{G}_\varphi$ is built out of variable gadgets, clause gadgets, and identification (XOR) gadgets, whose cycle
covers correspond to the satisfying assignments of $\varphi$. \ 
We keep the clause and XOR gadgets
of \cite{Valiant1979} unchanged and modify only the gadgets of the existential variables.

\vspace{1mm} \noindent {\bf The Augmented Variable Gadget.} \ 
Valiant's variable gadget for variable $x$ has a positive and a negative path, through the literal nodes
$x_P$ and $x_N$, each with a self-loop, and a node $x'$ closing the gadget cycle via the edge $x'\!\to\!x$.
For every existential variable $x$, we insert two \emph{gate nodes}, $G_P,G_N$, one per path, so
that each path reaches $x'$ through its gate:
\[
    x_P\to G_P\to x', \ \mbox{ and } \  x_N\to G_N\to x',
\]
the gates carrying no self-loop. This yields the core matrix $C_x$ on the ordered node set
$(x,x_P,x_N,x',G_P,G_N)$ shown on the top-left in Figure~\ref{fig:proof-gadget}; we call
this six-node subgraph the \emph{(existential) core} of $x$, i.e.\ the augmented gadget before the XOR and
clause connectors are attached. The edges used as
\emph{connectors} by the XOR gadgets are the original path edges leaving $x_P$ and $x_N$. The
gate edges are never connectors; so, the XOR and clause gadgets are attached exactly as in
\cite{Valiant1979}. 
Since the existential cores are pairwise vertex-disjoint, each
gate-column overwrite  is confined to its own core, so the per-core assignment
choices are independent.

\begin{figure}[h]
    \centering
    \begin{tikzpicture}[->, >=stealth, node distance=0.6cm and 1cm]
        \node (x) {$x$};
        \node (xp) [above left=0.8cm and 0.25cm of x] {$x_P$};
        \node (xn) [above right=0.8cm and 0.25cm of x] {$x_N$};
        \node (gp) [above=1.2cm of xp] {$G_P$};
        \node (gn) [above=1.2cm of xn] {$G_N$};
        \node (xpr) [above=3.0cm of x] {$x'$};
        \draw[->] (x) -- (xp);
        \draw[->] (x) -- (xn);
        \draw[->] (xp) -- (gp);
        \draw[->] (xn) -- (gn);
        \draw[->] (gp) -- (xpr);
        \draw[->] (gn) -- (xpr);
        \draw[->] (xpr) -- (x);
        \draw[->] (xp) edge[loop left] ();
        \draw[->] (xn) edge[loop right] ();

        \node (x2) [right=5.5cm of x] {$x$};
        \node (xp2) [above left=0.8cm and 0.25cm of x2] {$x_P$};
        \node (xn2) [above right=0.8cm and 0.25cm of x2] {$x_N$};
        \node (gp2) [above=1.2cm of xp2] {$G_P$};
        \node (gn2) [above=1.2cm of xn2] {$G_N$};
        \node (xpr2) [above=3.0cm of x2] {$x'$};
        \draw[->] (x2) -- (xp2);
        \draw[->] (x2) -- (xn2);
        \draw[->] (xp2) -- (gp2);
        \draw[->] (gp2) -- (xpr2);
        \draw[->] (gn2) -- (xpr2);
        \draw[->] (xpr2) -- (x2);
        \draw[->] (xp2) edge[loop left] ();
        \draw[->] (xn2) edge[loop right] ();
        \draw[->] (gn2) edge[loop right] ();
        \draw[->, densely dashed] (gp2) to[bend left=20] (gn2);

        \node (Mx) [below=1.0cm of x] {\textbf{Core Matrix $C_x$}};
        \node (M1) [below=0.15cm of Mx] {$
            \begin{array}{c|cccccc}
                    & x & x_P & x_N & x' & G_P & G_N \\ \hline
                x   & 0 & 1   & 1   & 0  & 0   & 0   \\
                x_P & 0 & 1   & 0   & 0  & 1   & 0   \\
                x_N & 0 & 0   & 1   & 0  & 0   & 1   \\
                x'  & 1 & 0   & 0   & 0  & 0   & 0   \\
                G_P & 0 & 0   & 0   & 1  & 0   & 0   \\
                G_N & 0 & 0   & 0   & 1  & 0   & 0   \\
            \end{array}$};

        \node (Mx2) [below=1.0cm of x2] {\textbf{Matrix $C_x^{\mathsf{T}}$, $x=\mathsf{true}$ \; ($G_N \!\leftarrow\! x'$)}};
        \node (M2) [below=0.15cm of Mx2] {$
            \begin{array}{c|cccccc}
                    & x & x_P & x_N & x' & G_P & G_N \\ \hline
                x   & 0 & 1   & 1   & 0  & 0   & 0   \\
                x_P & 0 & 1   & 0   & 0  & 1   & 0   \\
                x_N & 0 & 0   & 1   & 0  & 0   & 0   \\
                x'  & 1 & 0   & 0   & 0  & 0   & 0   \\
                G_P & 0 & 0   & 0   & 1  & 0   & 1   \\
                G_N & 0 & 0   & 0   & 1  & 0   & 1   \\
            \end{array}$};
    \end{tikzpicture}
    \caption{\ {\bf The Augmented Variable Gadget.} Left: The core gadget $C_x$, obtained from
    Valiant's variable gadget by inserting the gate nodes $G_P,G_N$ on the positive and
    negative paths, so that they reach $x'$ through their gates
    ($x_P\!\to\!G_P\!\to\!x'$, $x_N\!\to\!G_N\!\to\!x'$). The connector edges used by the XOR
    gadgets are the original path edges out of $x_P$ and $x_N$. The gate edges are never used
    as connectors. The bare core has $\mathsf{perm}(C_x)=0$: With no self-loops on the gates,
    an idle gate cannot be covered; so, no cover exists until a value is selected. \  Right: The
    instance selecting $x=\mathsf{true}$, obtained by duplicating column $x'$ over column
    $G_N$. Since column $x'$ has its ones in rows $G_P$ and $G_N$, the copy gives column $G_N$
    the entries $(G_P,G_N)$ and $(G_N,G_N)$: It severs $x_N\!\to\!G_N$, and creates the gate
    self-loop $G_N\!\to\!G_N$ (drawn) together with the edge $G_P\!\to\!G_N$ (dashed). Thus,
    $x_N$ can only be covered by its own self-loop, and the negative path is idle. The instance
    for $x=\mathsf{false}$ is symmetric ($G_P\!\leftarrow\!x'$). Each selected instance has
    permanent $2$.}
    \label{fig:proof-gadget}
\end{figure}
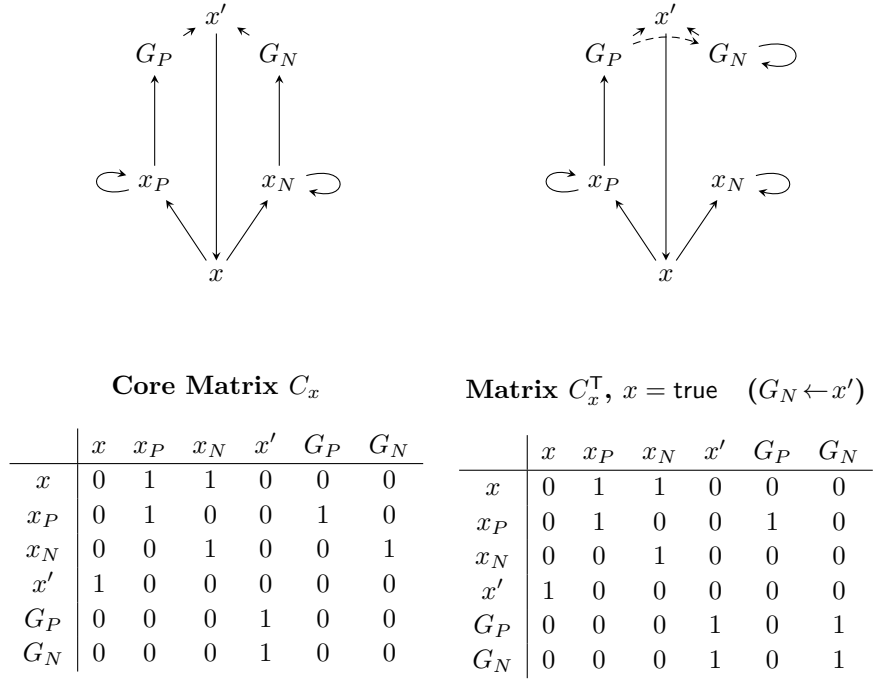


\vspace{1mm} {\bf Existential Variable Assignment.} \ This is done through a controlled column duplication mechanism. 
\ The matrix size is fixed, and on each existential gadget, a single operation is allowed: Column
$x'$ is duplicated over a target column, overwriting it. \ Overwriting column $G_N$ (resp.\ $G_P$)
by a copy of column $x'$ yields the matrix $C_x^{\mathsf{T}}$ (resp.\ $C_x^{\mathsf{F}}$). The
instance $C_x^{\mathsf{T}}$, encoding $x=1$, is shown on the right of
Figure~\ref{fig:proof-gadget}.

\vspace{2mm}
\begin{lemma}\label{lem:gadget} \em
For every existential variable $x$, $\perm(C_x)=0$, while
$\perm(C_x^{\mathsf{T}})=\perm(C_x^{\mathsf{F}})=2$. \  Moreover, every cycle cover of
$C_x^{\mathsf{T}}$ traverses the positive connector edge, and leaves the negative connector edge
unused, with $x_N$ covered by its self-loop; symmetrically for $C_x^{\mathsf{F}}$.
\end{lemma}

\noindent {\bf Proof:} \ 
In $C_x$, the gates have no self-loop; so,  the gate of the inactive path cannot be covered. Then, 
$C_x$ has no cycle cover, and $\perm(C_x)=0$. Column $x'$ has its two ones in rows $G_P$ and $G_N$.
Overwriting column $G_N$ by column $x'$ deletes the entry $(x_N,G_N)$, and creates $(G_P,G_N)$ and
$(G_N,G_N)$. It, thus, severs $x_N\!\to\!G_N$, and gives $G_N$ a self-loop. Row $x_N$ then, has its
self-loop as its only out-edge; so, in every cover $x_N$ is covered by its self-loop, the negative path is idle,
$x$ maps to $x_P$, and the positive path $x\!\to\!x_P\!\to\!G_P\!\to\!x'\!\to\!x$ is used. \ The two
covers of $C_x^{\mathsf{T}}$ correspond to the two ways of covering $\{G_P,G_N\}$ that are compatible with
the duplicated column, each with weight $1$. Hence, $\perm(C_x^{\mathsf{T}})=2$. The case
$C_x^{\mathsf{F}}$ is symmetric. \boxtheorem

\h When $x$ occurs in several clauses, the single literal node $x_P$ is replaced by a chain of
 nodes $p_1,\dots,p_k$, one per occurrence, each with a self-loop; so that the positive
path becomes $x\to p_1\to\cdots\to p_k\to G_P\to x'$. The negative path is treated symmetrically, with 
the gate lying below the whole chain. Overwriting column $G_N$ by column $x'$ severs the last
negative edge and, by downward induction along the chain, forces every negative occurrence onto
its self-loop. The $k$ positive occurrences are then, simultaneously active. This yields the
consistency across occurrences that the identification gadgets require, and
$\perm(C_x^{\mathsf{T}})=\perm(C_x^{\mathsf{F}})=2$, independently of $k$.

\vspace{1mm}
\noindent {\bf The Reduction.}\ 
Apply Valiant's construction to $\varphi(\bar{x},\bar{y})$, replacing each existential variable
gadget by its augmented core $C_x$; let $A:=A(\mc{G}_\varphi)$ be the adjacency matrix of the
resulting graph. Since a BID requires a row-stochastic matrix, we normalize: let
$s_i=\sum_j A_{ij}$ be the out-degree of node $i$ and $M:=D^{-1}A$, with
$D=\mathrm{diag}(s_1,\dots,s_n)$, the row-stochastic normalization of $A$ (every node lies on a
cycle, so $s_i\ge 1$ and $M$ is a valid BID probability matrix). Writing $Z:=\prod_i s_i$, the
permanent satisfies $\perm(M_{\mathbf k})=\perm(A_{\mathbf k})/Z$. The BID has one tuple per column
of $A$, and we impose, per existential core,
\[
  k_{x'}=2,\qquad k_{G_P}\le 1,\qquad k_{G_N}\le 1,
\]
with $k_j=1$ for every other column. All operators lie in $\{=,\le\}$, so this is a legal
\textbf{CMPC[D]$^{\{=,\le\}}$} instance, and $P^{\mathcal C}(C_{\mathbf k})=\perm(M_{\mathbf k})/\prod_i k_i!$.

\vspace{1mm}
\noindent {\bf Feasible vectors encode assignments.} \ 
Fix a feasible $\mathbf k$. The constraints leave only the gate multiplicities free, and preserving
the matrix size ($\sum_i k_i=n$) forces exactly $r$ gate-columns deleted in total. We claim
$\perm(A_{\mathbf k})>0$ requires \emph{exactly one} deletion per core. Indeed, in a core with
\emph{no} gate deleted the $k_{x'}=2$ copies leave more core-columns than core-rows, so no cycle
cover matches them; in a core with \emph{both} gates deleted, rows $x_P,x_N$ keep only their
self-loops and $x$ has no free out-column, so again no cover exists. Hence each core deletes exactly
one gate, and by Lemma~\ref{lem:gadget} the two choices $(k_{G_P},k_{G_N})=(1,0)$ and $(0,1)$ are
precisely the selecting duplications encoding $x=\mathsf{true}$ and $x=\mathsf{false}$. Feasible
vectors with positive probability are thus in bijection with assignments
$\alpha(\mathbf k)\in\{0,1\}^r$.

\vspace{1mm}
\noindent {\bf Probability counts satisfying $\bar y$.} \ 
For such a $\mathbf k$, Lemma~\ref{lem:gadget} puts every connector edge in one of the two states of
Valiant's variable gadget --- the active connector traversed, the idle one unused with its literal
node on a self-loop, exactly as an idle track is covered in \cite{Valiant1979}. As the connectors
and XOR gadgets are unchanged and no gate meets an XOR, the clause/XOR correctness is inherited, and
\[
  \perm\bigl(A_{\mathbf k}\bigr)
  \;=\; 2^{r}\,K_0 \cdot
        \bigl|\{\beta\in\{0,1\}^l : \varphi(\alpha(\mathbf k),\beta)=1\}\bigr|,
\]
where $K_0>0$ is Valiant's global constant and $2^{r}$ collects the two internal gate-coverings per
core. Writing $N(\alpha)=|\{\beta:\varphi(\alpha,\beta)=1\}|$ and using
$\prod_i k_i! = \prod_x 2! = 2^{r}$, the two powers of $2^{r}$ cancel:
\[
  P^{\mathcal C}(C_{\mathbf k})
  \;=\;\frac{\perm(M_{\mathbf k})}{\prod_i k_i!}
  \;=\;\frac{\perm(A_{\mathbf k})}{Z\,\prod_i k_i!}
  \;=\;\gamma\cdot N(\alpha(\mathbf k)),
  \qquad \gamma:=\frac{K_0}{Z}>0,
\]
a fixed constant, independent of $\mathbf k$.

\vspace{1mm}
\noindent {\bf Threshold and conclusion.} \ 
Set $\theta:=\gamma\,2^{l-1}$. A feasible $\mathbf k$ satisfies $P^{\mathcal C}(C_{\mathbf k})\ge\theta$
iff it has positive probability and $\gamma\,N(\alpha(\mathbf k))\ge\gamma\,2^{l-1}$, i.e.\ iff
$\alpha(\mathbf k)$ makes a majority of the $\beta$ satisfy $\varphi$. Hence some feasible $\mathbf k$
meets the threshold iff $\varphi$ is a positive instance of \textsc{E-MajSat}. The construction is
carried out in polynomial time, so \textbf{CMPC[D]$^{\{=,\le\}}$} is $\mathsf{NP^{PP}}$-hard, and with
membership (established earlier), $\mathsf{NP^{PP}}$-complete. \qed

\subsection{MPC[O] is in FP\textsuperscript{NP\textsuperscript{PP}} (Theorem \ref{th:complresultsmp}(d)).}
Since \textbf{MPC[D]} $\in \mathsf{NP}^{\mathsf{PP}}$, we solve \textbf{MPC[O]} in
$\mathsf{FP}^{\mathsf{NP}^{\mathsf{PP}}}$.
Recall MPC[O] asks for a maximizer
\[
  \mathbf k^\star \ \in\ \argmax_{\mathbf k\in\mathbb N^m,\ \sum_i k_i=n} P^{\mathcal C}(C_{\mathbf k}).
\]
We compute one in two phases, each using polynomially many calls to the $\mathsf{NP^{PP}}$ oracle
\textbf{MPC[D]}, which decides ``$\exists\,\mathbf k:\ P^{\mathcal C}(C_{\mathbf k})\ge\theta$''.

\emph{Phase 1: the optimal value.}\ 
By (a), every achievable probability is a rational $\perm(M_{\mathbf k})/Q$ with a common
denominator $Q$ of polynomially many bits; hence there are finitely many achievable values, any two
differing by at least $1/Q^2$. Binary search on $\theta\in[0,1]$, querying MPC[D] at each step,
therefore isolates the optimum
$P^\star=\max_{\mathbf k}P^{\mathcal C}(C_{\mathbf k})$ after polynomially many queries.

\emph{Phase 2: a maximizer.}\ 
We reconstruct $\mathbf k^\star$ component by component by self-reduction. Having fixed
$k_1^\star,\dots,k_{j-1}^\star$, try each candidate value $v\in\{0,\dots,n\}$ for $k_j$ and ask the
oracle whether some completion with $k_j=v$ (and the earlier components fixed) attains
$P^{\mathcal C}\ge P^\star$; set $k_j^\star$ to the first $v$ that succeeds. After $m$ such rounds,
$\mathbf k^\star$ is fully determined and satisfies $P^{\mathcal C}(C_{\mathbf k^\star})=P^\star$.

Both phases make polynomially many $\mathsf{NP^{PP}}$ queries with polynomial-time bookkeeping, so
$\textbf{MPC[O]}\in\mathsf{FP^{NP^{PP}}}$.
\boxtheorem

\section{Appendix: Proof of Proposition \ref{prop:recov}.}\label{app:recov}

\noindent {\bf Proof}: \ We consider the MAR case. The MCAR case is simpler. We recall that we assumed that MGs are $\mbb{I}$-deterministic. (The extension to MGs satisfying Condition {\bf {C}} only should be easy.)
\ It is good enough to prove that the joint distributions $P(\bar{A}^o,\bar{C}^m,\mbb{I}^{\bar{C}},\bar{C}^\star)$ are recoverable.
\ Due to the $\mbb{I}$-determinacy, it suffices to estimate  $P(\bar{A}^o,\bar{C}^m,\mbb{I}^{\bar{C}})$, and express it in terms of $P(\bar{A}^o,\bar{C}^\star,\mbb{I}^{\bar{C}})$.

\h Now, $P(\bar{A}^o,\bar{C}^m,\mbb{I}^{\bar{C}}) = P(\bar{A}^o_1,\bar{C}^m | \bar{A}^o_2) \times P(\bar{A}^o_2)$, where $\bar{A}^o_1$ are parents of variables in $\mbb{I}^{\bar{C}}$. Then, $(\bar{A}^o_1\bar{C}^m \! \independent  \mbb{I}^{\bar{C}})|\bar{A}^o_2$. It follows that $P(\bar{A}^o_1,\bar{C}^m | \bar{A}^o_2) \times P(\bar{A}^o_2) = P(\bar{A}^o_1,\bar{C}^m | \mbb{I}^{\bar{C}} = \bar{0}, \bar{A}^o_2) \times P(\bar{A}^o_2) = P(\bar{A}^o_1,\bar{C}^\star | \mbb{I}^{\bar{C}} = \bar{0}, \bar{A}^o_2) \times P(\bar{A}^o_2)$. These probabilities are expressed in terms of the observed variables in $D^\star$. The algorithm specified in the definition of recoverability amounts, in this case, to computing relative frequencies from the observed data. \boxtheorem

\ignore{
\section{Styles of lists, enumerations, and descriptions}\label{sec:itemStyles}  }











\begin{thebibliography}{2}

\bibitem{we} Anonymous. {\em Title}. Accepted at {\em Hidden} Conference, 2026.

\bibitem{ahuja1993}
Ahuja,~R.~K., Magnanti,~T.~L. and  Orlin,~J.B. \
\newblock {\em Network Flows: Theory, Algorithms, and Applications}.
\newblock Prentice Hall,  1993.

\bibitem{ali1966}
Ali,~S.~M. and Silvey,~S.~D.
\newblock A General Class of Coefficients of Divergence of one Distribution
from Another.
\newblock {\em Journal of the Royal Statistical Society, Series B},
1996, 28(1):131--142.

\ignore{
\bibitem{allisonML}
Allison,~P.~D. \ {\em Maximum Likelihood Estimation}. \ \ Quantitative Applications in the Social Sciences, 96. Sage Publications, 1993.}

\ignore{
\bibitem{allisonMI}
Allison,~P.~D. \ {\em Missing Data}. \ Quantitative Applications in the Social Sciences, 136. Sage Publications, 2002.
}

\ignore{
\bibitem{uai26corr}
Bertossi,~L., Toumani,~F. and Buron,~M. \ Database Querying under Missing Values Governed by Missingness Mechanisms. arXiv paper 2604.06520.
}


\ignore{
\bibitem[Bertossi et al.,2013]{li}
Bertossi,~L. and Li,~L. \ Achieving Data Privacy through Secrecy Views and Null-Based Virtual Updates. \ {\em IEEE Transactions on Knowledge and Data Engineering}, 2013, 25(5):987-1000.}

\ignore{
\bibitem{p2p}
Bertossi,~L. and Bravo,~L. \ Consistency and Trust in Peer Data Exchange Systems. {\em Theory and Practice of Logic Programming}, 2017, 17(2):148-204.}

\ignore{
\bibitem{Bishop2006} Bishop,~C.~M. \ {\em Pattern Recognition and Machine Learning}. \ Springer, 2006.
}

\ignore{
\bibitem[Takeaki,2001]{uno2001}
Takeaki,~U. \ \textit{A Fast Algorithm for Enumerating Bipartite Perfect Matchings}.,  ISAAC 2001 (LNCS vol 2223), 2001.

\bibitem[Chen et al.,2022]{chen2022minimum}
Chen, L., Kyng, R., Liu, Y. P., Peng, R., Gutenberg, M. P., \& Sachdeva, S. (2022).
\newblock Minimum cost flows, MDPs, and \(\ell_1\)-regression in nearly linear time for separable convex instances.
\newblock In \emph{Proceedings of the 54th Annual ACM SIGACT Symposium on Theory of Computing} (STOC 2022).
}

\bibitem{console}
Console,~M., Libkin,~L. and Peterfreund,~L. \
Querying Incomplete Numerical Data: Between Certain and Possible Answers. \ Proc. PODS 2023, pp. 349-358.

\bibitem{dalvi2004}
Dalvi,~N. and Suciu,~D. \
Efficient Query Evaluation on Probabilistic Databases. In Proc VLDB 2004.

\bibitem{sarma2006working}
Das Sarma,~A., Benjelloun,~O., Halevy,~A. and Widom,~J.
\newblock Working Models for Uncertain Data.
\newblock Proc ICDE 2006.





\ignore{
\bibitem{devore}
Devore,~J.~L., Berk,~K.~N. and Carlton,~M.~A. \ {\em Modern Mathematical Statistics with Applications}. 3rd Ed., Springer, 2021.
}




\bibitem{darwicheCommACM} Darwiche,~A. \
Bayesian Networks. {\em Communications of the ACM}, 2010, 53(12):80-90.

\ignore{
\bibitem[Dechter,2019]{dechter}
Dechter,~R. \
{\em Reasoning with Probabilistic and Deterministic Graphical Models}. \ 2nd Ed., Synthesis Lectures on Artificial Intelligence and Machine Learning, Morgan \& Claypool Pubs., 2019.}

\bibitem{benny}
De Sa,~C.,  Ilyas,~I.,  Kimelfeld,~B., Re,~C. and  Rekatsinas,~T. \
A Formal Framework for Probabilistic Unclean Databases. Proc. ICDT 2019, pp. 6:1-6:18.





\ignore{
\bibitem{ducamp2020agrum}
Ducamp,~G., Gonzales,~C., Wuillemin,~P.-H.
aGrUM/pyAgrum: a toolbox to build models and algorithms for Probabilistic Graphical Models in Python.
Proc. International Conference on Probabilistic Graphical Models, 2020, PMLR.



\bibitem{Fink12}
Fink ,~R.,  Han,~L., Olteanu,~D. \
Aggregation in Probabilistic Databases via Knowledge Compilation. Proc. VLDB Endow. 5(5): 490-501 (2012).
}

\bibitem{lastLibkin}
Gheerbrant,~A., Libkin,~L., Rogova,~A. and Sirangelo,~C. \
Querying Incomplete Data: Complexity and Tractability via Datalog and First-Order Rewritings. \ {\em Theory and Practice of Logic Programming}, 2024,
24(2):279-309.

\bibitem{gelman}
Gelman,~A. and Hill,~J. \ {\em Data Analysis Using Regression and Multilevel/Hierarchical Models}. Cambridge
Univ. Press, 2007.



\ignore{
\bibitem{bennyICDT23}
Gilad,~A., Imber,~A. and Kimelfeld,~B. \ The Consistency of Probabilistic Databases with
Independent Cells. \ Proc. ICDT, 2023.
}

\ignore{
\bibitem{greco}
Greco,~S., Molinaro,~C. and Spezzano,~F. \ Incomplete Data and Data Dependencies in Relational Databases. Synthesis Lectures in Data Management, Morgan \& Claypool Pubs., 2012.
}

\bibitem{marko} Kumar Jha,~A. and Suciu,~D. Probabilistic Databases with MarkoViews.
Proc. VLDB, 2012, 5(11):1160-1171

\ignore{
\bibitem[Gribkoff et al.,2014]{mostProbDB}
Gribkoff,~E., Van den Broeck,~G. and Suciu,~D. \
The Most Probable Database Problem. \ Proc. BUDA 2014, collocated with SIGMOD, 2014.}

\ignore{
\bibitem{Fukuda1994}
Fukuda, K. , Matsui, T. \
Finding all the perfect matchings in bipartite graphs.
{\em Applied Mathematics Letters}. 1994.
}

\ignore{
\bibitem{hastie}
Hastie,~T., Tibshirani,~R.
and Friedman,~J. \
{\em The Elements of Statistical Learning}. 2nd Ed., Springer, 2017.
}

\bibitem{heckerman95}
Heckerman,~D., Geiger,~D. and Chickering,~D.~M. \ Learning Bayesian Networks: The Combination of
Knowledge and Statistical Data. \ {\em Machine Learning}, 1995, 20:197-243.

\bibitem{heckerman}
Heckerman,~D. \ A Tutorial on Learning with Bayesian Networks. In  {\em Innovations in Bayesian Networks}, Springer 2008, pp. 33-82.   Updated as arXiv paper 2002.00269, 2022.


\bibitem{imielinski}
Imielinski,~T. and Lipski Jr.,~W. \
Incomplete Information in Relational Databases. \ {\em Journal of the ACM}, 1984, 31(4):761-791.



\ignore{
 \bibitem{guy2}
Khosravi,~P.,  Liang,~Y., Choi,~Y. and Van den Broeck,~G. \ What to Expect of Classifiers? Reasoning about Logistic Regression with Missing Features. \ Proc. IJCAI, 2019, pp. 2716-2724.
}

\ignore{
\bibitem{hungarian}
Kuhn,~H.~W. \ The Hungarian Method for the Assignment Problem. \ {\em Naval Research Logistics Quarterly}, 1955, 2(1-2):83-97.
}

\ignore{
\bibitem{markoviews}
Kumar Jha,~A. and  Dan Suciu, D. \
Probabilistic Databases with MarkoViews. Proc. VLDB 5(11):1160-1171,  2012.
}

\bibitem{rubinBook}
Little,~R.~J. and Rubin,~D.~B. \ {\em Statistical Analysis with Missing Data}. \ John Wiley \& Sons, 3rd Ed., 2019.

\ignore{
\bibitem{littman1998}
Littman,~M.~L. and Goldsmith,~J. and Mundhenk,~M. \ {\em The Computational Complexity of Probabilistic Planning}, JAIR, 1998.
}

\ignore{
\bibitem{MacKay2003} MacKay,~D.~J.  \ {\em Information Theory, Inference, and Learning Algorithms}. Cambridge University Press, 2003.
}




\ignore{
\bibitem[Mohan et al.,2013]{mohan2013}
Mohan,~K., Pearl,~J. and Tian, ~J. \
Graphical Models for Inference with Missing Data. \
{\em Proc NIPS}, 2013, pp. 1277-1285.
}

\ignore{
\bibitem{mohan2014}
Mohan,~K and Pearl,~J. \textit{On the Testability of Models with Missing Data}. \ Proc. AISTAT, 2014.
}

\bibitem{mohan13}
Mohan,~M., Pearl,~J. and Tian,~J. \
Graphical Models for Inference with Missing Data. \
Proc. NIPS 2013. Vol.  26.


\bibitem{mohanthesis2017}
Mohan,~K. \ \textit{Graphical Models for Inference with Missing Data}. \ PhD Thesis, UCLA, 2017. 

\bibitem{mohan21}
Mohan,~K. and Pearl,~J. \ Graphical Models for Processing Missing Data. \
{\em Journal of the American Statistical Association},
2021, 116(534):1023–1037.

\ignore{
\bibitem[Olteanu et al.,2015]{Olteanu2015}
Olteanu,~D. and Z\'{a}vodn\'{y},~J. \ Size Bounds for Factorised Representations of Query Results. \
{\em ACM TODS}, 2015, 40(1).
}

\bibitem{Papadimitriou1994}
Papadimitriou,~C. \ {\em Computational Complexity}, Addison-Wesley, Reading, MA, 1994.


\bibitem{pearl}
Pearl,~J. \ {\em Causality: Models, Reasoning and Inference}. \ Cambridge Univ. Press, 2nd edition, 2009.

\ignore{
\bibitem[Pearl,2010]{pearl2010}
Pearl,~J. \ The Foundations of Causal Inference. \ {\em Sociological Methodology}, 2010, 40:75–149.  \url{http://www.jstor.org/stable/41336883}. 
}



\bibitem{re07}
R\'{e},~C., Dalvi,~N. and Suciu,~D. \
Efficient Top-k Query Evaluation on Probabilistic Data. \ Proc. ICDE,  2007, pp.  886-895.



\bibitem{re2005trio}
R\'e,~C., Dalvi,~N. and Naughton,~J.~F.
\ Trio: A System for Integrated Management of Data, Accuracy, and Lineage.
\newblock Proc. CIDR, 2005, pp. 262–276.

\ignore{
\bibitem{reiter}
Reiter,~R.\
 A Sound and Sometimes Complete Query Evaluation Algorithm for Relational Databases with Null Values. \ {\em J. of the ACM}, 1986, 33(2):349-370.
}


\ignore{
\bibitem{Olteanu2012}
Olteanu, Dan and Z\'{a}vodn\'{y}, Jakub. \
Factorised representations of query results: size bounds and readability.\ Proc ICDT 2012.}

\ignore{
\bibitem{re09}
R\'{e},~C., Suciu,~D. \
The Trichotomy of HAVING Queries on a Probabilistic Database. \ {\em VLDB J.}, 2009, 18(5):1091-1116.}


\bibitem{royden1988}
Royden,~H.~L.
\newblock {\em Real Analysis}, 3rd ed.
\newblock Macmillan, New York, 1988.

\ignore{
 \bibitem{salimi19}
Salimi,~B., Rodriguez,~L., Howe,~B. and Suciu,~D. \ Interventional
Fairness: Causal Database Repair for Algorithmic Fairness. \ Proc. SIGMOD, 2019, pp. 793–810.
}

\ignore{
 \bibitem{salimi24}
Salimi,~B., Milani,~M., Pirhadi,~A., Cloninger,~A. and Moslemi,~M.~H. \ OTClean: Data Cleaning for Conditional Independence Violations using Optimal Transport. \ Proc. SIGMOD, 2024.
}

\bibitem{rubin76}Rubin,~D.~B. \
 Inference and Missing Data.  {\em Biometrika}, 1976, 63(3):581–592.

\ignore{\bibitem{Senellart2018}
Senellart,~P., Jachiet,~L., Maniu,~S. and Ramusat,~Y. \
ProvSQL: Provenance and Probability Management in PostgreSQL. \ Proc VLDB, 2018.}

\bibitem{soliman2007topk}
Soliman,~M.,  Ilyas,~I and Chang,~K. \ Top-k Query Processing in Uncertain Databases. In Proc. ICDE, 2007.

\ignore{
\bibitem[Soliman et al.,2008]{soliman2008}
Soliman,~M.,  Ilyas,~I. and Chang,~K. \ Probabilistic Top-k and Ranking-Aggregate Queries. {\em ACM TODS}, 2008, 33(3).
}

\ignore{
\bibitem{strozecki2023}
Strozecki,~Y. \ Enumeration Complexity: Incremental Time, Delay and Space.  arXiv, 2023., \url{https://arxiv.org/abs/2309.17042}.
}

\bibitem{suciu}
Suciu,~D., Olteanu,~D., R\'e,~C. and Koch,~C. \ {\em Probabilistic Databases}. \
Synthesis Lectures on Data Management, Morgan \& Claypool Pubs., 2011.

\bibitem{forall} Suciu,~D. \ Probabilistic Databases for All. \ Proc. PODS 2020, pp. 19-31.

\ignore{
\bibitem{toran91}
Tor\'{a}n,~J. Complexity Classes Defined by Counting Quantifiers. {\em J. ACM.}, 1991.}

\bibitem{libkin}
Toussaint,~E., Guagliardo,~P., Libkin,~L. and Sequeda,~J. \
Troubles with Nulls, Views from the Users. Proc. VLDB, 2022, 15(11):2613-2625.


\bibitem{broeck2015}
Van den Broeck,~G., Mohan,~K., Choi,~A., Darwiche,~A. and Pearl,~J.  Efficient Algorithms for Bayesian Network
Parameter Learning from Incomplete Data. \  Proc. UAI, 2013.



\bibitem{NOW}
Van den Broeck,~G. and Suciu,~D. \
Query Processing on Probabilistic Data:
A Survey. \ {\em Foundations and Trends
in Databases}, 2015,
7(3-4):197–341. \ NOW Publishers.

\bibitem{villani}
Villani,~C. \ {\em Optimal Transport. Old and New.}, Springer, 2009.




\bibitem{larry}
Wasserman,~L. \ {\em All of Statistics}. \ Springer, 2010.

\ignore{
\bibitem{vegh2016strongly}
Végh, L. A. (2016).
\newblock A strongly polynomial algorithm for a class of minimum-cost flow problems with separable convex objectives.
\newblock \emph{SIAM Journal on Computing}, 45(5), 1729--1761.
}

\ignore{
\bibitem{wong1}
Wong,~S.~K.~M. and Wang,~Z.~W. \  On Axiomatization of Probabilistic
Conditional Independence. \ Proc. Conf. Uncertainty in Artificial Intelligence, 1994, pp. 591–597.
}

\ignore{\bibitem[Wong et al.,2000]{wong2}
Wong,~S.~K.~M., Butz,~C.~J.  and Wu,~D. \
On the Implication Problem for Probabilistic
Conditional Independency. \ {\em IEEE Transactions on Systems, Man, and Cybernetics, Part A}, 2000, 30(6):785-805.}

\ignore{
\bibitem[Stats]{websiteMI}
\url{https://stats.oarc.ucla.edu/stata/seminars/mi_in_stata_pt1_new}
}

\end{thebibliography}

\begin{thebibliography}{2}
\setcounter{enumiv}{35}

\bibitem{colbourn95}
Colbourn,~C.~J., Provan,~J.~S. and Vertigan,~D. \ The Complexity of Computing the Tutte Polynomial on Transversal Matroids. {\em Combinatorica}, 1995, 15:1–10.

\bibitem{dagum92}
Dagum,~P.  and Luby,~M. \ Approximating the Permanent of Graphs with Large Factors. \  {\em Theoretical Computer Science}, 1992, 102:283–305.

\bibitem{dinic1970}
Dinic,~E.~A. Algorithm for Solution of a Problem of Maximum Flow in a Network with Power Estimation. \emph{Soviet Mathematics Doklady}, 1970, 11:1277--1280.

\bibitem{orlin1988}
Orlin, James. A faster strongly polynomial minimum cost flow algorithm. STOC, 1988.

\bibitem{Minoux86}
Minoux, M.. Solving Integer Minimum Cost Flows with Separable Convex Cost Objective Polynomially. In: Gallo, G., Sandi, C. (eds) `Netflow at Pisa'. Mathematical Programming Studies, Springer 1986, vol 26.

\bibitem{Valiant1979}
Valiant,~L.~G. \
\newblock The Complexity of Computing the Permanent.
\newblock {\em Theoretical Computer Science}, 8(2):189--201, 1979.


\end{thebibliography}
\end{document}